\documentclass[aip,pre]{revtex4-2}
\usepackage[final]{changes}
\definechangesauthor[color=red]{1}
\definechangesauthor[color=red]{2} 
\definechangesauthor[color=blue]{3} 
\usepackage{stackengine}
\usepackage{siunitx}
\usepackage{amsmath,amssymb}
\usepackage{eucal}  
\usepackage{mathrsfs}
\usepackage{placeins} 
\usepackage{czt}  
\usepackage{tabularx}
\usepackage{array} 
\usepackage{booktabs} 
\usepackage{tikz}  %
\usetikzlibrary{tikzmark}
\usepackage{tzplot}
\usepackage{hyperref}
\usepackage{cancel}  
\usepackage[normalem]{ulem}
\usepackage{graphicx}
\usepackage{dcolumn}
\usepackage{bm}
\usepackage{color}
\usepackage{enumitem}
\usepackage{empheq} 
\usepackage{url}
\usepackage{gensymb}  
\usepackage{stackengine}
\DeclareFontFamily{U}{mathc}{}
\DeclareFontShape{U}{mathc}{m}{it}%
{<->s*[1.03] mathc10}{}
\DeclareMathAlphabet{\mathscr}{U}{mathc}{m}{it}

\newcommand{\pv}{\mbox{;}\,}
\newcommand{\doubleangle}[1]{\langle\kern-0.2em\langle #1 \rangle\kern-0.2em\rangle}

\def\xsum{\mathop{\sum\nolimits'}}
\def\xprod{\mathop{\prod\nolimits'}}
\newcommand{\R}{\mathbb{R}}

\newcommand{\I}{\text{i}} 
\newtheorem{teorema_}{Theorem}{\bf}{\it}
\newtheorem{proposition}{Proposition}{\bf}{\it}
{\bf}{\it}
{\bf}{\it}
{\bf}{\it}
{\bf}{\it}

{\bf}{\it}
\newtheorem{lemma}{Lemma}{\bf}{\bf}

\begin{document}
\title{Exact and limit results for the CTRW in presence of drift and position dependent noise intensity}
\author{Marco Bianucci}
\email[]{marco.bianucci@cnr.it}
\affiliation{Istituto di Scienze Marine, Consiglio Nazionale delle Ricerche (ISMAR - CNR),\\
19032 Lerici (SP), Italy}
\author{Mauro Bologna}
\affiliation{Departamento de Ingenier\'ia El\'ectrica-Electr\'onica, Universidad de Tarapac\'a, Arica, Chile}
\author{Riccardo Mannella}
\affiliation{Dipartimento di Fisica, Universit\`a di Pisa, 56100 Pisa, Italy}

\date{\today}
%

%
%
\begin{abstract}
Continuous-time random walks (CTRWs) with drift and position-dependent jumps provide a
highly general framework for describing a wide range of natural and engineered systems.
We analyze the stochastic differential equation (SDE) associated with this class of
models, in which the driving noise $\xi(t)$ consists of spike (shot) events, and we
derive two exact analytical results.
First, we obtain a closed-form expression for the $n$-time correlation functions of
$\xi(t)$, expressed as a sum over all $2^{\,n-1}$ ordered partitions of the observation
times (Proposition~\ref{prop:procedure_fligh}). Second, using the $G$-cumulant
formalism, we derive an \emph{exact} non-local master equation (ME) for the probability
density function of the CTRW variable $x(t)$, valid without invoking diffusive limits,
fractional scaling assumptions, or closure hypotheses (Proposition~\ref{prop:exactME}).
In interaction representation, this ME retains the same structural form as that of the
standard CTRW without drift or position-dependent jumps.
Our main result is the emergence of a \textbf{universal local master equation}: at long
times, the exact non-local ME is universally and accurately approximated by a
time-local ME whose only coefficient is the instantaneous renewal rate $R(t)$. 
From this equation, exact in the well known Poissonian case, both local and
global properties of the PDF can be readily inferred. For example,
the temporal behavior of the PDF is directly controlled by that of the
rate function $R(t)$: if the waiting-time distribution decays as a
power law with exponent $\mu>2$, then $R(t)\to\mathrm{const}$ and the
system converges to the Poissonian equilibrium. By contrast, for
$\mu<2$, the rate decays in time and the effective diffusion induced
by the noise slowly weakens, without leading to a 
stationary state.
Numerical
experiments confirm its remarkable accuracy even far beyond regimes where a naive
time-scale separation would justify it.
\end{abstract}

\maketitle 

%

\section{\label{sec:intro}Introduction}
Many natural, social, economic, and engineered systems evolve under
the action of \emph{intermittent} external impulses whose timing is
\emph{non-Markovian} and whose effects depend sensitively on the
current state of the system. Such forcing is neither diffusive nor Gaussian, and
its statistics typically exhibit memory, heavy tails, and burstiness. A
very general representation of this interplay between deterministic
dynamics and renewal-type stochastic perturbations is the
state-modulated, renewal-driven stochastic differential equation (SDE)
\begin{equation}
    \dot x = -C(x,t) - f\!\left(x,\xi \pv t\right)[t],
    \label{eq:central_model}
\end{equation}
where $C$ is the drift (unperturbed velocity field) and
$f\!\left(x,\xi \pv t\right)[t]$ is a renewal point process whose
jumps and waiting times (WTs) may depend on the instantaneous state.

To keep notation transparent while retaining the essential structure,
we focus on the widely used and physically motivated setting in which:
(i) the drift is time-independent, $C(x,t)=C(x)$; (ii) the renewal
impulses factorize into a state-dependent gain $I(x)$ and a purely
temporal renewal process $\xi[t]$,
\(
f(x,\xi \pv t)[t] = I(x)\,\xi[t];
\)
and (iii) WTs
and jump amplitudes are independent of each
other, with PDFs $\psi(t)$ 
and $p(\xi)$, respectively, 
that do not depend on~$x$.
Under these assumptions, a realization $\xi(t)$ is a weighted sum of
Dirac-delta spikes occurring at random event times with random weights
(see Fig.~\ref{fig:trajectory_leapers}).
We use the terms “spike”, “shock”, and “intermittent”
interchangeably. Shot noise is the paradigmatic example of such
discrete-event fluctuations, appearing across many fields, from current
statistics and quantum transport to photon counting and interferometric
noise \citep{Schottky1918,Blanter2000,Loudon2000,Caves1981}. As a
modeling ingredient, spike renewal processes widely serve as driving
terms in generalized CTRW-based Langevin equations, for instance in
\emph{neuroscience} (stochastic synaptic input)
\citep{Brunel2000,Gerstner2002} and in \emph{climate science} (impulsive
forcing of slow modes) \citep{Hasselmann1976,Majda1999}.

The resulting SDE,
\begin{equation}
    \dot x = -C(x) - I(x)\,\xi[t],
    \label{SDE}
\end{equation}
constitutes a minimal yet flexible model capturing three key features
often inseparable in complex dynamical systems:
\begin{enumerate}
    \item[(i)] external deterministic forcing (the drift);
    \item[(ii)] temporal memory through non-Poissonian renewal statistics;
    \item[(iii)] multiplicative effects arising from state-dependent gain.
\end{enumerate}
This combination naturally generates non-Gaussian propagators, anomalous
transport, and history-dependent relaxation, extending far beyond the
scope of classical diffusion and Fokker--Planck descriptions
\citep{mksPRE58}.

When $C(x)=0$ and $I(x)\equiv 1$, Eq.~\eqref{SDE} reduces to the
standard CTRW velocity model $\dot x = \xi[t]$
\citep{Montroll_Weiss_1965,Scher_Motroll_1975}. In this classical
framework, heavy-tailed WTs
lead to 
asymptotic 
fractional diffusion
equations, while external forces and space-dependent drifts produce
    generalized fractional Fokker--Planck equations
\citep{mksPRE58,hlsPRL105,Torrejon2018_SIAM,Kolokoltsov2007_arXiv,Meerschaert2018_Handbook}.
Equation~\eqref{SDE} extends these constructions by introducing
unperturbed driving drift, 
state-dependent (and thus multiplicative) spike noise, 
without making any assumption about the WT and jump PDFs. 

Non-Poissonian impulses with drift and state-dependent amplitude arise
in numerous applications. In climate dynamics, the growth and phase of
the El~Ni\~no--Southern Oscillation are modulated by intermittent
atmospheric bursts whose non-exponential statistics and state-dependent
impact lead to non-Gaussian predictability regimes
\citep{bGRL43,bcmmA2018,bcmmCHAOS28,sGL10,llJC38}. Early-warning signals
for climatic tipping depend sensitively on the noise law, and classical
variance-based indicators may fail under heavy-tailed forcing
\citep{mbPRX14,Layritz2023_arXiv}. In computational neuroscience,
synaptic shot noise is well described by renewal processes with
state-dependent efficacy, producing voltage fluctuations and
first-passage statistics that elude diffusion approximations
\citep{Privault2020,rgNC17,EPFL_RenewalBook}. In materials science,
intermittent crack-growth increments with Weibull- or
lognormal-distributed waiting times provide direct examples of
state-dependent renewal forcing \citep{NIST_Weibull,yAISM46,DayGoswami2002}.

\paragraph{Objectives and scope of this work}
This paper has a dual objective.

First, we derive an explicit and self-contained exact expression for the
\emph{multi-time correlation functions} of spike-type renewal processes,
$\langle \xi(t_1)\cdots\xi(t_n)\rangle_{t_0}$, generalizing and
completing our recent results for two-time correlations and for
step-type renewal processes \citep{bblmCSF196,bblmCSF202}.

The second objective is to use the previous result to derive,
via the $G$-cumulant formalism
\citep{bbJSTAT4,bCSF148,bCSF159,bmJPC4,bbmJSP191},
the \emph{exact nonlocal master equation} (ME), given in
Eq.~\eqref{ME_fin}, for the PDF $P(x\pv t)$ of the variable $x$
in Eq.~\eqref{SDE}. This result, stated as
Proposition~\ref{prop:exactME} in Section~\ref{sec:MEx},
is valid for arbitrary WT and jump PDFs with finite moments.
It requires no diffusive or fractional scaling limit and
invokes no phenomenological closure. 
Notably, when expressed in the
interaction representation, the structure of this ME is
identical to the well-known one of the CTRW case
(i.e., $C(x)=0$ and $I(x)=1$).
To the best of our
knowledge, this is the first time that an exact ME of such
generality has been obtained.

The central and structurally organizing result of this paper, however,
is not the exact ME in itself but rather its remarkable 
 simplification, precisely formulated by
Theorem~\ref{prop:ULME}, and less formally summarized as
follows.

\begin{proposition}[Universal Local ME — informal statement]
\label{prop:ULME_informal}
The exact nonlocal ME~\eqref{ME_fin} collapses, in a precise
asymptotic sense, onto a local-in-time equation formally
identical to the Poissonian ME, with the constant rate
$1/\tau$ replaced by the time-dependent renewal rate $R(t)$:
\begin{empheq}[box=\fbox]{align}
\label{ME_universal_intro}
\partial_t P(x\pv t)
\;\approx\;
\partial_x C(x)\,P(x\pv t)
\;+\;
R(t)\,\left[\hat{p}(\I\,\partial_x I(x))-1\right]\,P(x\pv t).
\end{empheq}
This holds both when the mean waiting time $\tau$ is finite
(Proposition~\ref{muG2}) and when it diverges with
$1<\mu<2$ (Proposition~\ref{prop:approxME}). It is not a
phenomenological approximation: it follows systematically
from the $G$-cumulant structure and from the dominance of
specific compositions in the correlation-function sum
(Propositions~\ref{prop:universal} and~\ref{prop:PDF_power}).
\end{proposition}

\noindent
This result admits a clear physical interpretation:
\emph{all the non-Markovian complexity of the renewal
process is distilled, at the level of the macroscopic
PDF, into the single scalar function $R(t)$ evaluated
at the local time $t$}. When $R$ is constant the
equation is exact; when $R(t)$ decays as a power law it
quantitatively captures the progressive quenching of
stochastic forcing.
The remainder of the paper is organized around
Proposition~\ref{prop:ULME_informal}.
Section~\ref{sec:summary} gives a concise map of all results.
Sections~\ref{sec:corr_leapers_2-4}--\ref{sec:corr_leapers} establish
the exact correlation functions (Proposition~\ref{prop:procedure_fligh}).
Section~\ref{sec:MEx} derives the exact ME
(Proposition~\ref{prop:exactME}) and the $G$-cumulant identification
that is the key algebraic step.
Sections~\ref{sec:further}--\ref{sec:universalCTRW} classify the
simplifications leading to
Proposition~\ref{prop:ULME_informal}, Eq.~\eqref{ME_universal_intro}, 
and establish its regime of validity.
Section~\ref{sec:numericaltechinques} discusses numerical implementation
and extensive validation.

Before that, however, it is worth noting that when $\xi[t]$ is shot
noise and $I(x)$ is nonconstant, a convention is required for
evaluating $I(x(t))$ ``during a spike'' (pre-, post-, or mid-kick),
analogously to It\^o vs.\ Stratonovich interpretations for white noise.
Throughout this work we adopt the \emph{Stratonovich prescription}
(see Sec.~\ref{sec:numericaltechinques}), corresponding to representing
each Dirac spike as the zero-width limit of a fixed-area smooth pulse.
This choice is natural within a stochastic Liouville framework along
single trajectories, which is the perspective used in this work.

\begin{figure}[ht]
\centering
\includegraphics[width=\textwidth]{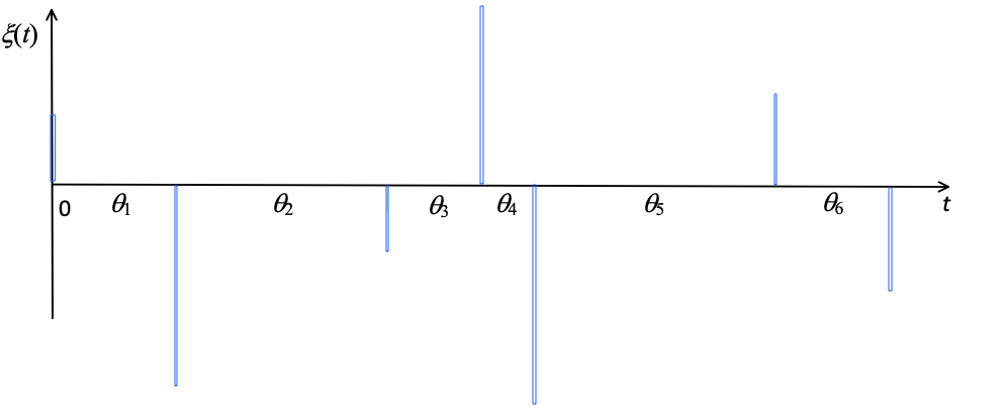}
\caption{A trajectory realization $\xi(t)$
for the case of L\'evy flight-CTRW, with $t_0=0$. We have $\xi( t)=\sum_{q=0}^{\infty} \xi_q \; \delta\left(t-\sum_{k=0}^q \theta_k\right)$, (see text for details). }
\label{fig:trajectory_leapers}
\end{figure}


\section{Summary of the main results.\label{sec:summary}}

The paper establishes three interlocking results, all organized around
the Universal Local ME stated in
Proposition~\ref{prop:ULME_informal} above.

\paragraph{Result 1: Exact $n$-time correlation functions}
(Proposition~\ref{prop:procedure_fligh}, Section~\ref{sec:corr_leapers}.)
The $n$-time joint correlation function
$\langle\xi(t_1)\cdots\xi(t_n)\rangle_{t_0}$ is given by a sum over
all $2^{n-1}$ ordered partitions (compositions) $\pi(n)$ of the times
$t_1\leq\cdots\leq t_n$:
\begin{align}
\label{corrGen_delta_pre}
\langle \xi(t_1)\xi(t_2)\cdots\xi(t_n)\rangle_{t_0}
=
\sum_{\pi(n)}
\prod_{B\in\pi(n)}
\overline{\xi^{\lvert B\rvert}}\,
R(t_B - t_{B-1})\,
\delta(\Delta t_{B}).
\end{align}
Each partition into $p$ blocks contributes a product of jump moments
$\overline{\xi^{m_i}}$, rate functions $R$, and Dirac deltas enforcing
time coincidences within blocks.
This expression is exact, requires no assumptions on the waiting-time or
jump distributions, and generalizes our earlier two-time result
\citep{bblmCSF196}.
The connection to the $G$-cumulant formalism
(see~\cite[Section~4.4.3, Eq.~(94)]{bbJSTAT4}) is immediate: the
partition structure of Eq.~\eqref{corrGen_delta_pre} mirrors
the moment--$G$-cumulant relation, which is what makes the subsequent ME
derivation exact.

\paragraph{Result 2: Exact nonlocal ME}
(Proposition~\ref{prop:exactME}, Section~\ref{sec:MEx}.)
Inserting Result~1 into the $G$-cumulant framework applied to the full
SDE~\eqref{SDE} yields the exact ME for $P(x\pv t)$:
\begin{align}
\label{ME_fin}
\partial_t P(x\pv t)
=&\;\partial_x C(x)\, P(x\pv t)
\nonumber \\
&+
\left[\hat p\!\left(\I\partial_x I(x)\right)-1\right]
\Bigg[
\int_{0}^{t}\!{\rm d}u\,
R'(t-u)\,
e^{\partial_x C(x) u}
P(x\pv u)
+ R(0)\,P(x\pv t)
\Bigg].
\end{align}
This is a partial integro-differential equation coupling drift and
memory.
It is exact and reduces, when $R$ is constant (exponential WT), to the
known Poissonian ME:
\begin{align}
\label{ME_fin_Pois}
\partial_t P(x\pv t)
=
\partial_x C(x)\, P(x\pv t)
+
\frac{1}{\tau}
\left[
\hat p\!\left(\I\partial_x I(x)\right)-1
\right]
P(x\pv t).
\end{align}
The reader should appreciate that the structure of the general
exact ME~\eqref{ME_fin}, when written in the interaction
representation (see Eq.~\eqref{MEGx_I}), \emph{is identical to
that of the well-known standard CTRW case} (corresponding to
the SDE~\eqref{SDE} when $C(x)=0$ and $I(x)=1$), reported here
in Eq.~\eqref{MEGx_}~\citep{Montroll_Weiss_1965,Scher_Motroll_1975}
(for details see Section~\ref{sec:Gcumulants_x}).
Although theoretically exact, Equation~\eqref{ME_fin} has been validated by extensive numerical
simulations, a representative subset of which is shown in
Figs.~\ref{fig:def_comp_t0.2_pde_dico-gauss_mu1.50_tz1.00}--\ref{fig:def_comp_t0.2_pde_dico_mu3.50_tz1.00}
(Note that in the latter figure we also exploit the analytical results
of Appendix~\ref{app:ME_resuls}; for details, see the figure caption).

\paragraph{Result 3: Universal Local ME Theorem}
(Propositions~\ref{muG2} and~\ref{prop:approxME},
Sections~\ref{sec:PLWTmu_gt_2}--\ref{sec:universalCTRW}.)
Equation~\eqref{ME_fin} admits, in a precise asymptotic regime, the
local-in-time reduction~\eqref{ME_universal_intro}.
The two cases are:
\begin{itemize}
    \item \label{en:1}\emph{Finite mean waiting time $\tau$} (Proposition~\ref{muG2}):
          by the Blackwell renewal theorem
          \citep{Blackwell1948,FellerVol2}, $R(t)\to\tau^{-1}$,
          so Eq.~\eqref{ME_universal_intro} converges to the Poissonian
          ME~\eqref{ME_fin_Pois}.  
    In this case, over long time intervals, the PDF converges 
    toward the equilibrium one, provided that this relaxation process is 
    compatible with the characteristics of the dynamics induced by the drift 
    field (see, e.g.,  Figs.~\ref{fig:Def_comp_t100-30_pde_eq80_mu2.50_tz2.00__ga1.2_Cubic_appr} and \ref{fig:Def_comp_t0.5-2-6_pde_eq80_mu2.50_tz1.00_ga5.00_be0.50}). 

    \item \label{en:2}\emph{Infinite mean waiting time} ($1<\mu<2$,
          Proposition~\ref{prop:approxME}):
          the dominant composition in the sum~\eqref{corrGen_delta_pre}
          is the single-block ($p=1$) term
          (Proposition~\ref{prop:universal}),
          which reduces all $G$-cumulants to their single-block form,
          yielding Eq.~\eqref{ME_universal_intro} with
          $R(t)\sim T^{-1}(T/t)^{2-\mu}$. In this case, the effect of the noise progressively weakens as time
          increases, but the PDF never reaches an equilibrium state (see, e.g.,
          Figs.~\ref{fig:Def_pdedico_t5-30-100_mu1.5_tz2.00_ga1.2_Cubic_appr}
          and
          \ref{fig:Def_comp_t0.5-2-6_pde_eq80_mu1.50_tz1.00_ga5.00_be0.50}).
\end{itemize}
In both cases the approximation is not a closure: it is a structural
consequence of the renewal algebra.
Equation~\eqref{ME_universal_intro} has been validated by extensive
numerical simulations across a wide range of parameters, drift
strengths, and jump distributions, including cases well outside the
formal regime of the proofs. A representative subset of these results
is shown in
Figs.~\ref{fig:Def_pdedico_t5-30-100_mu1.5_tz2.00_ga1.2_Cubic_appr}--\ref{fig:Def_comp_t0.5-2-6_pde_eq80_mu2.50_tz1.00_ga5.00_be0.50}.
As for Fig.~\ref{fig:def_comp_t0.2_pde_dico_mu3.50_tz1.00}, in
Fig.~\ref{fig:Def_comp_t0.5-2-6_pde_eq80_mu2.50_tz1.00_ga5.00_be0.50}
we also exploit the analytical results of Appendix~\ref{app:ME_resuls}.
For details, see the figure captions.

%
\begin{figure}
    \centering
    \includegraphics[width=1.\linewidth]{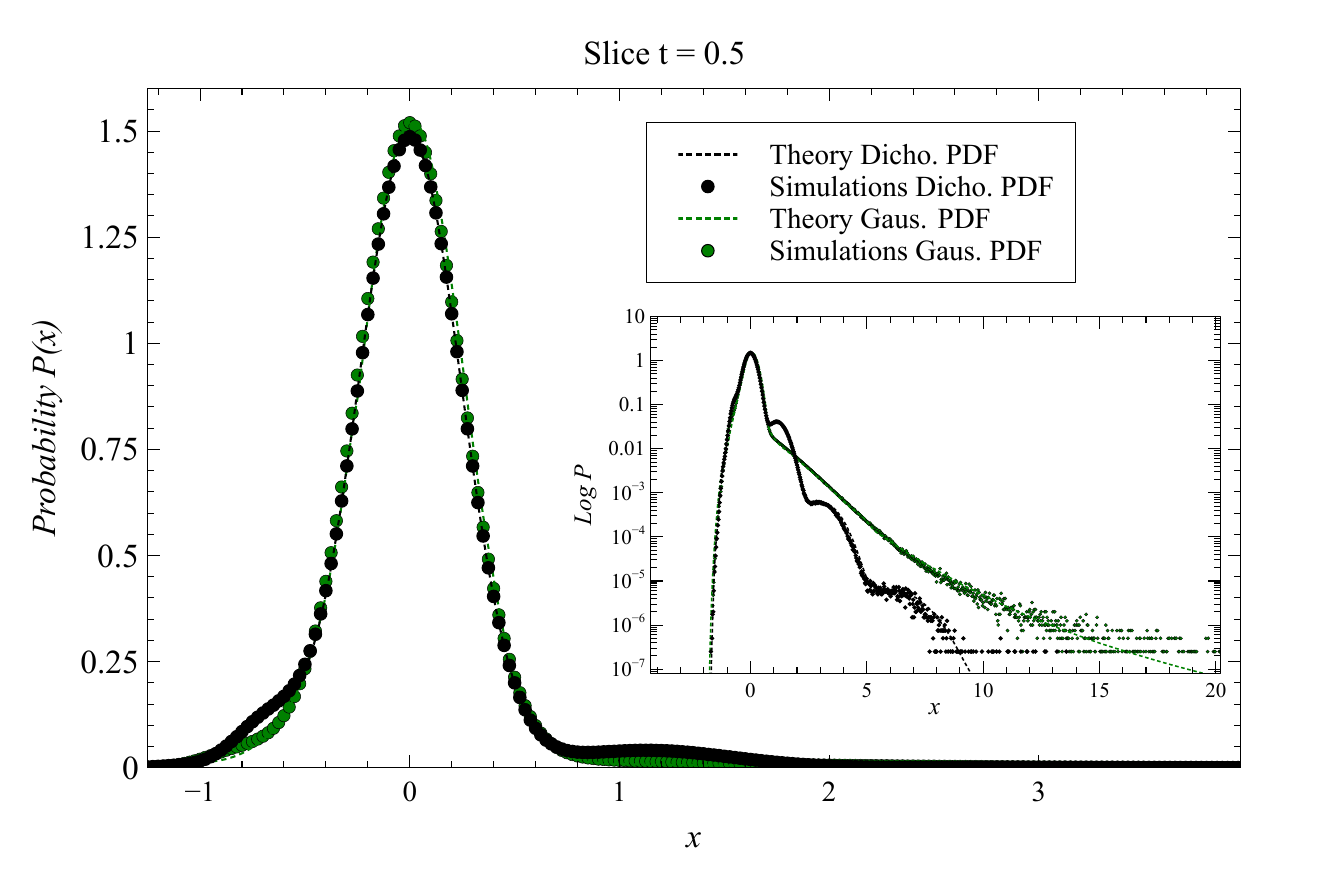}
    \caption{\( P(x\pv t) \) at times  \( t = 0.5 \) for the
SDE~\eqref{SDE} in the \textit{multiplicative} case, i.e., with \( C(x) =\gamma\,x \) and
\( I(x)=1+\beta \,x \), where $\gamma=1.0$ and $\beta =0.5$. The jump PDF is both dichotomous (black color) and Gaussian (green color), the WT is 
                $\psi(t)=(\mu -1){T}^{-1}\left(1+t/T\right)^{-\mu }$ with $T=1$ and $\mu=1.5$. The figure shows the results of numerical
simulations of the SDE together with the 
                 solution of the exact theoretical  PDE in Eq.~\eqref{ME_fin}
                 (dashed lines, barely visible). The insert shows the same plot but in log scale. 
As expected, the agreement
between theory and simulation is perfect. }
    \label{fig:def_comp_t0.2_pde_dico-gauss_mu1.50_tz1.00}
\end{figure}
\begin{figure}
    \centering
    \includegraphics[width=1.0\linewidth]{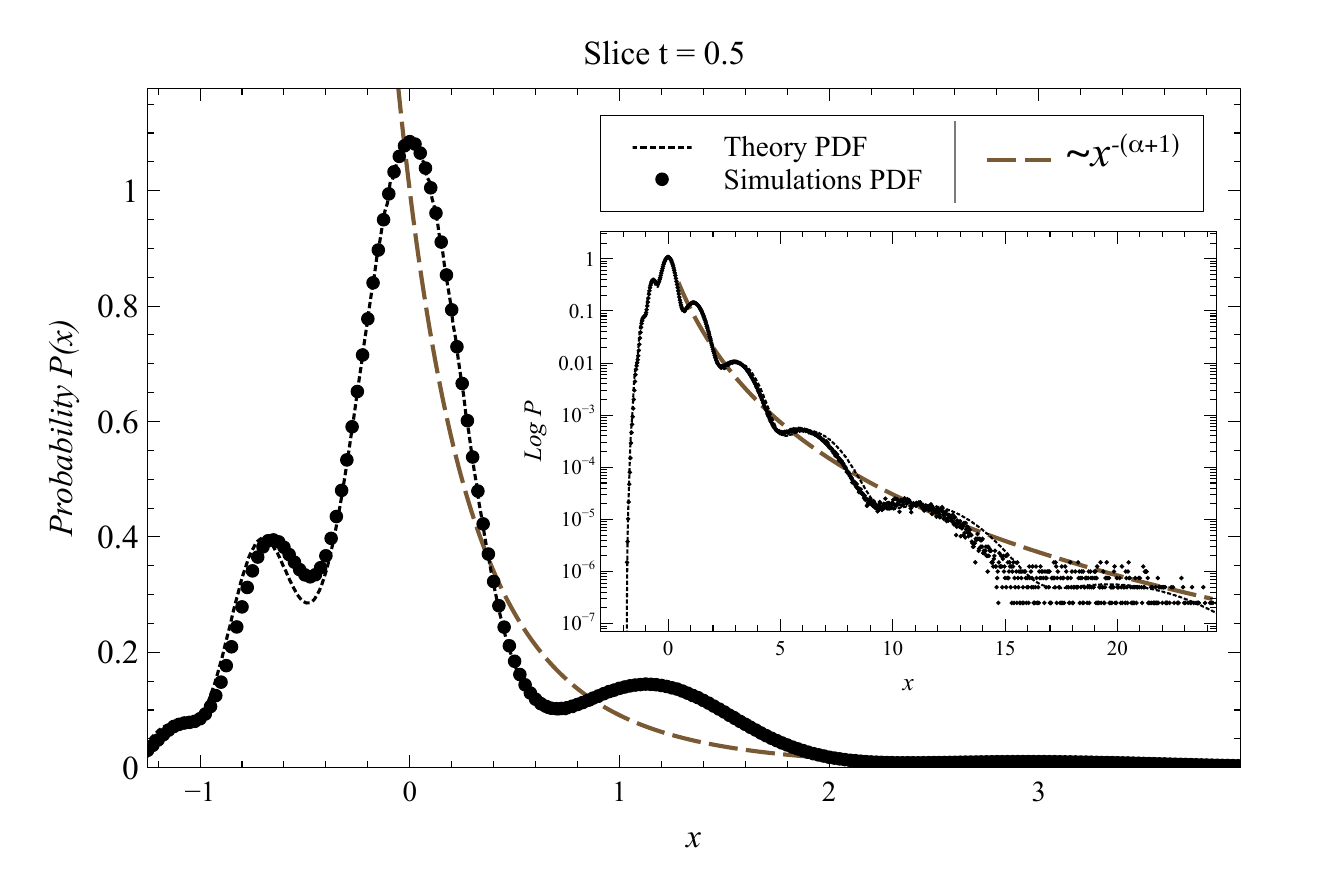}
    \caption{
    The same as fig.~\eqref{fig:def_comp_t0.2_pde_dico-gauss_mu1.50_tz1.00}, but for $\mu=3.5$ (from which $\tau=T/(\mu-2)=2/3$)
     and for the sole dichotomous jump PDF case. Because $\mu>2$ from the result~3, 
point~\ref{en:1}, we can exploit the analytical findings of Appendix~\ref{app:ME_resuls}.
      For example,  the tail of the PDF goes as $x^{-(\alpha+1)}$, where $\alpha\approx3.9$ is obtained 
    from the transcendental implicit Eq.~\eqref{power2} (dashed brown line).    
    }
    \label{fig:def_comp_t0.2_pde_dico_mu3.50_tz1.00}
\end{figure}

\begin{figure}[!htbp]
\centering
\includegraphics[width=0.8\textwidth]{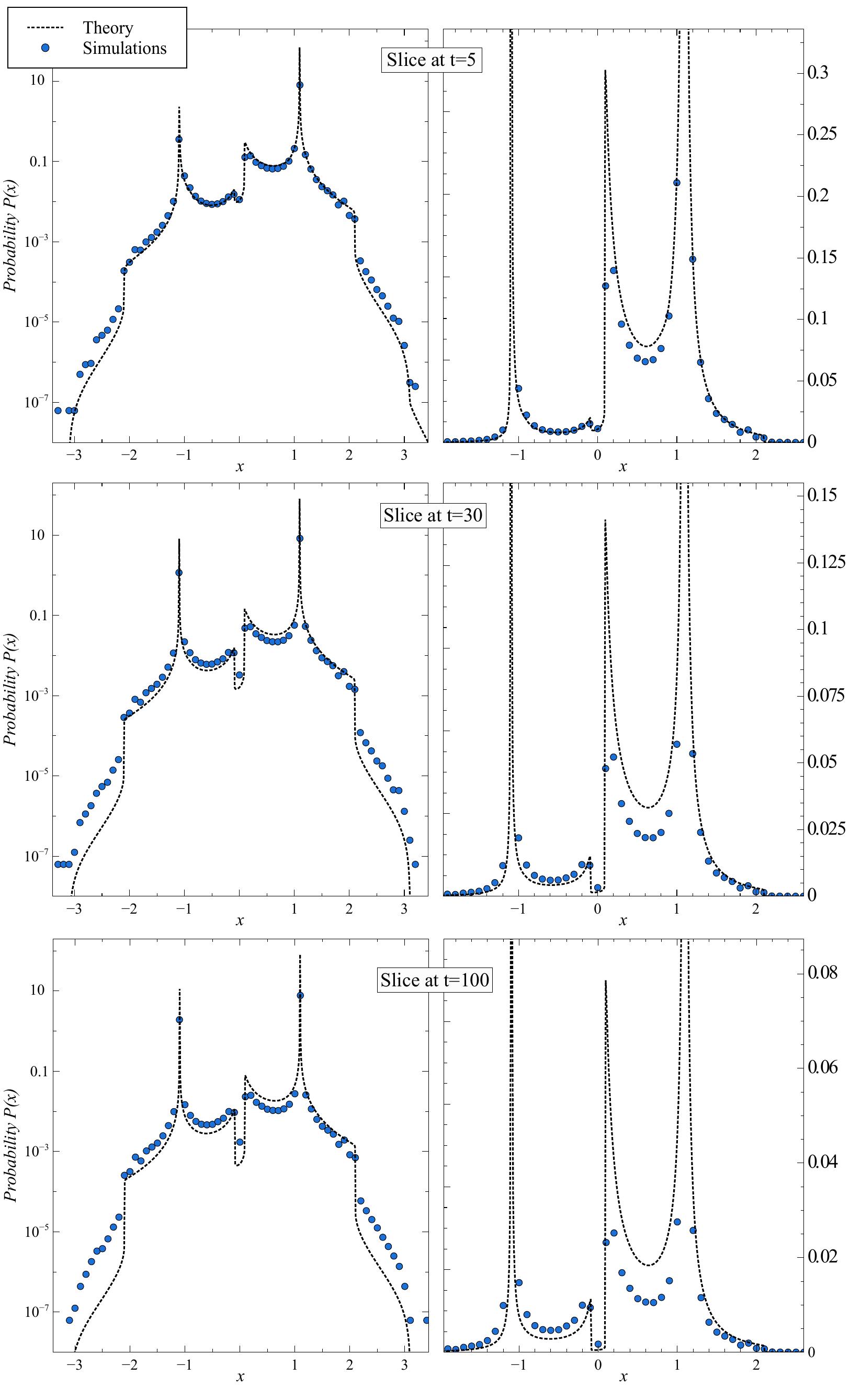}
\caption{Log-plot (left) and linear plot (right) of \( P(x,t) \) at 
$t=5$ (top), $t=30$ (center), and $t=100$ (bottom) for the SDE~\eqref{SDE}
in the \textit{additive} case (i.e., with $I(x)=1$) with strongly nonlinear drift,
\( C(x) = -x(1.2 - x^2) \).
The system is driven by renewal noise with dichotomous jumps PDF and WT PDF
\(
\psi(t)=(\mu -1)T^{-1}(1+t/T)^{-\mu}\), with $T=2$ and $\mu=1.5$ (infinite aging).
Circles correspond to direct numerical simulations of
the SDE~\eqref{SDE}, while the dotted curves show the solution of
the \textit{approximate} local-in-time ME given in
Eq.~\eqref{ME_universal_intro}.  
The vertical scale on the right hand plates decreases
as time increases, to show the vanishing of the PDF as predicted by Eq.~\eqref{ME_universal_intro}, and it is clipped:
at the two equilibrium points the PDF reaches a stationary value outside the range shown (see the left hand plates)
The agreement between theory and simulations is very good. }
\label{fig:Def_pdedico_t5-30-100_mu1.5_tz2.00_ga1.2_Cubic_appr}
\end{figure}
 
\begin{figure}[ht]
\centering
\includegraphics[width=1.0\textwidth]{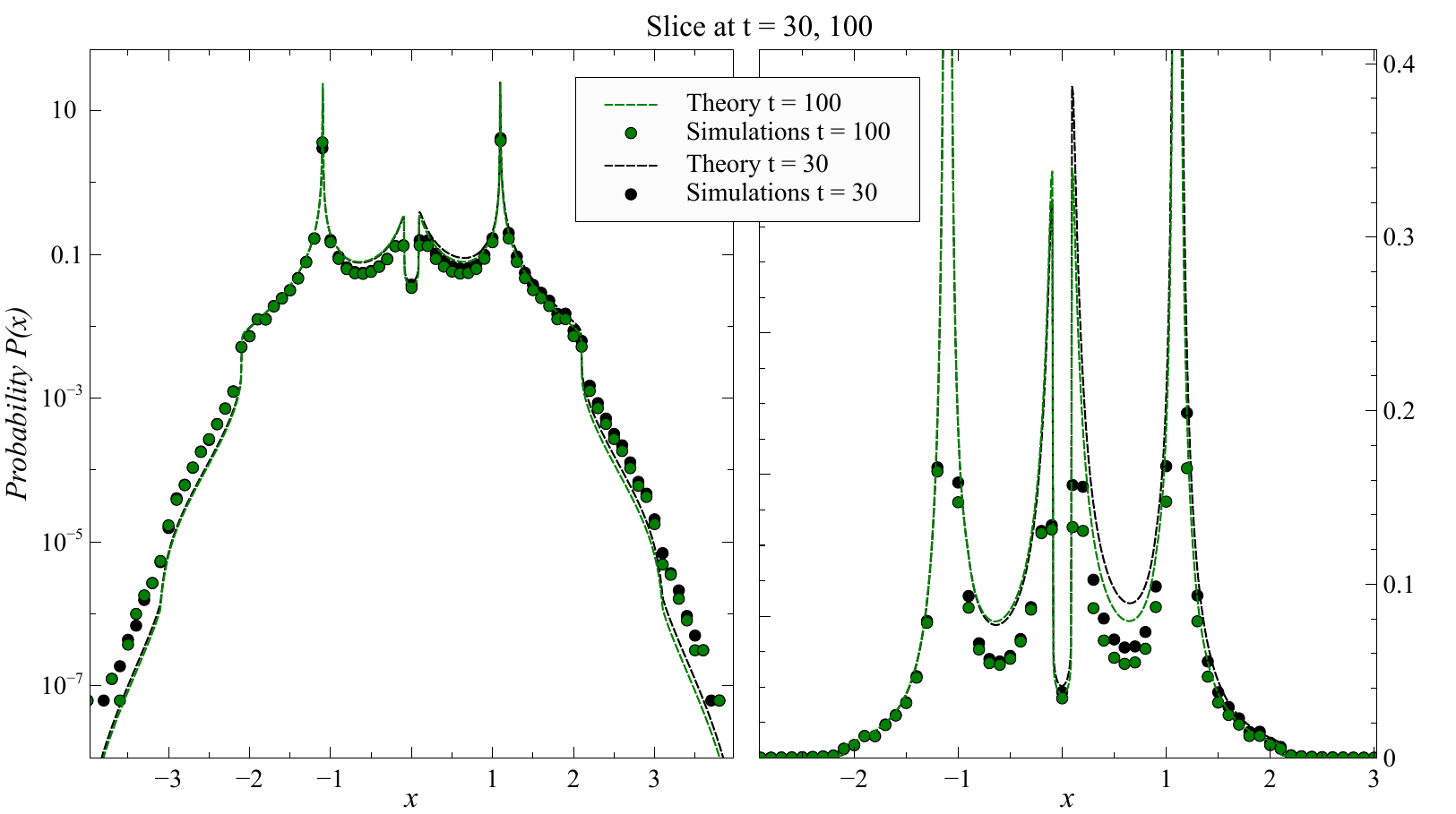}
\caption{Same as Fig.~\ref{fig:Def_pdedico_t5-30-100_mu1.5_tz2.00_ga1.2_Cubic_appr},
but with $\mu=2.5$ and shown only for $t=30$ and $t=100$.
The two curves are almost perfectly superimposed, indicating that the
Poissonian-like equilibrium has effectively been reached. 
}
\label{fig:Def_comp_t100-30_pde_eq80_mu2.50_tz2.00__ga1.2_Cubic_appr}
\end{figure}

\begin{figure}[ht]
\centering
\includegraphics[width=1.0\textwidth]{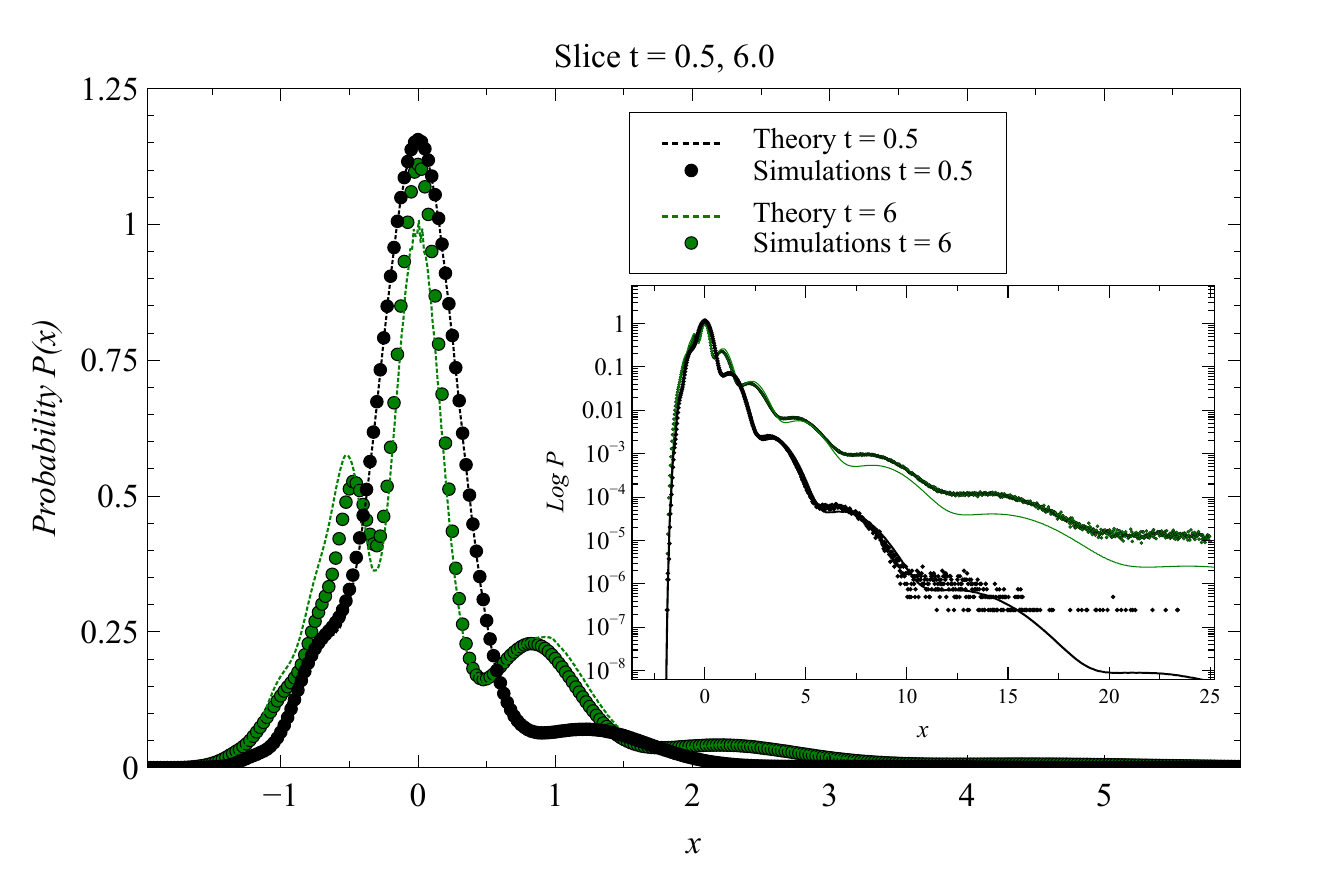}
\caption{Plot and log-plot (inset) of \( P(x,t) \) at $t=0.5$ (black) and
$t=6.0$ (green) for the SDE~\eqref{SDE} in the
\textit{multiplicative} case, with
\( C(x) = -\gamma x \) and \( I(x) = 1 + \beta  x \),
where $\gamma=0.1$ (slow drift) and $\beta =0.5$.
The system is driven by the same renewal noise with power-law
WT as in
Fig.~\ref{fig:Def_pdedico_t5-30-100_mu1.5_tz2.00_ga1.2_Cubic_appr},
i.e.\ with $\mu=1.5$ (infinite aging).
Circles represent direct numerical simulations of the
SDE~\eqref{SDE}, while the dotted curves correspond to the solution
of the \textit{approximate} local-in-time
ME
given in 
Eq.~\eqref{ME_universal_intro}.
The agreement between theory  (although in the simplified version of Eq.~\eqref{ME_universal_intro}) and simulations is very good.
}
            \label{fig:Def_comp_t2-6_pde_eq80_mu1.50_tz1.00_ga0.10_be0.50}
\end{figure}

\begin{figure}[ht]
\centering
\includegraphics[width=1.0\textwidth]{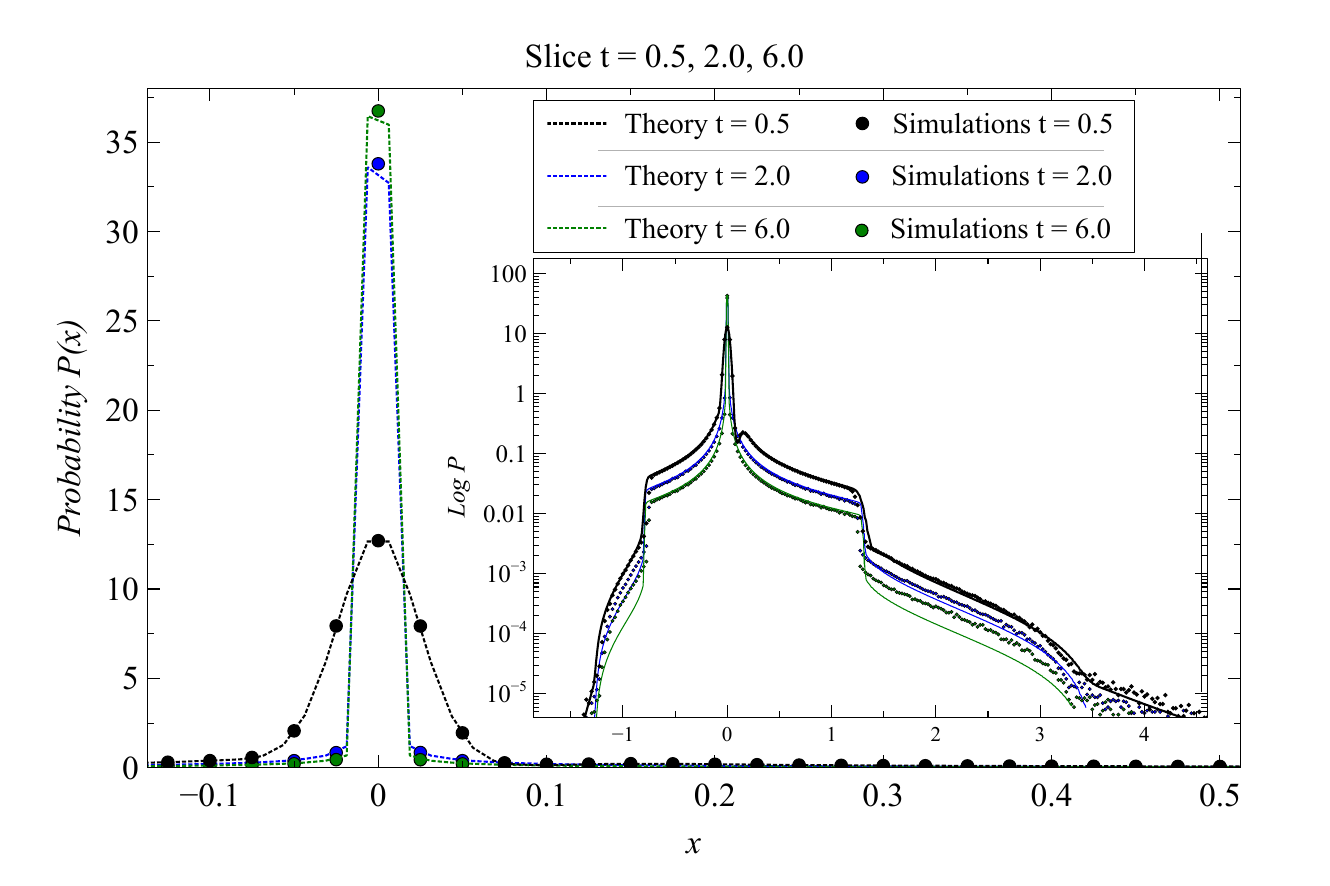}
\caption{Plot and log-plot (inset) of \( P(x,t) \) at $t=0.5$ (black).$t=2.0$ (blue)  and
$t=6.0$ (green) for the SDE~\eqref{SDE} in the
\textit{multiplicative} case, with
\( C(x) = -\gamma x \) and \( I(x) = 1 + \beta  x \),
where $\gamma=5$ (fast drift) and $\beta =0.5$.
The system is driven by the same renewal noise with power-law
WT as in
Fig.~\ref{fig:Def_pdedico_t5-30-100_mu1.5_tz2.00_ga1.2_Cubic_appr},
i.e.\ with $\mu=1.5$ (infinite aging).
Circles represent direct numerical simulations of the
SDE~\eqref{SDE}, while the dotted curves correspond to the solution
of the \textit{approximate} local-in-time
ME
given in 
Eq.~\eqref{ME_universal_intro}.  Compare this case to the case of Fig.~\ref{fig:Def_comp_t2-6_pde_eq80_mu1.50_tz1.00_ga0.10_be0.50}: here, for $t=6$ the dissipative drift is dominant: the solutions tends to a Dirac-delta function centered in $x=0$.
The agreement between theory (although in the simplified version of Eq.~\eqref{ME_universal_intro}) and simulations is very good.
            }
\label{fig:Def_comp_t0.5-2-6_pde_eq80_mu1.50_tz1.00_ga5.00_be0.50}
\end{figure}

        \begin{figure}[ht]
\centering
\includegraphics[width=1.0\textwidth]{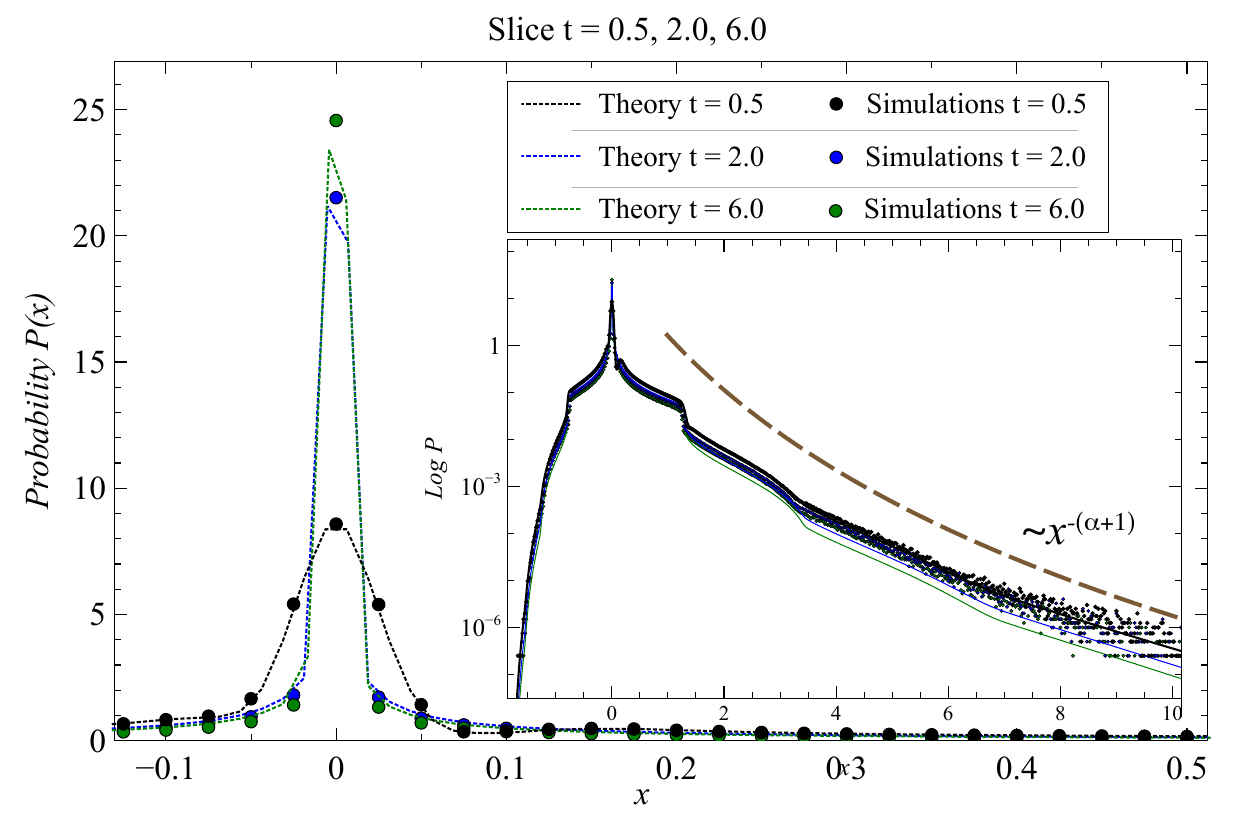}
\caption{
The same  as Fig.~\ref{fig:Def_comp_t0.5-2-6_pde_eq80_mu1.50_tz1.00_ga5.00_be0.50},
but with $\mu=2.5$, from which $\tau=2$.
In contrast to the heavy-tailed case, a characteristic time scale emerges for the convergence toward the corresponding Poissonian equilibrium PDF. The PDF agrees with the analytical findings of Appendix~\ref{app:ME_resuls}. For example, because $\gamma \tau=2>1$,
 close to $x=0$ the equilibrium solution (here for $t\gtrsim 6$) has an integrable divergence (see text after Eq.~\eqref{MEGx_dichoquat}. 
Moreover, for large $x$ values, the PDF scales as $x^{-(\alpha+1)}$,
where $\alpha \approx 6.7$, obtained from the transcendental implicit
Eq.~\eqref{power2} (dashed brown line).}
\label{fig:Def_comp_t0.5-2-6_pde_eq80_mu2.50_tz1.00_ga5.00_be0.50}
\end{figure}
%

%
%
\section{Preliminary  considerations and definitions\label{sec:preliminary}}

We assume
the time ordering defined by the notation $t_i \le t_j$ for $i <j$.
We denote with  $t_0$ the initial time at which the stochastic process begins. 
The time lag $t_1-t_0$ is important for measuring the aging of the
process.
The average process, indicated by the angle brackets $\langle...\rangle_{t_0}$,
is over
all the possible trajectories  $\xi(t)$, starting at the time $t_0$,   
each one corresponding
to a realization of the stochastic
process $\xi[ t]$, with its proper functional differential measure  $ P_{t_0}[\xi(t)]\delta\xi(t)$:
\begin{equation}
\label{n_corr_def}
\langle\xi( t_1)\xi( t_2)...\xi( t_n)\rangle_{t_0} =\int \xi( t_1)\xi( t_2)...\xi( t_n) P_{t_0}[\xi(t)]\delta\xi(t).
\end{equation}

Because all $\xi_k$ and  $\theta_k$ are independent random numbers, the PDF for the trajectory realization is 
\begin{equation}
\label{PDF_trajectory}
P_{t_0}[\xi(t)]\delta\xi(t)=p_0(\xi_0)\text{d}\xi_0\prod_{q=1}^\infty\psi(\theta_{\!q})\text{d}\theta_{\!q}\,p(\xi_q)\text{d}\xi_q
\end{equation}
where $\xi_0$ represents the value of $\xi$ at the initial time $t_0$, and 
$p_0(\xi_0)$ is the PDF used to sample the initial value 
of $\xi$ (i.e., at $t = t_0$). For example, if all trajectories start from the same initial 
value $\xi'$ (e.g., $\xi' = 0$), then $p_0(\xi_0) = \delta(\xi_0 - \xi')$. 
On the other hand, if the initial value $\xi_0$ is a random number with the same 
PDF as the random variable $\xi$, then $p_0(\xi_0) = p(\xi_0)$.
Of course, the influence of the initial PDF 
$p_0(\xi_0)$ is particularly significant (i.e., persistent over time) when the aging 
time of the process is long or infinite.

The average of a function of the random number $\xi$ will be indicate with a bar over the same
function, i.e.  
$\int f(\xi)p(\xi) \text{d}\xi:= \overline{f(\xi)}$.
Thus,  
$\int \xi^np(\xi) \text{d}\xi:= \overline{\xi^n}$. 
    With some abuse of notation, we also define  
$\int \xi_0^{n} p_0(\xi_0) \text{d}\xi_0 := \overline{\xi_0^{n}}$.
    A hat over a function, with argument $s$, denotes its Laplace
transform.

    We introduce some standard quantities
commonly used in the theory of stochastic renewal processes.

Let $\psi(t)$ be the waiting-time (WT) PDF, the mean WT is defined as:
\begin{equation}
\label{tau}
\tau := \int_0^\infty t\,\psi(t)\,\mathrm{d}t .
\end{equation}

    Let $\psi_i(t)$ be the probability density that the
$i$-th event occurs at time $t$: $\psi_i(t)$ can be 
    written as\footnote{$\psi_i(t)$  is also known as the $i$-fold convolution of $\psi(t)$ and Eq.~\eqref{psi_i} provides an alternative
definition.}
\begin{equation}
\label{psi_i}
\psi_i(t)
=\int \delta\!\left(t-\sum_{k=1}^i\theta_k\right)
\prod_{k=1}^i \psi(\theta_k)\,\mathrm{d}\theta_k .
\end{equation}

The probability density for an event to occur exactly at time $t$ is then given by
the rate function $R$, defined as 
\begin{equation}
\label{Rtilde}
R(t-t_0)
:= \sum_{i=1}^{\infty} \psi_i(t-t_0)
\;\Rightarrow\;
\hat{R}(s) = \frac{\hat{\psi}(s)}{1-\hat{\psi}(s)},
\end{equation}
%
    %
%

A related quantity is
\begin{equation}
\label{R}
\tilde{R}(t-t_0)
:= R(t-t_0) + \delta(t-t_0)
\;\Rightarrow\;
\hat{\tilde{R}}(s) = \frac{1}{1-\hat{\psi}(s)}.
\end{equation}
By defining $\psi_0(\theta) := \delta(\theta)$, $\tilde{R}$ can also be written as
$$\tilde{R}(t-t_0)
= \Theta(t-t_0)\sum_{n=0}^{\infty} \psi_n(t-t_0).$$

In the special case of an exponential WT distribution,
$\psi(\theta)=\tau^{-1}e^{-\theta/\tau}$, the rate function reduces to a
constant value: $R(t-t_0)=1/\tau$, corresponding to the usual Poissonian event
rate. By contrast, if the WT PDF exhibits a power-law decay,
$\psi(\theta)\sim T^{-1}(T/\theta)^{\mu}$, where $T$ is a characteristic time
scale, the rate function becomes time dependent. In this case, the mean waiting time $\tau$,
defined as in Eq.~\eqref{tau},
is finite for $\mu>2$, and $R(t)\sim \tau^{-1}[1+(T/t)^{\mu-2}]$ at large
times. For $1<\mu<2$, the mean waiting time diverges and asymptotically
$R(t)\sim T^{-1}(T/t)^{2-\mu}$ (e.g., \cite{bJSTAT2020}).

We turn now to the derivation of the multi-time correlation functions of the
stochastic process $\xi[t]$.

%
%
%
\section{The general frame and the heuristic approach to the two and four times correlation functions\label{sec:corr_leapers_2-4}}

A trajectory realization $\xi(t)$ of the stochastic process $\xi[t]$, depicted in 
Fig.~\ref{fig:trajectory_leapers}, consists of a sequence of Dirac-delta functions separated by
random times $\theta_1,$  $\theta_2,\ldots,$ $\theta_k, \ldots$, 
drawn from the WT PDF $\psi(\theta)$. 
Each pulse is weighted by a random amplitude $\xi_1, \xi_2, \xi_3, \ldots, \xi_k, \ldots$, 
sampled independently from the PDF $p(\xi)$.

In formula, the spike-function of time of  Fig.~\ref{fig:trajectory_leapers} can be written as 
(we set $\theta_0:=0$):
\begin{equation}
\label{trajectory_delta}
\xi( t)=\sum_{q=0}^{\infty} \xi_q \; \delta\left(t-t_0-\sum_{k=0}^q \theta_k\right) 
\end{equation}
where $\delta(t)$ is the Dirac-delta function.

From \eqref{n_corr_def},  \eqref{PDF_trajectory} and \eqref{trajectory_delta} 
we get
\begin{align}
\label{n_corr_def_2_delta}
&  \langle\xi( t_1)\xi( t_2)...\xi( t_n)\rangle_{t_0}\nonumber \\
&=\int \sum_{i_1=0}^{\infty} \xi_{i_1} \; \delta\left(t_1-t_0-\sum_{k_1=0}^{i_1} \theta_{k_1}\right)
\sum_{i_2=i_1}^{\infty} \xi_{i_2} \; \delta\left(t_2-t_0-\sum_{k_2=0}^{i_2} \theta_{k_2}\right)
\times...\nonumber \\
&...\times\sum_{i_n=i_{n-1}}^{\infty} \xi_{i_n} \; \delta\left(t_n-t_0-\sum_{k_n=0}^{i_n} \theta_{k_n}\right)
p_0(\xi_0)\text{d}\xi_0\prod_{q=1}^\infty\psi(\theta_{\!q})\text{d}\theta_{\!q}\,p(\xi_q)\text{d}\xi_q
\end{align}
%
%
\begin{figure}[ht]
\centering
\includegraphics[width=0.5\textwidth]{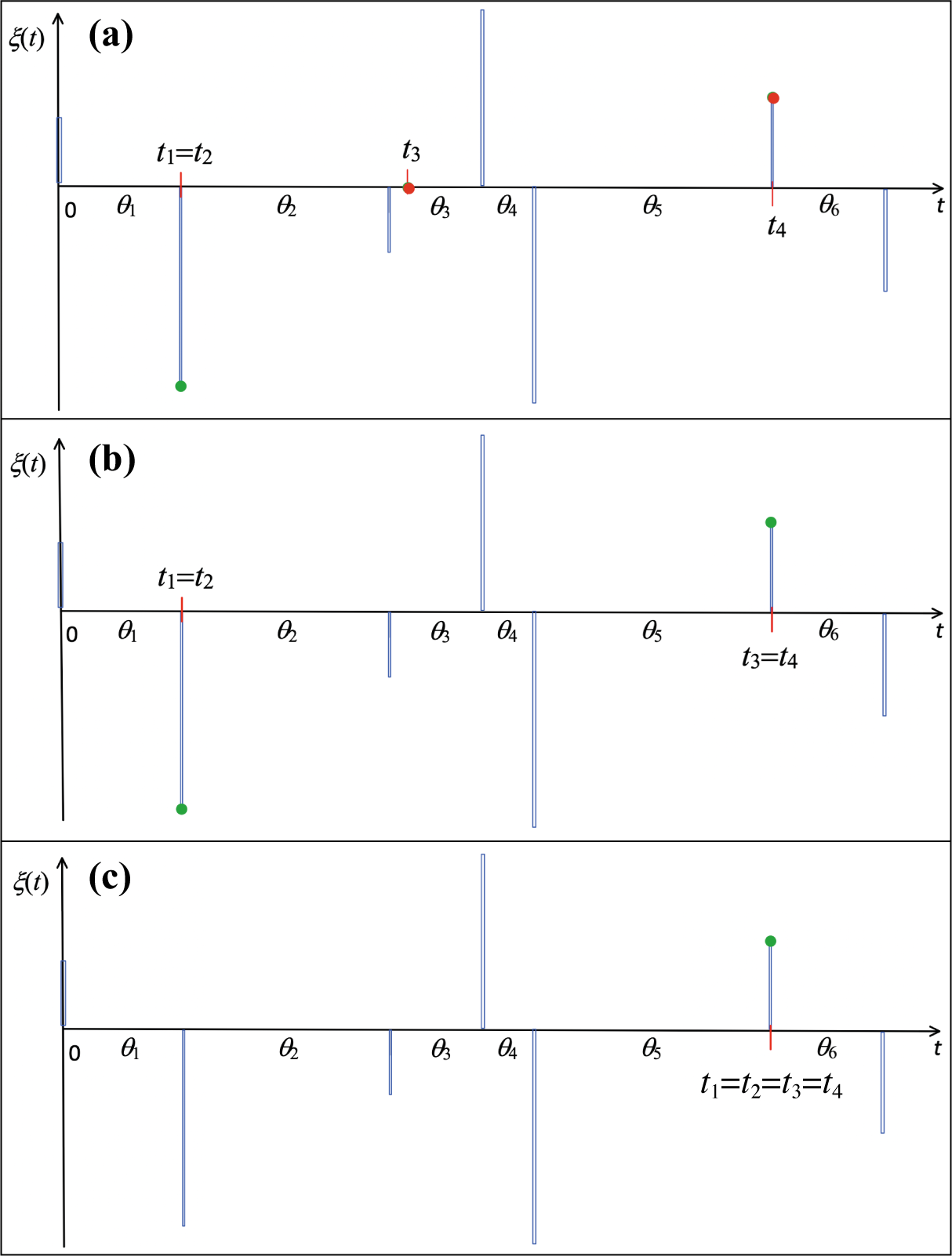}
\caption{A trajectory realization for the spike stochastic renewal process. The product $\xi( t_1)\xi( t_2)\xi( t_3)\xi( t_4)$ is zero in the case of the panel (a), where at the time
$t_3$ there is not an event. After averaging over the trajectories, assuming the odd moments  of $\xi$
are zero, the only non vanishing terms in the sum  of Eq.~\eqref{n_corr_def_2_delta}, are those
for which $ t_1, t_2...t_n$ are equal in pairs, for example, 
$ t_1 = t_2 $  
and $ t_3 = t_4 $ for $n=4$, as in panel (b) in figure, or in blocks containing an even number of times,
for example $ t_1 = t_2 = t_3 = t_4 $ for $n=4$, as  in panel (c) in figure. 
}
\label{fig:trajectory_leapers_corr_1}
\end{figure}

It is important to emphasize that  
Eq.~\eqref{n_corr_def_2_delta} is exact,  following directly from the definition of the 
multi-time joint correlation function of the  stochastic renewal process  $\xi[t]$.

The reader should appreciate that, due to the nonanalytic nature of the Dirac delta
function, caution must be exercised when interpreting
Eq.~\eqref{n_corr_def_2_delta}. Indeed, consider the two-point correlation function
evaluated at times \( t_1 \) and \( t_2 \ge t_1 \). This can be written as follows (for
simplicity we set \( t_0 = 0 \); in this case we omit the subscript ``\( t_0 \)'' in the
angle brackets):
\begin{align}
\label{2_corr_def_2_delta}
&  \langle \xi(t_1)\xi(t_2) \rangle
\nonumber \\
&= \int \sum_{i_1=0}^{\infty} \xi_{i_1} \;
\delta\!\left(t_1-\sum_{k_1=0}^{i_1} \theta_{k_1}\right)
\sum_{i_2=i_1}^{\infty} \xi_{i_2} \;
\delta\!\left(t_2-\sum_{k_2=0}^{i_1} \theta_{k_2}
- \xsum_{k_2=i_1+1}^{i_2} \theta_{k_2}\right) \times
\nonumber \\ &
p_0(\xi_0)\,\text{d}\xi_0 \times
\prod_{k_1=0}^{i_1} \psi(\theta_{k_1})\,\text{d}\theta_{k_1}
p(\xi_{k_1})\,\text{d}\xi_{k_1}
\xprod_{k_2=i_1+1}^{i_2} \psi(\theta_{k_2})\,\text{d}\theta_{k_2}
p(\xi_{k_2})\,\text{d}\xi_{k_2} \; .
\end{align}
Here we adopt the convention that a primed summation or product is equal to zero when
the upper limit is smaller than the lower one. This situation arises when
\( t_2 = t_1 \). We now consider the case \( t_2 = t_1 + \epsilon \), with
\( \epsilon \to 0 \).
In this case,
however small \( \epsilon \) is, the WT PDF and the PDF of \( \xi \) in the primed
product never vanish, whereas, as we have observed, they should be exactly zero when
\( \epsilon = 0 \). Thus, \( t_2 = t_1 \) is a point of discontinuity. This issue will be
addressed again later, where it will be examined in more detail.

Although Eq.~\eqref{n_corr_def_2_delta} may appear complex, it can be handled straightforwardly. 
For the case $n = 2$, assuming,  for the sake of simplicity, that $\overline{\xi}=0$, the result was derived in~\cite{bblmCSF196}:
\begin{align}
\label{cor2t_delta_pre}
\langle \xi(t_1) \xi(t_2) \rangle_{t_0} &= 
\overline{\xi_0^2} \,\delta(t_1 - t_0)\delta(t_2 - t_1) +
\overline{\xi^2} R(t_1 - t_0) \delta(t_2 - t_1).
\end{align}
This expression can be interpreted probabilistically: if $\xi(t_1)$ or 
$\xi(t_2)$ do not coincide with a spike (i.e., a transition event), their 
values are zero. If $t_1$ and $t_2$ correspond to distinct spikes, the average 
vanishes due to the independence of the $\xi_k$ and the assumption 
$\bar{\xi} = 0$. 

Therefore, the correlation is nonzero only when $t_2 = t_1$, as indicated by the 
Dirac delta $\delta(t_2 - t_1)$, and it is proportional to the probability that 
a transition event occurs at  that common time, after $t_0$, as captured by the 
rate function $R$. If no events have occurred, i.e., $t_1 = t_2 = t_0$, this 
situation is encoded by the term $\delta(t_1 - t_0)$.

    Note that $\xi_0$
contributes only at the initial time $t_0$. Since the spike process considered
  is usually interpreted as the velocity of a generalized CTRW, as in
Eq.~\eqref{SDE}, it would be
unnatural to start the random walk with a jump at $t_0$. Typically, the first
jump occurs only after the first random waiting time. This corresponds to
assuming $\xi_0 = 0$, i.e., $\overline{\xi_0^n} = 0$. Therefore we 
    can take
$\overline{\xi_0^2} = 0$, which, when inserted into
Eq.~\eqref{cor2t_delta_pre}, leads to
\begin{align}
\label{cor2t_delta_}
\langle \xi(t_1) \xi(t_2) \rangle_{t_0} &=
\overline{\xi^2}\,R(t_1 - t_0)\,\delta(t_2 - t_1).
\end{align}

Importantly, Eqs.~\eqref{cor2t_delta_pre}-\eqref{cor2t_delta_} depend only on the second moment 
$\overline{\xi^2}$ of $p(\xi)$ 
and is unaffected by its specific form. Whether $\xi$ is drawn from a dichotomous, Gaussian, or 
power-law distribution, the result remains valid.
For this reason, the result was referred to as universal in~\cite{bblmCSF196}.

No assumptions are made on the WT distribution $\psi(\theta)$; in particular, $\psi(t)$ may 
exhibit a power-law decay, $\psi(t) \sim  T^{-1} (T/t)^{\mu}$ with $\mu > 1$ and $T$ is the scaling time. 
Because of the dependence on $t_1-t_0$, the process is generally non-stationary. 
In fact, although the Dirac-delta functions
in \eqref{cor2t_delta_pre} make the stochastic process $\xi[t]$
to resemble white 
noise, the rate $R(t_1 - t_0)$ introduces an age-dependent modulation.

Following~\cite{bJSTAT2020} and \cite{mssbgCSF188}, if $\psi(t) \sim  T^{-1}(T/t)^{\mu}$, then:
\begin{enumerate}[label=\Roman*]
\item\label{I_I_} For $1 < \mu < 2$: $R(t) \sim T^{-1} (T/t)^{2 - \mu}$,
\item\label{II_I_}  For $\mu = 2$: $R(t) \sim T^{-1} / \log(t/T)$,
\item\label{III_I_} For $\mu > 2$: $R(t) \sim  \tau^{-1}\left[ 1 + (T/t)^{\mu - 2}\right]$;
\end{enumerate}
where the mean time  $\tau$ is defined in \eqref{tau}.
Thus, for $1 < \mu \le 2$, the shot noise intensity decreases with age, while for $\mu > 2$, 
$\xi[t]$ behaves like stationary white noise in the long-time limit~\cite{bblmCSF196}.

We now extend our analysis beyond the $n = 2$ case. While a rigorous derivation is provided in Appendix~\ref{app:spikes}, we shall continue to utilize the statistical interpretation of Eq.~\eqref{n_corr_def_2_delta}.
The presence of Dirac delta functions in this expression implies that, for any given trajectory realization, the product $\xi(t_1)\xi(t_2)\dots\xi(t_n)$ vanishes unless every time $t_i$ coincides with a transition event. Consequently, the sum in Eq.~\eqref{n_corr_def_2_delta} reduces to a summation over all possible partitions of the set $\{t_1, \dots, t_n\}$ into $p$ blocks of coincident times, where $1 \le p \le n$. 
Each partition into $p$ blocks represents the joint probability that a transition event occurs at each of those $p$ distinct timestamps. Specifically, a block of size $m_i$ (where $\sum_{i=1}^p m_i = n$) contributes a factor $\overline{\xi^{m_i}}$, representing the $m_i$-th moment of $\xi$.
To illustrate this reasoning, consider the four-time correlation function. For simplicity, we shall maintain the assumption that all odd moments vanish. We have
\begin{align}
\label{cor4t_delta}
&\langle\xi( t_1)\xi( t_2)\xi( t_3)\xi( t_4)\rangle_{t_0} \nonumber
\\[4pt]
&=
\left\{
\begin{array}{ll}
\left[ \overline{\xi_0^4} \,\delta(t_1-t_0)+ \overline{\xi^4}\, R(t_1-t_0)\right]  
\\
\times  \delta(t_2-t_1) \delta(t_3-t_2) \delta(t_4-t_3)+&(a)
\\
\, & \,\\
+\overline{\xi^2}\left[\overline{\xi_0^2}\,\delta(t_1-t_0)
+\overline{\xi^2}\,R(t_1-t_0) \right] \delta(t_2-t_1) 
\\
\times R(t_3-t_2) \delta(t_4-t_3)&(b)
\end{array} 
\right.
%
%
%
%
\end{align}

It is apparent from Eq.~\eqref{cor4t_delta} that for the four-time correlation function to remain non-zero, only two configurations are possible: either all four time points coincide at the same transition event (case $(a)$ in Eq.~\eqref{cor4t_delta}, illustrated in Fig.~\ref{fig:trajectory_leapers_corr_1}c), or the first two times coincide at one transition event while the remaining two coincide at a subsequent transition event (case $(b)$, illustrated in Fig.~\ref{fig:trajectory_leapers_corr_1}b).

In a similar manner to the two-time joint correlation function, it should be noted that the expression in Eq.~\eqref{cor4t_delta} depends exclusively on the second and fourth moments of $\xi$, rather than on any higher-order moments.\\

For a fixed $t_0$, the process is generally non-stationary, as it depends on $t_1$ through the rate function $R(t_1 - t_0)$. However, as outlined in case \ref{III_I_} above, for $\mu > 2$ (i.e., when the average waiting time $\tau$ is finite) and for $t_1 - t_0 \gg \tau$, the stochastic renewal spike process becomes stationary.

Conversely, for $1 < \mu < 2$, the rate function behaves as $R(t) \sim T^{-1} (T/t)^{2 - \mu}$ (refer to case \ref{I_I_} in the list above)~\cite{bJSTAT2020}, thus, 
stationarity is never achieved. We also note that, in this situation, in the limit of large $t_1$ and $t_3 - t_2 \gg \tau$,  in Eq.~\eqref{cor4t_delta} the  term $(a)$ becomes dominant. 
Specifically, this term is \textit{proportional to the two-time joint correlation function evaluated at the two temporal extremes} (i.e., the initial and final time points):
\begin{align}
\label{cor4t_delta_power_slow}
&\langle\xi( t_1)\xi( t_2)\xi( t_3)\xi( t_4)\rangle_{t_0}
\stackrel{t_1\gg T}{\Rightarrow} \,  
\frac{\overline{\xi^4}}{\overline{\xi^2}}\langle\xi( t_1)\xi( t_4)\rangle_{t_0}
\delta(t_2-t_1)\delta(t_3-t_2)
\nonumber\\
&
=\overline{\xi^4}\,
R(t_1-t_0)\delta(t_4-t_1)\delta(t_2-t_1)\delta(t_3-t_2)
\nonumber\\
&
\sim\overline{\xi^4}\, T^{-1} \left(\frac{T}{t_1-t_0}\right)^{2 - \mu}\delta(t_4-t_1)\delta(t_2-t_1)\delta(t_3-t_2).
\end{align}

It is noteworthy that, starting from Eq.~\eqref{cor2t_delta_pre}, 
straightforward algebraic manipulation yields
\begin{align}
\label{cor4t_delta_bis}
&\langle\xi( t_1)\xi( t_2)\xi( t_3)\xi( t_4)\rangle_{t_0} 
=\langle\xi( t_1)\xi( t_2)\rangle_{t_0}\langle\xi( t_3)\xi( t_4)\rangle_{t_2}\nonumber \\&
+\langle\xi( t_1)\xi( t_4)\rangle_{t_0}\delta(t_2-t_1)\delta(t_3-t_2)\delta(t_4-t_3)
\,
\left(\frac{\overline{\xi^4}}{\overline{\xi^2}}-\overline{\xi^2}\right).
\end{align}
Consequently, if $\xi$ is a dichotomous random variable, the second term on the right-hand side of Eq.~\eqref{cor4t_delta} vanishes. The expression then reduces to the product of two separate two-time joint correlation functions, where the initial time of the second function corresponds to the final time of the first. That is:
\begin{align}
\label{cor4t_delta_dicho}
&\langle\xi( t_1)\xi( t_2)\xi( t_3)\xi( t_4)\rangle_{t_0}
=\langle\xi( t_1)\xi( t_2)\rangle_{t_0}\langle\xi( t_3)\xi( t_4)\rangle_{t_2}
\nonumber\\
&
=\left(\overline{\xi^2}\right)^2 \tilde R(t_1-t_0) \delta(t_2-t_1)\tilde R(t_3-t_2)\delta(t_4-t_3)
\end{align}
This represents a generalization to the non-Poissonian regime of the factorization property inherent to the four-time correlation function of a Poissonian dichotomous stochastic process (also known as telegraph noise). This property will appear more generally in the context of $n$-time correlation functions.


\section{The exact formal expressions for the $n$-time correlation function\label{sec:corr_leapers}}

As previously noted, Appendix~\ref{app:spikes} provides a rigorous derivation 
of the general multi-time correlation function, starting from the definition 
established in Eq.~\eqref{n_corr_def_2_delta}. \\  

Here, we present the same result using a more intuitive, albeit less formal, 
approach; specifically, we leverage the statistical arguments previously 
employed to interpret the $n=2$ and $n=4$ cases.

\begin{proposition}
\label{prop:procedure_fligh}
Let us consider $n$ ordered times $t_1 \le t_2 \le \dots \le t_n$. The
$n$-time joint correlation function for the stochastic process defined
as random spikes with renewal (the noise driving the L\'evy flights or
CTRW) can be obtained through the following procedure:
\end{proposition}
\begin{enumerate}[label=(\roman*)]
\item \label{itema}
Consider a sequence of $n$ ordered times $t_1, t_2, \dots, t_n$, 
delimited by a bar ``$|$'' before the first time and another after the last.
\item \label{itemb}
Consider any composition of this sequence (i.e., a partition where the order matters), 
consisting of $p \le n$ subsequences (or blocks) separated by bars. 
Specifically, we partition the ordered sequence into blocks 
$| \{ m_1 \} | \{ m_2 \} | \dots | \{ m_p \} |$, where each $m_i$ denotes the 
(even) number of elements in the $i$-th block. 
Naturally, these must satisfy $\sum_{i=1}^p m_i = n$. 
For example, a partition into $p$ blocks might be
\begin{equation}
\label{compositions_}
\vert \underbrace{t_1\,t_2\,t_3\,t_4\,t_5\,t_6}_{m_1=6}\bm{\mid}
\underbrace{t_7\,t_8}_{m_2=2}\vert 
\underbrace{t_9\,t_{10}\,t_{11}}_{m_3=3}\vert \,...\,\vert 
\underbrace{t_{n-3}\,t_{n-2}\,t_{n-1}\,\,t_n}_{m_p=4}\vert 
\end{equation}
The number of such possible compositions is $2^{n-1}$, as a bar can 
be either ``on'' or ``off'' at each even position (excluding the two endpoints).
\item \label{itemc} For each block of $m_i$ times:
\begin{itemize}
\item   supply $m_i-1$ Dirac delta functions to enforce the coincidence of all time variables within the block;
\item   retain only one time variable (e.g., the first or the last) and discard the others, as they are rendered redundant by the delta functions.
Thus, for example, from Eq.~\eqref{compositions_} we obtain:
\begin{align}
\label{corrGen_delta_temp2_}
&
\bm{\delta}_1(\bm{\Delta t}_{m_1})
\,\bm{\delta}_2(\bm{\Delta t}_{m_2})
\,\bm{\delta}_3(\bm{\Delta t}_{m_3})
...\bm{\delta}_p(\bm{\Delta t}_{p})
\nonumber\\&
\times \vert t_{m_1}\vert t_{m_1+m_2}\vert t_{m_1+m_2+m_3}\vert ...
\vert t_n\vert .
\end{align}
where, for $m_i=1$ (only one element/time in the block)
\begin{equation}
\label{def_delta1}
\bm{\delta}_i(\bm{\Delta t}_{m_i}):=1
\end{equation}
and, for $m_i>1$
\begin{align}
\label{def_delta}
&   \bm{\delta}_1(\bm{\Delta t}_{m_1}):=              \delta(t_{2}-t_{1})     \delta(t_{3}-t_{2})...  \delta(t_{m_1}-t_{m_1-1})\nonumber \\
&   \bm{\delta}_2(\bm{\Delta t}_{m_2}):=              \delta(t_{m_1+2}-t_{m_1+1})...\delta(t_{m_1+m_2}-t_{m_1+m_2-1})
\nonumber \\
&   \bm{\delta}_3(\bm{\Delta t}_{m_3}):=              \delta(t_{m_1+m_2+2}-t_{m_1+m_2+1})
...\delta(t_{m_1+m_2+m_3}-t_{m_1+m_2+m_3-1})
\nonumber \\
&...\nonumber \\
&   \bm{\delta}_i(\bm{\Delta t}_{m_i}):=      \delta(t_{m_1+...+m_{i-1}+2}-t_{m_1+...+m_{i-1}+1})...\nonumber \\
&...\delta(t_{m_1+...+m_{i-1}+m_i}-t_{m_1+...+m_{i-1}+m_i-1}).
\end{align}
Assuming time-ordering, with $t_{i+1} \geq t_i$, $\bm{\delta}_i(\bm{\Delta t}_{m_i})$ 
can also be interpreted as a single Dirac delta function whose argument is the 
difference between the endpoints of the block (except for the case $m_i=1$).
\item 
In the transformed partition
\begin{equation*}
    \vert t_{m_1}\vert t_{m_2}\vert t_{m_3}\vert \dots \vert t_{m_p}\vert
\end{equation*}
replace the first block $\vert t_{m_1} \vert$ with 
$\overline{\xi_0^{m_1}}\delta(t_{m_1}-t_0) + \overline{\xi^{m_1}}R(t_{m_1}-t_0)$. 
For $i > 1$, replace each $i$-th block with 
$\overline{\xi^{m_i}}R(t_{m_i} - t_{m_{i-1}})$. 
Specifically, with the exception of the first block (which accounts for the possibility that no transition occurs before $t_1$), each $i$-th block is replaced by the rate function $R$ evaluated at the time difference between the current and preceding blocks.
\end{itemize}
Thus, at the end, the
example partition in \eqref{compositions_} becomes 
\begin{align}
&
\left( \overline{\xi_0^{m_1}}\,\delta(t_{m_1}-t_0)+\overline{\xi^{m_1}}
R(t_{m_1}-t_0)\right)\, \bm{\delta}_1(\bm{\Delta t}_{m_1})
\nonumber \\&
\times\overline{\xi^{m_2}}R(t_{m_1+m_2}-t_{m_1})\,\bm{\delta}_2(\bm{\Delta t}_{m_2})
\;\overline{\xi^{m_3}}R(t_{m_1+m_2+m_3}-t_{m_1+m_2})\,\bm{\delta}_3(\bm{\Delta t}_{m_3})
\nonumber \\&
\times...\times \overline{\xi^{m_p}}
R(t_{n}-t_{_{n-m_{p}}})\,\bm{\delta}_p(\bm{\Delta t}_{p})
\end{align}
\item \label{iteme} $\langle\xi(t_1)\xi(t_2)\dots\xi(t_n)\rangle$ is obtained by summing all   $2^{n-1}$ compositions 
so worked. A way to do that is by first summing all those corresponding to
a fixed number $p$ of blocks  (they are 
$N(p)=\frac{\left(n-1\right)!}{(p-1) !\left[n-p\right] !}$) and then summing
for all $p=1, 2,...n$ (of course, $\sum_{p=1}^{n} N(p)=2^{n-1}$):
\begin{align}
\label{corrGen_delta}
&\langle\xi(t_1)\xi(t_2)...\xi(t_n)\rangle_{t_0}
\nonumber \\&
=
\sum_{p=1}^{ n} \bigg[
\sum_{\{m_i\}:\sum_{i=1}^p m_i={ n}}  
\left( \overline{\xi_0^{m_1}}\,\delta(t_{m_1}-t_0)+\overline{\xi^{m_1}}
R(t_{m_1}-t_0)\right)\, \bm{\delta}_1(\bm{\Delta t}_{m_1})
\nonumber \\&
\times\overline{\xi^{m_2}}R(t_{m_1+m_2}-t_{m_1})\,\bm{\delta}_2(\bm{\Delta t}_{m_2})
\;\overline{\xi^{m_3}}R(t_{m_1+m_2+m_3}-t_{m_1+m_2})\,\bm{\delta}_3(\bm{\Delta t}_{m_3})
\nonumber \\&
\times...\times \overline{\xi^{m_p}}
R(t_{n}-t_{_{n-m_{p}}})\,\bm{\delta}_p(\bm{\Delta t}_{p})
\bigg].
\end{align}
\end{enumerate}

We reiterate that in practical applications where the spike stochastic
renewal process serves as the velocity of a generalized CTRW, as in
Eq.~\eqref{SDE}, it is customary to set $\xi_0=0$ (i.e.,
$\overline{\xi_0^n} = 0$), thereby eliminating the first Dirac delta
function in Eq.~\eqref{corrGen_delta}.

Equation~\eqref{corrGen_delta} can then be conveniently rewritten as a
sum over compositions of $n$ distinguishable objects, as in
Eq.~\eqref{corrGen_delta_pre}.

In the simplified case where the odd moments of $\xi$ vanish, and by 
exploiting the result in Eq.~\eqref{cor2t_delta_pre}, the expression in 
Eq.~\eqref{corrGen_delta} can be rewritten in terms of the two-time 
correlation function as
\begin{align}
\label{corrGen_delta_2}
&\langle\xi(t_1)\xi(t_2)...\xi(t_n)\rangle_{t_0}
=
\sum_{p=1}^{ n/2} \bigg[
\sum_{\{m_i\}:\sum_{i=1}^p m_i={ n}}  
\frac{\overline{\xi^{m_1}}\,\overline{\xi^{m_2}}...\overline{\xi^{m_p}}}
{(\overline{\xi^{2}})^p}
\left\langle\xi(t_1)\xi(t_{m_1}\right\rangle_{t_0}
\times\nonumber \\&\times 
\left(
\left\langle\xi(t_{m_1+1})\xi(t_{m_1+m_2})\right\rangle_{t_{m_1}}
-\overline{\xi^{2}}\delta(t_{m_1+1}-t_{m_1})
\right)\times
...
\nonumber \\
&...\times 
\left(\left\langle\xi(t_{n-m_{p}+1})\xi(t_{n})\right\rangle_{t_{_{n-m_{p}}}}
-\overline{\xi^{2}}\delta(t_{n-m_{p}+1}-t_{n-m_{p}})\right) 
\bigg].
%
\end{align}

It is important to emphasize that the spike-like nature of the noise, 
reflected by the Dirac delta functions in Eq.~\eqref{corrGen_delta}, 
imposes a stringent constraint on which terms survive the summation 
over compositions in Eq.~\eqref{corrGen_delta} for a fixed set of 
times $\{t_1, t_2, \dots, t_n\}$. Specifically, only the term 
(or composition) consistent with both the chosen time values and 
the Dirac delta constraints contributes to the correlation function. 
For example, if all $n$ times are distinct, none of the Dirac delta 
constraints are satisfied except for the minimal partition, and the 
sum in Eq.~\eqref{corrGen_delta} reduces to the term corresponding 
to $p = n$. At the opposite extreme, if all times are identical, 
only the $p = 1$ term survives, corresponding to the full 
contraction of all times into a single block. Nonetheless, by 
invoking the same detailed argument  developed at the end of
Appendix~\ref{app:4-time}, we can 
always safely retain all terms on the r.h.s. of Eq.~\eqref{corrGen_delta} 
without affecting the final result.

The procedure outlined in Proposition~\ref{prop:procedure_fligh} may initially seem intricate, but
it is, in fact, straightforward. To illustrate, we provide a couple of examples (one in Appendix~\ref{app:nequals8case}) 
        where, for the sake of simplicity and ``readability'', we assume that the odd moments of $\xi$ are zero, that $t_0=0$, $t_1>t_0$ and/or  $\overline{\xi_0^2} = 0$. 

Consider the case $ n = 6 $. The possible compositions
of the six times, made of blocks with even elements, are as follows: 
\begin{align*}
&\vert t_1\,t_2\,t_3\,t_4\,t_5\,t_6\vert ,&(p=1)\\
&\vert t_1\,t_2\,t_3\,t_4\vert \,t_5\,t_6\vert \text{ and }  \vert t_1\,t_2\vert \,t_3\,t_4\,t_5\,t_6\vert  &(p=2)\\
&\vert t_1\,t_2\vert \,t_3\,t_4\vert \,t_5\,t_6\vert , &(p=3)
\end{align*}
Then, following steps  \ref{itema}-\ref{iteme}, we obtain
\begin{align}
\label{corr6_}
\begin{array}{ll}
\langle\xi(t_1)\xi(t_2)\xi(t_3)\xi(t_4)\xi(t_5)\xi(t_6)\rangle&  \\
& \\
=\overline{\xi^{6}}\left[\delta(t_2-t_1)\delta(t_3-t_2)\delta(t_4-t_3)\delta(t_5-t_4)\right]
R(t_1)\delta(t_6-t_1) &  (p=1) \\
& \\
+\overline{\xi^{4}}\,\overline{\xi^{2}}\bigg\{ \left[\delta(t_2-t_1)\delta(t_3-t_2)\delta(t_4-t_3)\right] R(t_1)\delta(t_4-t_1)R(t_5-t_4)\delta(t_6-t_5) \\
+ 
\left[\delta(t_4-t_3)\delta(t_5-t_4)\delta(t_6-t_5)\right]
R(t_1)\delta(t_2-t_1)R(t_3-t_2)\delta(t_6-t_3)\bigg\} & (p=2) \\
& \\
+(\overline{\xi^{2}})^3 R(t_1)\delta(t_2-t_1) R(t_3-t_2)\delta(t_4-t_3) R(t_5-t_4)\delta(t_6-t_5) & (p=3)
\end{array}
\end{align}
that, by employing \eqref{cor2t_delta_pre}, and considering $t_1>t_0$, can also be written as in \eqref{corrGen_delta_2}:
%

\begin{align*}
& \langle\xi(t_1)\xi(t_2)\xi(t_3)\xi(t_4)\xi(t_5)\xi(t_6)\rangle
\nonumber \\
&=\frac{\overline{\xi^{6}}}{\overline{\xi^{2}}}\langle \xi(t_1)\xi(t_6)\rangle
\delta(t_3-t_2)\delta(t_4-t_3)\delta(t_5-t_4)\delta(t_6-t_5)
\nonumber \\
&+\frac{\overline{\xi^{4}}}{\overline{\xi^{2}}}\bigg[\langle \xi(t_1)\xi(t_4)\rangle \delta(t_3-t_2)\delta(t_4-t_3)
\left(\langle \xi(t_5)\xi(t_6)\rangle_{t_4}-\overline{\xi^{2}}\delta(t_5-t_4)\delta(t_6-t_5)\right)
\nonumber \\
&+\langle \xi(t_1)\xi(t_2)\rangle \delta(t_4-t_3)\delta(t_5-t_4)\delta(t_6-t_5)
\left(\langle \xi(t_3)\xi(t_6)\rangle_{t_2}-\overline{\xi^{2}}\delta(t_3-t_2)\delta(t_6-t_3)\right)
\bigg]
\nonumber \\[3pt]
&+\langle \xi(t_1)\xi(t_2)\rangle
\left(
\langle \xi(t_3)\xi(t_4)\rangle_{t_2}-\overline{\xi^{2}}\delta(t_3-t_2)\delta(t_4-t_3)\right)
\nonumber \\
&\left(
\langle \xi(t_5)\xi(t_6)\rangle_{t_4}-\overline{\xi^{2}}\delta(t_5-t_4)\delta(t_6-t_5)\right)
\end{align*}
With a little algebra, from the previous expression we have
\begin{align}
\label{cor6t_delta_4}
& \langle\xi(t_1)\xi(t_2)\xi(t_3)\xi(t_4)\xi(t_5)\xi(t_6)\rangle  
\nonumber \\
&=\langle \xi(t_1)\xi(t_6)\rangle \left(\frac{\overline{\xi^{6}}}{\overline{\xi^{2}}}
+\left(\overline{\xi^{2}}\right)^2-2\,\overline{\xi^{4}}\right)
\delta(t_3-t_2)\delta(t_4-t_3)\delta(t_5-t_4)\delta(t_6-t_5)
\nonumber \\
&+\left(\frac{\overline{\xi^{4}}}{\overline{\xi^{2}}}-\overline{\xi^{2}}\right)
\bigg[\langle \xi(t_1)\xi(t_4)\rangle \langle \xi(t_5)\xi(t_6)\rangle_{t_4}
\delta(t_3-t_2)\delta(t_4-t_3)
\nonumber \\
&+\langle \xi(t_1)\xi(t_2)\rangle \langle \xi(t_3)\xi(t_6)\rangle_{t_2}
\delta(t_5-t_4)\delta(t_6-t_5)
\bigg]
\nonumber \\
&+
\langle \xi(t_1)\xi(t_2)\rangle \langle \xi(t_3)\xi(t_4)\rangle_{t_2}
\langle \xi(t_5)\xi(t_6)\rangle_{t_4}.
\end{align}

All the considerations discussed for the $n=4$ case also apply here. 
In particular, we observe that, as for $n=4$, if $\xi$ is a 
dichotomous random variable, the only term that remains is the one 
in the last line of Eq.~\eqref{cor6t_delta_4}, resulting in a 
generalized (non-stationary) "factorization property".
Appendix~\ref{app:nequals8case} shows how to derive the eight-time joint correlation function.

As a further check of the key result established in Proposition~\eqref{prop:procedure_fligh} 
and the corresponding Eq.~\eqref{corrGen_delta}, in 
Appendix~\ref{app:Montroll_Weiss} 
we apply the Laplace transform technique to this equation to recover the 
standard Montroll-Weiss and Scher result~\cite{Montroll_Weiss_1965,Scher_Motroll_1975} 
for the Fourier-Laplace transform of the PDF of the CTRW defined by $\dot{x}= \xi[t]$.

As previously noted, the expression in Eq.~\eqref{corrGen_delta}---and 
particularly its formal representation in Eq.~\eqref{corrGen_delta_pre}---closely 
resembles the general relationship between correlation functions and 
cumulants. Furthermore, the fact that the partitions in Eq.~\eqref{corrGen_delta} 
preserve the ordering of elements to form compositions highlights the 
specific relevance of $G$-cumulants in this context 
(see~\cite[Section~4.4.3, Eq.~(94)]{bbJSTAT4}). 
As discussed in~\citep{bblmCSF202}, this observation, combined with the results 
established in~\cite{bbJSTAT4,bCSF148,bCSF159}, provides a robust 
framework for addressing cases where the stochastic renewal process 
acts as a noise source driving a variable subject to an external 
field, as in Eq.~\eqref{SDE}. Indeed, this approach enables us to 
derive the general exact ME presented in Eq.~\eqref{ME_fin}, along 
with its simplified form in Eq.~\eqref{ME_universal_intro}, in the 
following sections.

\section{The Master Equation for the variable of interest $x$ driven by a drift and a multiplicative shot noise\label{sec:MEx}}
In this section we present the derivation of the exact ME for the variable $x$ governed by the
SDE~\eqref{SDE}, introduced in Eq.~\eqref{ME_fin}.

To obtain this result, we make extensive use of the definition of
$M$-cumulants, where $M$ denotes the mapping applied to the exponential 
function that relates the characteristic function to the cumulant generator. 
Specifically, we adopt the choice $M=G$, which is further specified and 
motivated below. While the definition and properties of $M$- and 
$G$-cumulants originate from the early works of 
Kubo~\cite{kuboGenCumJPSJ17,kuboGenCumJMP4}, we refer the reader to 
Refs.~\cite{bbJSTAT4,bCSF148,bCSF159} for a more systematic 
presentation and further results. To keep the paper self-contained 
and to provide a streamlined introduction tailored to our purposes, 
in Appendix~\ref{app:M-cumulants} we briefly summarize the basic concepts 
of $M$-cumulants that are relevant to the present analysis.

\subsection{$G$-cumulants for  the CTRW $x(t)$\label{sec:Gcumulants_CTRW}}
To introduce the procedure we will adopt later more simply, we begin 
with the well-known case of the standard CTRW, where 
$x(t)=\int_0^t \xi(u)\text{d}u+x_0$. For any fixed time $t$, $x(t)$ 
can be treated as a univariate random variable rather than a 
stochastic process. Its characteristic function (CF) is 
therefore given by
\begin{align}
    \label{CF_CTRW}
    \hat{P}_{CTRW}(k; t) &:= \langle \exp[\text{i}k\,x(t)]\rangle \hat{P}(k; 0) \nonumber \\
    &= \langle \exp\left[\text{i}k\int_{0}^t \xi(u) \, \text{d}u\right]\rangle \hat{P}(k; 0).
\end{align}
The above equation clearly shows that the CF of the \emph{univariate} CTRW at any fixed 
time $t$ is simply the CF of the \emph{stochastic} process $\xi[t]$ 
evaluated at a constant wavenumber, multiplied by the CF of the 
initial PDF of $x$ (for a more detailed discussion, see \cite{bCSF148}).

Thus, by expanding the exponential function in a power series, we can express 
the CF of $x(t)$ in terms of the correlation functions of the 
process $\xi[t]$, which we denote as $\hat{P}_{\xi[t]}(k, t)$ 
in \ref{sec:Gcumulants_xi} (see Eq.~\eqref{MG}):
\begin{align}
    \label{CF_CTRW_xi}
    \hat{P}_{CTRW}(k; t) &= \sum_{n=0}^\infty (\text{i} k)^n 
    \int^t_0{\text{d}u_n}\int^{u_{n}}_0{\text{d}u_{n-1}}\dots 
    \int^{u_{2}}_0\text{d}u_1\, 
    \langle \xi (u_1)\dots \xi (u_n)\rangle \hat{P}(k; 0),
\end{align}
from which it follows that
\begin{align}
    \langle x^n \rangle &= \left.\frac{1}{\text{i}^n}\frac{\partial^n}{\partial k^n}
    \hat{P}_{CTRW}(k; t)\right|_{k=0} \nonumber \\
    &= \int^t_0{\text{d}u_n}\int^{u_{n}}_0{\text{d}u_{n-1}}\dots 
    \int^{u_{2}}_0\text{d}u_1 \langle \xi (u_1)\dots \xi (u_n)\rangle 
    + \left.\frac{1}{\text{i}^n}\frac{\partial^n}{\partial k^n}\hat{P}(k; 0)\right|_{k=0}.
\end{align}

We denote with ${K }^{(G)}_{x}(k\pv t)$ the $G$-cumulant generator for the univariate random process $x(t)$~\cite{bbJSTAT4,bCSF148,bCSF159}.  As for the CF, 
        also ${K }^{(G)}_{x}(k\pv t)$ can be expressed in terms of $G$-cumulants of
        the stochastic process $\xi[t]$ (see also Appendix~\ref{app:M-cumulants}):
\begin{align}
\label{CGx}
&      
{K }^{(G)}_{x}(k\pv t)
:=\sum_{n=1}^\infty (\I k)^n 
\int^t_0{\text{d}u_n}\int^{u_{n}}_0{\text{d}u_{n-1}}\dots 
\int^{u_{2}}_0\text{d}u_1
\nonumber \\
&
\doubleangle{ \xi (u_1)\xi (u_2)\dots \xi (u_n)}^{(G)}.%
\end{align}
The relationship between (and the definition of) $G$-cumulants and correlation functions is given by the following 
identity:
\begin{equation}
\label{K_CTRW}
\hat{{P }}_{CTRW}(k\pv t)=  \exp_G\!\left[{K }^{(G)}_{x}(k\pv t)\right]
\end{equation}
where the symbol $G$ (or, more generally, $M$) applied to a
function or operator (here, the exponential function) indicates
that the $G$ map (or the generic map $M$) acts on that object.
The $G$ map corresponds to the \emph{total time ordering} (TTO) map
(see again Appendix~\ref{app:M-cumulants} for details).
        Expanding both sides of Eq.~\eqref{K_CTRW} in powers of $k$ and equating
coefficients term by term, we obtain the corresponding relations between
moments and  $G$-cumulants:
\begin{align}
	\label{mukappaCombiMulti}
	&\int^t_0{\text{d}u_n}\int^{u_{n}}_0{\text{d}u_{n-1}}\dots 
	\int^{u_{2}}_0\text{d}u_1
	\langle\xi(u_{1}),\xi(u_{2}),...,\xi(u_{n})\rangle		\nonumber \\
	=&\int^t_0{\text{d}u_n}\int^{u_{n}}_0{\text{d}u_{n-1}}\dots 
	\int^{u_{2}}_0\text{d}u_1\sum_{\pi(n)}  \prod_{B\in\pi(n)} \doubleangle{   \xi(u_{1}),\xi(u_{2}),...,\xi(u_{n})}^{(G)}
\end{align}
where $\pi(n)$ runs through the list of al set-composition (or ordered partitions in groups) of $n$ distinguishable objects (the $n$ times) and  $B$ runs through the list of all blocks of the composition $\pi(n)$.

We look for a ME of the type:
\begin{equation}
\label{MEGx}
\partial_t      \hat{P}_{CTRW}(k\pv t)=\int_0^t \text{d}u\,G(k\pv t,u)
\hat{P}_{CTRW}(k\pv u).
\end{equation} 
Integrating iteratively the ME~\eqref{MEGx} and comparing the result with 
Eq.~\eqref{K_CTRW} we get:
\begin{align}
\label{K_Gx}
&       {K }^{(G)}_{x}(k\pv t)
:=
\int^t_0{\text{d}u}\int^{u}_0{\text{d}u'}G(k\pv u,u').%
\end{align}  
By exploiting the definition ${K }^{(G)}_{x}(k\pv t)$ as the $G$-cumulant generator
as explicated in Eq.~\eqref{CGx}, we obtain the Green function of the ME~\eqref{MEGx} in terms of
        $G$-cumulants of $\xi[t]$:
\begin{align}
\label{GCum_x}
& G(k\pv u,u')
=\sum_{n=1}^\infty (\I k)^n \int_{u'}^{u} \text{d}u_{n-1} \int_{u'}^{u_{n-1}} \text{d}u_{n-2}...\int_{u'}^{u_{3}} \text{d}u_{2}
\nonumber\\
&\times 
\doubleangle{\xi(u') \xi(u_{2}) \dots \xi(u_{n-1})\xi(u)}^{(G)} 
\end{align}

\subsection{From $G$ cumulants to the ME of the CTRW\label{sec:Gcumulants_renewal}} 

We have already observed that Eq.~\eqref{corrGen_delta_pre} closely resembles the general relation between correlation 
functions and $G$-cumulants reported in Eq.~\eqref{mukappaCombiMulti}. In effect, by
using the key result Eq.~\eqref{corrGen_delta_pre} in the l.h.s. of Eq.~\eqref{mukappaCombiMulti} we get:
		\begin{align}
	\label{corr-cum}
	&\int^t_0{\text{d}u_n}\int^{u_{n}}_0{\text{d}u_{n-1}}\dots 
	\int^{u_{2}}_0\text{d}u_1
	\sum_{\pi(n)}
	\prod_{B\in\pi(n)}
	\overline{\xi^{\lvert B\rvert}}\,
	R(u_B - u_{B-1})\,
	\delta(\Delta u_{B})\nonumber \\
	=&\int^t_0{\text{d}u_n}\int^{u_{n}}_0{\text{d}u_{n-1}}\dots 
	\int^{u_{2}}_0\text{d}u_1\sum_{\pi(n)}  \prod_{B\in\pi(n)} \doubleangle{ \xi(u_1) \xi(u_{2}) \dots \xi(u_{n-1})\xi(u_N)}^{(G)}
\end{align}

We note that a direct identification of the integrands on the left-hand 
and right-hand sides is prevented by the temporal coupling between 
consecutive blocks in the former: in each ``block'' of the compositions, 
the function $R$ depends not only on the time associated with that block 
but also on the time of the preceding one. This does \emph{not} occur 
in the product of cumulants on the right-hand side. This issue is 
resolved when $R$ is constant, which occurs only when the WT distribution 
decays exponentially (the Poissonian case).

Nevertheless, if $R$ is not a constant, we can exploit the integration by part
technique applied to the ``multi-integral'' in the  l.h.s. of 
Eq.~\eqref{corr-cum}. Doing that, and rewriting in a more explicit form the sum over compositions, the
l.h.s. of    Eq.~\eqref{corr-cum} becomes

%
\begin{align}
\label{corrGen_delta_CF_}
&\sum_{p=1}^{ n}  k^n
\sum_{\{m_i\}:\sum_{i=1}^p m_i={ n}}  
\int^t_0{\text{d}u_n}\int^{u_{n}}_0{\text{d}u_{n-1}}\dots 
\int^{u_{2}}_0\text{d}u_1\,
\nonumber \\&
\times\overline{\xi^{m_1}}
\left[R^\prime(u_{2}-u_1)      +
R(0)\delta(u_{2}-u_{1})\right]  
\delta(u_{3}-u_{2})...  \delta(u_{m_1}-u_{m_1-1})
\nonumber \\&
\times\overline{\xi^{m_2}}
\left[R^\prime(u_{m_1+2}-u_{m_1+1})+
R(0)\delta(u_{m_1+2}-u_{m_1+1})\right]
\nonumber \\&
\times\delta(u_{m_1+3}-u_{m_1+2})
...\delta(u_{m_1+m_2}-u_{m_1+m_2-1})
\nonumber \\&
\times\overline{\xi^{m_3}}\,
\left[R^\prime(u_{m_1+m_2+2}-u_{m_1+m_2+1})+
R(0)\delta(u_{m_1+m_2+2}-u_{m_1+m_2+1})\right]
\nonumber \\&
\times\delta(u_{m_1+m_2+3}-u_{m_1+m_2+2})
...\delta(u_{m_1+m_2+m_3}-u_{m_1+m_2+m_3-1})\times...
\nonumber \\&
\times \overline{\xi^{m_p}}
\Bigg[R^\prime(u_{m_1+..+m_{p-1}+2}-u_{m_1+...+m_{p-1}+1})
+
R(0)\delta(u_{m_1+...+m_{p-1}+2}-u_{m_1+...+m_{p-1}+1})\Bigg]
\nonumber \\&
\times\delta(u_{m_1+...+m_{p-1}+3}-u_{m_1+...+m_{p-1}+2})
...\delta(u_{m_1+...+m_{p-1}+m_p}-u_{m_1+...+m_{p-1}+m_p-1}).
\end{align}
Comparing this expression with the right-hand side of
Eq.~\eqref{corr-cum}, one is naturally led to the following identification:
\begin{align}
\label{M-GC}
\doubleangle{\xi(t_{1}) \xi(t_{2}) \dots \xi(t_{n})}^{(G)} 
&= \overline{\xi^{n}}
\left[
R^\prime(t_{2}-t_{1})
+ R(0)\,\delta(t_{2}-t_{1})
\right]
\nonumber\\
&\quad \times
\delta(t_{3}-t_{2}) \dots \delta(t_{n}-t_{n-1}) .
\end{align}
Finally, by using Eq.~\eqref{GCum_x} one finds
\begin{align}
\label{Gx}
G(k\pv t,u)
= \left[ \hat p(k) - 1 \right]
\left[
R^\prime(t-u)
+ R(0)\,\delta(t-u)
\right] .
\end{align}
Here, $\hat p(k)$ denotes the CF (i.e., the Fourier
transform) of the PDF of the random variable $\xi$:
\begin{equation}
\label{pk}
\hat p(k)
:= \int e^{\I k \xi}\, p(\xi)\,\mathrm{d}\xi
= \overline{e^{\I k \xi}} .
\end{equation}

From Eq.~\eqref{Gx}, it should be noted that the Green function
$G(k\pv t, u)$ depends only on the time difference $t-u$, and that its
dependence on $k$ and time is separable; that is, it can be written as
the product of a function of $k$ and a function of time.

By inserting Eq.~\eqref{Gx} into the ME~\eqref{MEGx}, we obtain
\begin{equation}
\label{MEGx_}
\partial_t \hat{P}_{CTRW}(k\pv t)
=
\left[ \hat p(k)-1 \right]
\left[
\int_0^t \mathrm{d}u\,
R^\prime(t-u)\,
\hat{P}_{CTRW}(k\pv u)
+ R(0)\,\hat{P}(k\pv t)
\right].
\end{equation}
By taking the Laplace transform of this equation and performing some
straightforward algebra, one readily obtains
\begin{align}
\label{MW}
\hat{P}_{CTRW}(k\pv s)
&=
\frac{1}{s}\,
\frac{1}{1-\hat R(s)\left[\hat p(k)-1\right]}
\,\hat{P}(k\pv 0).
\end{align}
%
By expressing the Laplace transform of the rate function $R$ in terms 
of the Laplace transform of the waiting-time PDF, as given in 
Eq.~\eqref{Rtilde}, we recover the standard Montroll--Weiss and Scher result 
\cite{Montroll_Weiss_1965,Scher_Motroll_1975} for the Fourier--Laplace transform 
of the CTRW PDF. As noted in Appendix~\ref{app:Montroll_Weiss}, this same 
result is also obtained by directly employing the correlation 
functions provided in Eq.~\eqref{corrGen_delta}.\\

While the $G$-cumulant procedure has been used here to reproduce a 
well-known result, its primary advantage lies in its broad applicability. 
Indeed, the method extends almost seamlessly beyond the free CTRW with 
additive noise, and applies equally when the stochastic variable $x$ is 
subject to a drift field $-C(x)$ and/or to multiplicative noise, i.e., in the 
general setting described by Eq.~\eqref{SDE}. The key modification with 
respect to the standard CTRW case is the requirement to work with random 
Liouville operators rather than scalar random variables.

\subsection{$G$-cumulants for the random variable $x$ of the SDE \eqref{SDE}\label{sec:Gcumulants_x}}
We now extend the $G$-cumulant formalism of
Sections~\ref{sec:Gcumulants_CTRW}--\ref{sec:Gcumulants_renewal}
to the general case in which the noise is multiplicative and a drift
velocity field is present.
The result,  one of the main of the present paper, is the following:
\begin{proposition}\label{prop:exactME}
If  $\xi[t]$ is a spike stochastic renewal process and for arbitrary
WT PDFs and jump PDFs with finite moments, the
exact master equation (ME) for the PDF of the general
PDE in
Eq.~\eqref{SDE} is given in Eq.~\eqref{ME_fin}
\end{proposition}

To demonstrate Proposition~\ref{prop:exactME} we start from 
 the  general SDE~\eqref{SDE}. 

For any given realization of the noise $\xi(u)$, where $u \in [0,t]$, the time evolution of the probability density function (PDF) of the system~\eqref{SDE}---denoted by $P(x,\xi(t) \pv t)$---is governed by a continuity equation. This relationship is expressed by the following stochastic Liouville equation:        
        %
\begin{align}
\label{stochLiouv}
\partial_t P(x,\xi(t)\pv t)={\mathcal{L}}_a P(x,\xi(t)\pv t) +{\mathcal{L}}_I \xi(t)P(x,\xi(t)\pv t),
\end{align} 
where 
\begin{equation}
\label{La}
{\mathcal{L}}_a:=\partial_x C(x)
\end{equation}
and 
\begin{equation}
\label{LI}
{\mathcal{L}}_I:=-\partial_x I(x)
\end{equation}
%
denote the unperturbed and perturbation Liouvillians, respectively. Transforming to the interaction representation, we obtain:
\begin{equation}
\label{continuity_}
{\partial}_t \tilde P(x,\xi(t)\pv t)=
\tilde {\mathcal{L}}_I (t)\xi(t) P(x,\xi(t)\pv t)
\end{equation} 
where 
\begin{equation}
\label{Ptilde}
\tilde P(x,\xi(t)\pv t) :=e^{-{\mathcal{L}}_a t}
P(x,\xi(t)\pv t)
\end{equation}
%
and
%
\begin{equation}
\label{LITilde}
\tilde {\mathcal{L}}_I (t):=- e^{{-\mathcal{L}}_at} {\partial}_{x} I(x) e^{{\mathcal{L}}_at}.
\end{equation}
%

Assuming that the initial preparation of the ensemble, $P(x\pv 0)$, 
is independent of the possible realizations of $\xi$ (i.e., the 
initial PDF factorizes as $P(x\pv 0) p_0(\xi)$), the temporal 
integration of Eq.~\eqref{continuity_}, combined with an average 
over all realizations of $\xi(u)$ for $u \in [0, t]$, leads to the 
reduced PDF for the variable $x$ at any time $t \geq 0$:
\begin{equation}
\label{pExp_}
\tilde P(x\pv t):=\left \langle\tilde  P(x,\xi(t)\pv t)\right \rangle
=\left\langle 
\exp_O\!\left[ \int^t_0 \text{d}u\,\tilde {\mathcal{L}}_I (u) \xi(u) \right]\right \rangle P(x\pv 0)
\end{equation}
where, as in Eq.~\eqref{K_CTRW}, the symbol $O$ indicates that
the $O$ map must be combined with the exponential function. The
$O$ map corresponds to the \emph{partial time-ordering} (PTO) map.
Consequently, $\exp_O[\dots]$ denotes the standard time-ordered
(or $t$-ordered) exponential function, often represented as
$\overleftarrow{\exp}[\dots]$
\footnote{See \ref{sec:Gcumulants_xi} for details on the difference
between the ``$O$'' (PTO) map and the ``$G$'' (TTO) map.}.

By comparing Eq.~\eqref{pExp_} with Eq.~\eqref{CF_CTRW}, the time-evolution operator $\langle \exp_O \left[ \int_{0}^{t} \text{d}u \tilde{\mathcal{L}}_I(u) \xi(u) \right] \rangle$ in Eq.~\eqref{pExp_} can be interpreted as a generalized characteristic function (CF). While it no longer represents a simple Fourier transform of a scalar PDF, it serves as the operator-valued generating function for the moments of the random operator
$   \mathscr{S}(t) := \int_0^t \text{d}u \tilde{\mathcal{L}}_I(u) \xi(u)$:
\begin{align}
\label{pExp2_}
&      \tilde P(x\pv t)
= \left\langle \exp_O\!\left[\mathscr{S}(t)\right]\right \rangle P(x\pv 0)
=\left\{\sum_{n=0}^\infty \frac{1}{n!}
\left\langle \left\{\mathscr{S}(t)\right\}_O^n\right\rangle \right\}P(x \pv t)\nonumber\\
&=\left\{\sum_{n=0}^\infty 
\int^t_0{\text{d}u_n}\int^{u_{n}}_0{\text{d}u_{n-1}}\dots 
\int^{u_{2}}_0\text{d}u_1\,
\tilde {\mathcal{L}}_I(u_n)\tilde {\mathcal{L}}_I(u_{n-1})\dots \times \right.
\nonumber \\
&\left.       \dots \times \tilde {\mathcal{L}}_I(u_1)
\langle \xi (u_1)\xi (u_2)\dots \xi (u_n)\rangle \right\} P(x\pv 0)\nonumber \\
&= \left\{\sum_{n=0}^\infty 
\int^t_0{\text{d}u_n}\int^{u_{n}}_0{\text{d}u_{n-1}}\dots 
\int^{u_{2}}_0\text{d}u_1\,
\right.
\nonumber \\
&       \sum_{\pi(n)}
\prod_{B\in\pi(n)}\left[\tilde {\mathcal{L}}_I(u_B)\right]^{\lvert B\rvert}
\overline{\xi^{\lvert B\rvert}}\,
R(u_B - u_{B-1})\,
\delta(\Delta u_{B})\Bigg\} P(x\pv 0).
\end{align}

where, in the last line, we have exploited the central result in 
Eq.~\eqref{corrGen_delta_pre} (or the equivalent one in Eq.~\eqref{corrGen_delta}).

Therefore, as in Sections~\ref{sec:Gcumulants_CTRW}, we can formally introduce the corresponding
$G$-cumulants via
\begin{align}
\label{P_cum_temp}
\tilde P(x\pv t) =& \exp_G
\Bigg[\sum_{n=1}^\infty 
\int_0^t \mathrm{d}u_n \int_0^{u_n} \mathrm{d}u_{n-1}\dots \nonumber \\
&\dots 
\int_0^{u_2} \mathrm{d}u_1 \left[\tilde{\mathcal{L}}_I(u_n)\right]^n 
\doubleangle{\xi(u_1)\xi(u_2)\dots\xi(u_n)}^{(G)}
\Bigg] P(x\pv 0),
\end{align}
that yields to the ME akin to that of Eq.~\eqref{MEGx}:
        \begin{equation}
\label{MEGx_GEN}
\partial_t      \tilde P(x\pv t)=\int_0^t \text{d}u\,G\left(-\I\tilde{\mathcal{L}}_I(u)\pv t,u\right)
\tilde P(x\pv t).
\end{equation} 
with 
        \begin{align}
\label{GCum_x_L}
& G(-\I \tilde{\mathcal{L}}_I(t)\pv u,u')
=\sum_{n=1}^\infty \left[\tilde{\mathcal{L}}_I(t)\right]^n \int_{u'}^{u} \text{d}u_{n-1} \int_{u'}^{u_{n-1}} \text{d}u_{n-2}...\int_{u'}^{u_{3}} \text{d}u_{2}
\nonumber\\
&\times 
\doubleangle{\xi(u') \xi(u_{2}) \dots \xi(u_{n-1})\xi(u)}^{(G)} 
\end{align}

Comparing Eq.~\eqref{pExp2_} with Eq.~\eqref{P_cum_temp} we arrive at the counterpart of Eq.~\eqref{corr-cum}:
		\begin{align}
	\label{corr-cum_Gen}
	&\int_0^t \mathrm{d}u_n \int_0^{u_n} \mathrm{d}u_{n-1} \dots \int_0^{u_2} \mathrm{d}u_1\nonumber \\
	&\times			\sum_{\pi(n)} \prod_{B \in \pi(n)} 
	\left[{\mathcal{L}}_I(u_B)\right]^{|B|} \,\overline{\xi^{|B|}}\, R(u_B - u_{B-1})\, \delta(\Delta u_B) \nonumber \\
	&= \int_0^t \mathrm{d}u_n \int_0^{u_n} \mathrm{d}u_{n-1} \dots \int_0^{u_2} \mathrm{d}u_1 \nonumber \\
	&\times	
	\sum_{\pi(n)} \prod_{B \in \pi(n)} \left[{\mathcal{L}}_I(u_B)\right]^{|B|} \doubleangle{\xi(u_1)\xi(u_2)\dots \xi(u_{n-1})\xi(u_n))}^{(G)}.
\end{align}

Following the approach in Section~\ref{sec:Gcumulants_renewal}, integrating the left-hand side of Eq.~\eqref{corr-cum_Gen} by parts leads to the same identification of $G$-cumulants as provided in Eq.~\eqref{M-GC}. By substituting Eq.~\eqref{M-GC} into Eq.~\eqref{GCum_x_L} and subsequently into Eq.~\eqref{MEGx_GEN}, we obtain 
a generalized ME that is formally identical to that for the free CTRW
case (Eq.~\eqref{MEGx_}), but applied directly to the interaction
representation of the PDF of $x$, rather than its Fourier transform,
and where the wave number $k$ is replaced with $-\I \tilde{\cal L}_I(t)$%
\footnote{The $G$-cumulant technique provides a general formal framework
for deriving the ME with memory kernels in Eq.~\eqref{MEGx_GEN}.
Alternatively, one may recursively integrate the same equation and
directly verify that this procedure yields the right-hand side of
Eq.~\eqref{corr-cum_Gen}, with the cumulants of $\xi[t]$ given by
Eq.~\eqref{M-GC}.}:
\begin{equation}
\label{MEGx_I}
\partial_t \tilde P(x\pv t) =
\left[\hat p\left(-\I \tilde{\mathcal{L}}_I(t)\right)-1\right]
\left[
\int_0^t \mathrm{d}u\, R^\prime(t-u) \tilde P(x\pv u) + R(0)\,\tilde P(x\pv t)
\right].
\end{equation}

Removing the interaction representation, Eq.~\eqref{MEGx_I} leads to the central
result, Eq.~\eqref{ME_fin}, which can also be written as
\begin{align}
\label{MEGx_I_LaPOI}
&\partial_t P(x\pv t) = {\cal L}_a P(x\pv t) \nonumber \\
&+ \int \mathrm{d}\xi\, p(\xi) \sum_{n=1}^\infty (-1)^n \frac{\xi^n}{n!} [\partial_x I(x)]^n
\Bigg[
\int_0^t \mathrm{d}u\, R^\prime(u) e^{{\cal L}_a u} P(x\pv t-u)
+ R(0) P(x\pv t)
\Bigg].
\end{align}

        In the absence of drift, the Laplace transform yields:
\begin{align}
\label{MW_I}
P(x\pv s) &= \frac{1}{s} \frac{1}{1 - \hat R(s)\left[\hat p(-\I {\cal L}_I) - 1\right]} P(x\pv 0),
\end{align}
where
$\hat p(-\I {\cal L}_I)-1 = \sum_{n=1}^\infty (-1)^n \frac{\overline{\xi^n}}{n!} 
[\partial_x I(x)]^n$.
Eq.~\eqref{MW_I} generalizes the Montroll-Weiss and Scher result
\cite{Montroll_Weiss_1965,Scher_Motroll_1975} to the case of state-dependent
(multiplicative) noise.

\section{Notable simplifying cases: the Universal Local ME
         Theorem\label{sec:further}}

We now arrive to the central result of this paper, anticipated in
Proposition~\ref{prop:ULME_informal}, as a formal theorem, and 
in Eq.~\eqref{ME_universal_intro} analytically.
In fact, the two propositions that follow
(Propositions~\ref{muG2} and~\ref{prop:approxME})
together constitute its proof in the two complementary regimes.

\begin{teorema_}[Universal Local ME]
\label{prop:ULME}
Let $x(t)$ satisfy the SDE~\eqref{SDE} with $\xi[t]$ a spike stochastic
renewal process whose jump PDF $p(\xi)$ has finite moments and whose
WT PDF $\psi(t)$ generates the renewal rate $R(t)$.
Under the conditions of either finite mean
waiting time $\tau<\infty$ or $\psi(t)\sim (t/T)^{-\mu}$
with $1<\mu<2$ and drift timescale $\gg T$,
the ME for $P(x\pv t)$ is well approximated by the
simplified local in time version given in 
Eq.~\eqref{ME_universal_intro}.
\end{teorema_}
The following three properties hold:

\begin{enumerate}[label=(\roman*)]

\item \textbf{Universality with respect to the jump distribution.}
The operator $\hat{p}(\I\partial_x I(x))$ depends on $p(\xi)$ only
through its Fourier transform evaluated at $\I\partial_x I(x)$.
For specific jump distributions it reduces to known closed forms,
for example
(see Appendix~\ref{app:ME_resuls}):
\begin{align}
    \hat{p}(\I\partial_x I)\,P
    &= \cosh\left[a\,\partial_x I(x)\right]\,P
    &&\text{(symmetric dichotomous, } \xi=\pm a\text{)},
    \label{eq:op_dicho_thm}\\
    \hat{p}(\I\partial_x I)\,P
    &= \exp\!\left[\tfrac{a^2}{2}
       \left(\partial_x I(x)\right)^2\right]P
    &&\text{(Gaussian, variance }a^2\text{)},
    \label{eq:op_gauss_thm}\\
    \hat{p}(\I\partial_x I)\,P
    &= \frac{\sinh\!\left[\sqrt{3}\,a\,\partial_x I(x)\right]}
            {\sqrt{3}\,a\,\partial_x I(x)}\,P
    &&\text{(uniform on }[-\sqrt{3}a,\sqrt{3}a]\text{)}.
    \label{eq:op_flat_thm}
\end{align}
These correspond to the explicit
MEs~\eqref{MEGx_dicho},~\eqref{MEGx_gauss}, and~\eqref{MEGx_flat},
respectively.

\item \textbf{Exact in the Poissonian limit.}
When $R(t)\equiv\tau^{-1}$ (exponential waiting-time distribution),
Eq.~\eqref{ME_universal_intro} is not an approximation but coincides
exactly with the Poissonian ME~\eqref{ME_fin_Pois}.

\item \textbf{Non-Markovian memory encoded in $R(t)$ alone.}
The entire history-dependence of the renewal process enters
Eq.~\eqref{ME_universal_intro} through the single scalar function $R(t)$.
When $R(t)\to\tau^{-1}$ (Blackwell theorem,
\citep{Blackwell1948,FellerVol2}),
Eq.~\eqref{ME_universal_intro} reduces to~\eqref{ME_fin_Pois}.
When $R(t)\sim T^{-1}(T/t)^{2-\mu}$ with $1<\mu<2$, the stochastic
forcing is progressively quenched and the deterministic drift
$\partial_x C(x)P$ eventually dominates.

\end{enumerate}

\noindent
The proof proceeds along two routes corresponding to the two
complementary regimes.

\medskip
\noindent\textit{Proof in the finite-mean-waiting-time case
(Proposition~\ref{muG2}).}
For $\tau$ finite ($\mu>2$ in the case of power-law WT PDF),
the Blackwell theorem \citep{Blackwell1948,FellerVol2}
guarantees $R(t)\to\tau^{-1}$ for $t\gg T$.
By Proposition~\ref{muG2_}, the $n$-time correlation functions of
$\xi[t]$ then converge to the Poissonian ones, and the exact
ME~\eqref{ME_fin} reduces to the local form~\eqref{ME_universal_intro}.
The detailed argument is given in Section~\ref{sec:PLWTmu_gt_2}.

\medskip
\noindent\textit{Proof in the heavy-tailed case
(Proposition~\ref{prop:approxME}).}
For $1<\mu<2$, Proposition~\ref{prop:universal} establishes that the
dominant term in the correlation-function
sum~\eqref{corrGen_delta_pre} is the single-block ($p=1$)
composition.
This implies that the $G$-cumulants of $\xi[t]$ reduce to the
approximate form (Eq.~\eqref{cum-corr_2}):
\begin{equation}
    \doubleangle{
      \xi(u_1)\xi(u_2)\cdots\xi(u_n)
    }^{(G)}
    \;\sim\;
    \overline{\xi^n}\,
    R(u_1-t_0)
    \prod_{j=2}^{n}\delta(u_j - u_{j-1}),
\end{equation}
which, when inserted into the Green-function
expansion~\eqref{GCum_x_L} and then into~\eqref{MEGx_GEN}, yields
Eq.~\eqref{ME_universal_intro} directly (see
Section~\ref{sec:universalCTRW}).
\hfill$\square$

\medskip
\noindent
Extensive numerical evidence for
Theorem~\ref{prop:ULME} is presented in
Figs.~\ref{fig:Def_pdedico_t5-30-100_mu1.5_tz2.00_ga1.2_Cubic_appr}--\ref{fig:Def_comp_t0.5-2-6_pde_eq80_mu2.50_tz1.00_ga5.00_be0.50}.
A striking feature of these results is that
Eq.~\eqref{ME_universal_intro} provides excellent agreement with direct
SDE simulations even at \emph{short} times and even when the drift
timescale is not large compared to $T$, i.e.\ well outside the
parameter regime covered by the formal proof.
This suggests that the Universal Local ME has a wider domain of validity
than is currently established analytically, a point we return to in
the conclusions.

The remainder of this section analyses the two cases of the theorem in
detail.

\subsection{Simplifications arising from specific classes of WT PDFs}

Since our general results depend explicitly on the rate function $R(t)$, defined in Eq.~\eqref{R}, rather than the waiting-time (WT) PDF $\psi(t)$ itself, we first recall how the asymptotic behavior of $\psi(t)$ governs that of $R(t)$.

In the Poissonian renewal case, the rate function is constant, $R(t)=1/\tau$. 
Consequently, the ME for the PDF of the variable $x$ in the SDE~\eqref{SDE} 
reduces to the local-in-time PDE given in Eq.~\eqref{ME_fin_Pois}. This 
equation has been derived previously and has provided valuable analytical 
insights in several concrete cases \citep{sprmPRE83,bpPRE98}. For instance, 
in Appendix~\ref{app:ME_resuls}, we determine the asymptotic behavior of the PDF 
satisfying Eq.~\eqref{ME_fin_Pois} for standard settings, such as linear drift or 
linear multiplicative intensity $I(x)$.

Regarding stationarity, Eq.~\eqref{corrGen_delta} shows that, for $t_1>t_0$, the stationarity of the multi-time correlation functions of $\xi(t)$ is determined solely by the first rate function $R$ appearing in each term of the sum—specifically, the one evaluated at $t_1-t_0$. Hence, the Poissonian spike renewal process is intrinsically stationary. More generally, for different WT PDFs with a finite mean $\tau$, Blackwell's renewal theorem~\citep{Blackwell1948,FellerVol2} implies that the rate function becomes time-independent at large times, leading to an asymptotically stationary regime. Typical examples include the Weibull (stretched-exponential) and Gamma WT PDFs.

In the case of Pareto or Manneville-type WT PDFs, where $\psi(t) \sim T^{-1}(T/t)^{\mu}$, the rate function exhibits the following scaling behaviors \citep{bJSTAT2020,mssbgCSF188}:

\begin{enumerate}[label=\Roman*]
\item\label{I_I} For $1 < \mu < 2$: $R(t) \sim T^{-1} (T/t)^{2 - \mu}$,
\item\label{II_I} For $\mu = 2$: $R(t) \sim T^{-1} / \log(t/T)$,
\item\label{III_I} For $\mu > 2$: $R(t) \sim \tau^{-1}\left[1 + (T/t)^{\mu - 2}\right]$.
\end{enumerate}
These facts lead to the following result.

\begin{proposition}
\label{prop:stationary}
In the Poissonian case, the spike stochastic renewal process $\xi[t]$
is stationary. For non-exponential WT PDFs, the process is
asymptotically stationary only if the average waiting time $\tau$ is
finite. On the other hand, for Pareto/Manneville WT PDFs of the form
$\psi(t) \sim T^{-1}(T/t)^{\mu}$ with $1<\mu\le 2$ (cases~\ref{I_I}
and~\ref{II_I}), the process never reaches stationarity.
\end{proposition}

\subsection{WT PDFs with finite mean\label{sec:PLWTmu_gt_2}}

From both the general expression for the multi-time correlation functions,
Eq.~\eqref{corrGen_delta}, and the Blackwell renewal theorem~\citep{Blackwell1948,FellerVol2}, 
we obtain:

\begin{proposition}
\label{muG2_}
If the mean waiting time $\tau$ of the spike stochastic renewal process $\xi[t]$ is finite, then for time lags much larger than the characteristic time scale $T$, the multi-time correlation functions of $\xi[t]$ converge to those of a Poissonian spike renewal process with the same mean waiting time.
\end{proposition}

This implies that when $\tau$ is finite (as in case~\ref{III_I}) and the characteristic time scale of the drift-driven dynamics in the SDE~\eqref{SDE} is significantly larger than $T$, the process $\xi[t]$ is effectively perceived by the system as being approximately Poissonian with rate $1/\tau$. Consequently, the ME for the PDF of $x$ is well-approximated by the local-in-time PDE given in Eq.~\eqref{ME_fin_Pois}. Indeed, starting from the exact ME in Eq.~\eqref{ME_fin}—which is valid for arbitrary WT PDFs—we obtain the following stronger result:

\begin{proposition}
\label{muG2}
Under the same conditions as Proposition~\ref{muG2_}, and for $t \gg T$, the ME for the PDF of $x$ asymptotically converges to the Poissonian form presented in Eq.~\eqref{ME_fin_Pois}.
\end{proposition}

Figure~\ref{fig:Def_comp_t100_dicho_gausspde_mu2.50_tau4_DB} illustrates numerical simulations of the SDE~\eqref{SDE} for a strongly nonlinear additive case driven by renewal noise with a Pareto/Manneville waiting-time distribution, $\psi(t) = (\mu-1)T^{-1}\left(1+t/T\right)^{-\mu}$. Here, $\mu=2.5$, and we consider both dichotomous and Gaussian jump PDFs (see the figure caption for further details). The simulation results are compared with the solutions of the PDEs corresponding to exponential waiting times, evaluated using the same mean waiting time $\tau = T/(\mu-2)$. These PDEs are provided explicitly in Appendix~\ref{app:ME_resuls}), Eqs.~\eqref{MEGx_dicho} and \eqref{MEGx_gauss} for the dichotomous and Gaussian cases, respectively.

Figure~\ref{fig:Def_comp_t100_dicho_gauss_flat_pde_mu3.00_tau2_DB} presents a similar scenario to that of Fig.~\ref{fig:Def_comp_t100_dicho_gausspde_mu2.50_tau4_DB}, but with $\mu=3.5$. This figure additionally includes the case of a flat PDF, for which the theoretical result is obtained in Appendix~\ref{app:ME_resuls}, Eq.~\eqref{MEGx_flat}.

The excellent agreement between the numerical simulations of the SDE and the PDEs for the corresponding Poissonian cases provides strong support for Proposition~\ref{muG2}. To demonstrate that the system operates well outside the regime where the central limit theorem would justify replacing $\xi[t]$ with white noise, the figures also include the PDFs obtained from the white-noise limit—specifically, the Stratonovich-type Fokker--Planck equation given in Eq.~\eqref{WNcase}.

\begin{figure}[ht]
\centering
\includegraphics[width=\textwidth]{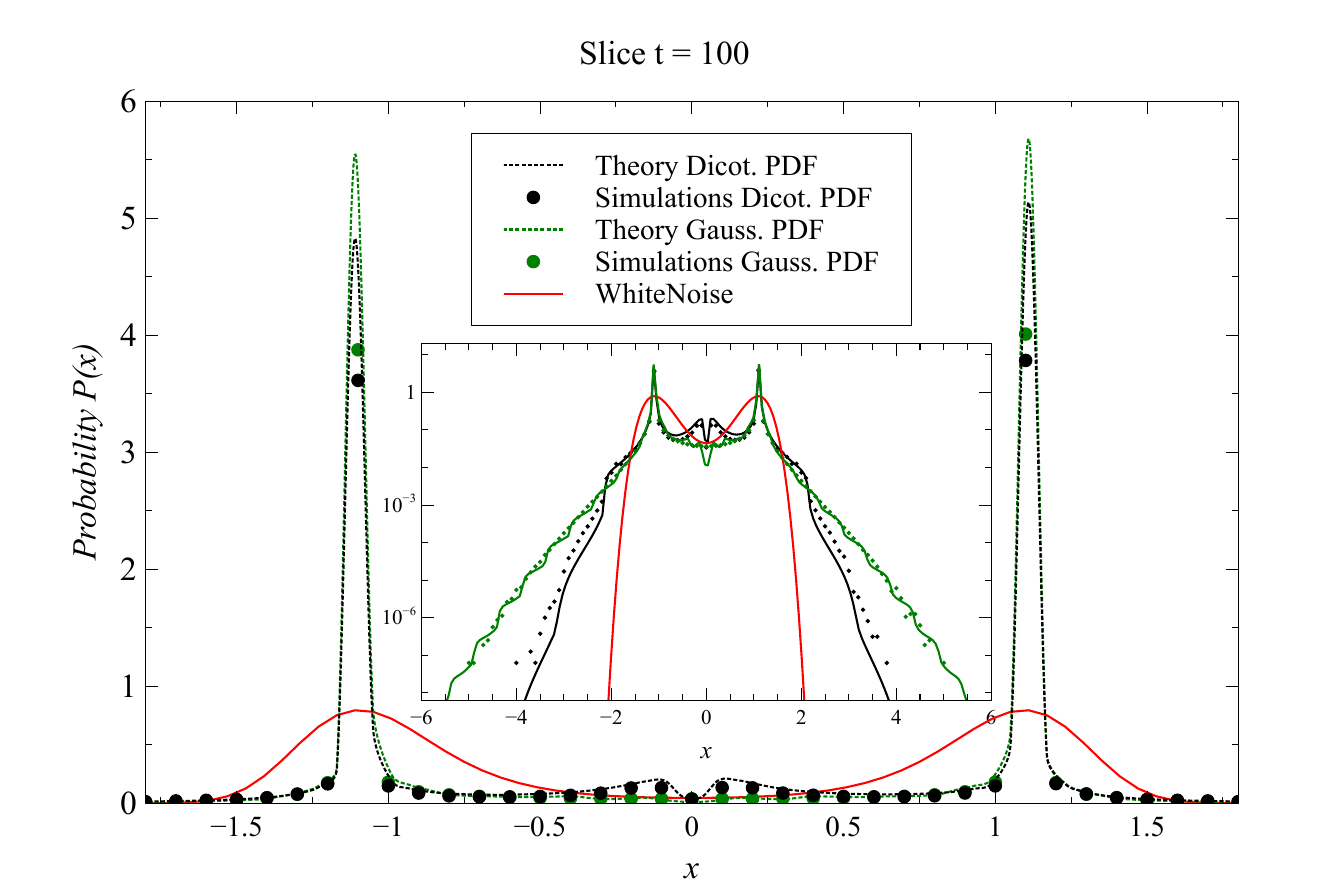}
\caption{\( P(x\pv t) \) at \( t = 100 \) for the SDE~\eqref{SDE} in the
\textit{additive} case, with a strongly nonlinear drift, i.e.,
\( C(x) = -x(1.2 - x^2) \) and \( I(x)=1 \), driven by renewal noise
with both dichotomous and Gaussian PDFs and a \textit{power-law} WT,
$\psi(t) = T^{-1}(\mu-1)\left(1+t/T\right)^{-\mu}$, 
with $T=2$, $\mu=2.5$ and $3$, corresponding to $\tau=4$ and $2$, respectively. The figure shows
the results of numerical simulations of the SDE~\eqref{SDE} (circles)
together with the analytical solutions given in
Eqs.~\eqref{MEGx_dicho}--\eqref{MEGx_gauss} (dotted lines). The red
solid line corresponds to the case in which $\xi[t]$ is replaced by
white noise; it is apparent that the central limit theorem does not
apply in this case. Inset: the same data in logarithmic scale.}
\label{fig:Def_comp_t100_dicho_gausspde_mu2.50_tau4_DB}
\end{figure}
\begin{figure}[ht]
\centering
\includegraphics[width=\textwidth]{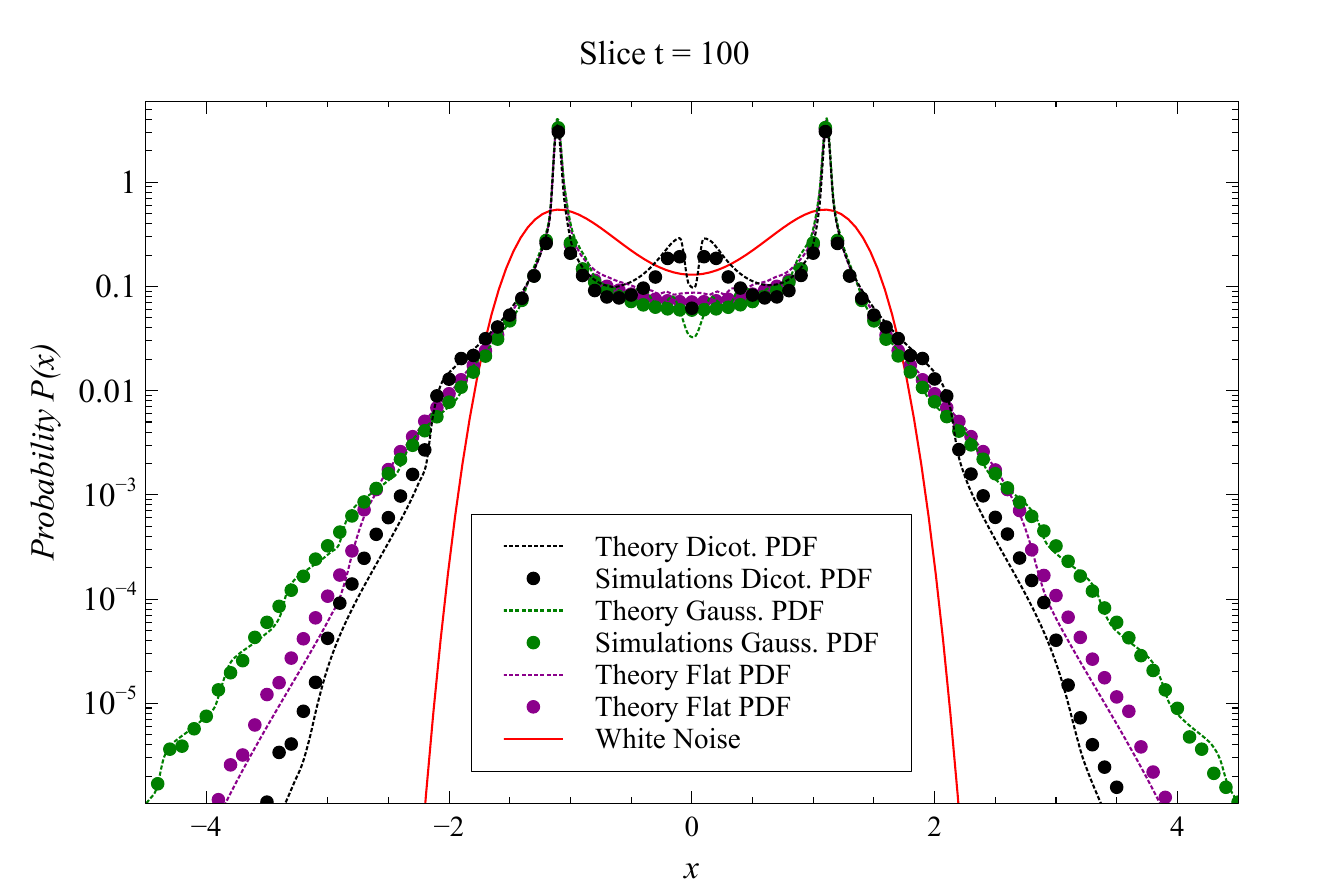}
\caption{Same as Fig.~\ref{fig:Def_comp_t100_dicho_gausspde_mu2.50_tau4_DB}
but with $\mu=3.0$ and $T=2$, from which $\tau=2$. Moreover, here we have
added the case of flat PDF, thus, to improve readability, we show
only the plot in log scale.
}
\label{fig:Def_comp_t100_dicho_gauss_flat_pde_mu3.00_tau2_DB}
\end{figure}
The next section deals with the important case $1 < \mu \le 2$, for which
the WT PDF does not have a finite mean waiting time. In this regime, we
derive an important universal asymptotic result.

\section{The universal limit behavior of the $n$-time  correlation functions for the non stationary case.\label{sec:limt_result}}
If $1 < \mu \le 2$ (case \ref{I_I} and \ref{II_I} of Section~\ref{sec:further}) the system never achieves stationarity, 
as the aging time is infinite. In this situation, given the asymptotic expression of the rate function we have the following:
\begin{proposition}        
\label{prop:universal}
If the WT  PDF exhibits a power-law decay with time scale $T$ and with $1<\mu\le 2$, then for time lags large  relative to $T$, the dominant term 
in the sum \eqref{corrGen_delta} is the one with
$p=1$:
\begin{align}
\label{resultmul2}
&\langle\xi(t_1)\xi(t_2)...\xi(t_n)\rangle_{t_0}
\rightarrow \overline{\xi^{n}}\,\tilde R(t_n-t_0)\delta(t_n-t_1)
\delta(t_2-t_1)\times...
\times 
\delta(t_{n-1}-t_{n-2})
\nonumber \\&
=\frac{\overline{\xi^{n}}}{\overline{\xi^{2}}}\langle\xi(t_1)\xi(t_n)\rangle_{t_0}
\delta(t_2-t_1)\times...
\times 
\delta(t_{n-1}-t_{n-2})
\nonumber \\&
\sim \overline{\xi^{n}}\,T^{-1}\left(\frac{t_n-t_0}{T}\right)^{-(2-\mu)}
\delta(t_n-t_1)\delta(t_2-t_1)\delta(t_3-t_2)\times...
\times 
\delta(t_{n-1}-t_{n-2}).
\end{align}
%
\end{proposition}
Note that  the limit result \eqref{resultmul2} depends only on the average of $\xi^n$ and has a power law time decay with exponent $\mu-2<0$.
By combining Proposition~\ref{prop:universal} with the universal result for the two times correlation function found in \cite{bblmCSF196},  we arrive at
the following
\begin{lemma}
\label{lem:2}
Under the conditions outlined in Proposition~\ref{prop:universal},
the common asymptotic expression for the $n$-time 
correlation functions \uline{for
any  stochastic renewal process of the spike type},
is given by Eq.~\eqref{resultmul2}, and it corresponds to the universal two-time
correlation function evaluated at the extreme times.
\end{lemma}
%


The reader should appreciate that the universal property stated in Lemma~\ref{lem:2} implies a corresponding universal statistical behavior for any Brownian variable with drift perturbed by renewal noise, as modeled in Eq.~\eqref{SDE}.

Before proving Proposition~\ref{prop:universal}, it is worth highlighting that, in general, once the $n$ time-points of the correlation function are fixed, only one of the $2^{n-1}$ compositions in Eq.~\eqref{corrGen_delta} survives. Consequently, at first glance, the aforementioned Proposition and Lemma might appear vacuous.

However, we recall that in practical applications, the spike stochastic renewal process $\xi[t]$ serves as a noise source within an SDE, such as Eq.~\eqref{SDE}. As a result, the $n$-time correlation functions of $\xi[t]$ always appear within $n$-fold time integrals. When applying Eq.~\eqref{corrGen_delta}, this $n$-fold integral naturally decomposes into a sum of $2^{n-1}$ distinct integrals, each corresponding to a specific composition of the $n$ time variables.

Proposition~\ref{prop:universal}, in conjunction with Lemma~\ref{lem:2}, asserts that under a time-scale separation between the dynamics of the unperturbed system \eqref{SDE} and those of the noise, the dominant contribution among these $2^{n-1}$ integrals arises from the term corresponding to the single-block composition ($p = 1$). This term contains the $n$-fold time integral defined in Eq.~\eqref{resultmul2}.

Another significant implication of Proposition~\ref{prop:universal} concerns cases where the PDF of the random variable $\xi$ follows a power-law distribution, $p(\xi) \sim \xi^{-\beta}$. In such instances, the convergence expressed in Eq.~\eqref{resultmul2} is achieved not only in the limit of large time lags but also by increasing the order $n$.

This behavior occurs because, for such heavy-tailed distributions, the coefficients in Eq.~\eqref{corrGen_delta} of the form $\left[(\overline{\xi^{m_1}})(\overline{\xi^{m_2}})\dots(\overline{\xi^{m_p}})\right]$—comprising $p$ moments where $\sum_{i=1}^p m_i = n$—are dominated for large $n$ by the $p = 1$ term. This term involves a single block with the $n$-th moment $\overline{\xi^n}$, precisely as described in Proposition~\ref{prop:universal}.

\section{The universal limit behavior of the $n$-time correlation functions for
heavy-tailed PDFs of $\xi$.\label{sec:limt_result_p}}

This section introduces Proposition~\ref{prop:PDF_power}, which is
analogous to Proposition~5 of~\cite{bblmCSF202}. Since the arguments
and conclusions are essentially identical, we restrict ourselves here
to stating the result. For the detailed discussion and proof, the
reader is referred to the cited paper:
\begin{proposition}
\label{prop:PDF_power}
In the case of a power-law PDF for the random variable $\xi$, i.e.,
$p(\xi) \sim \xi^{-\beta }$ with $\beta  < n+1$ and $n$ a given integer, let us
redefine $\overline{\xi^m}$ for $m \le n$ as the empirical average of $\xi^m$
computed over a large but finite number $N$ of realizations. In this setting,
$\overline{\xi^n}$ increases with $N$, and for any fixed $n$, the convergence
described by Eq.~\eqref{resultmul2} is achieved simply by increasing $N$,
regardless of the time lags.
\end{proposition}
%


\section{The  universal limit statistical behavior of generic systems driven by $\xi[t]$
for heavy-tailed power law  $\psi(t)$ \label{sec:universalCTRW}}
Let us consider the variable of interest $x$, governed by the SDE
in Eq.~\eqref{SDE}. The general expression for its PDF, written in the
interaction representation, depends on all multi-time correlation
functions of $\xi[t]$ and is given in Eq.~\eqref{pExp2_}. The
corresponding ME is reported in Eqs.~\eqref{ME_fin} and
\eqref{MEGx_I_LaPOI}, but the resulting expressions are rather
cumbersome.

However, a crucial simplification comes from the following
\begin{proposition}\label{prop:approxME}
Let us consider the general SDE of Eq.~\eqref{SDE}.
If  $\xi[t]$ is a spike stochastic renewal process with arbitrary
 jump PDFs with finite moments 
 and  WT PDF decaying as $(t/T)^{-\mu}$ with
 $1<\mu< 2$, and if the unperturbed (drift) dynamics of $x$ 
 is characterized by
 a timescale larger than $T$, then for $t\gg T$ the ME for the PDF of $x$
 is well approximated by that in Eq.~\eqref{ME_universal_intro}.
\end{proposition}

To demonstrate Proposition~\ref{prop:approxME}, we note that in this situation
Proposition~\ref{prop:universal} applies, thus
 the dominant contribution in the composition sum of
Eq.~\eqref{corrGen_delta} is the term corresponding to $p=1$.

This observation can be exploited in two different ways. The first
approach consists in using the $G$-cumulant method to derive the Green
function associated with the ME~\eqref{MEGx_GEN} for the
PDF of $x$. Alternatively, one may directly evaluate the PDF of $x$ by
expanding the exponential series in Eq.~\eqref{corrGen_delta}.

We begin with the first approach. To this end, it is useful to recall
the relation giving multi-time $G$-cumulants  in terms of multi-time correlation
functions (the inverse of the one exploited in Eq.~\eqref{mukappaCombiMulti},see, e.g., Sec.~4.4.3 of~\citep{bbJSTAT4}):
\begin{align}
\label{cum-corr}
\doubleangle{\xi(u_1)\xi(u_2)\dots\xi(u_n)}
=
\sum_{\pi(n)} (-1)^{\left|\pi\right|-1}
\prod_{B\in\pi(n)}
\left\langle \prod_{i\in B}\xi(u_i)\right\rangle ,
\end{align}
where $\pi(n)$ runs again through the list of all set-compositions (ordered partitions) of $n$ 
distinguishable objects, $B$ labels the blocks
of the composition $\pi$, and $\left|\pi\right|$ is the number of blocks in the composition
(the number $p$ in Eq.~\eqref{corrGen_delta}).

Using Proposition~\ref{prop:universal} and Eq.~\eqref{resultmul2} in the
r.h.s. of Eq.~\eqref{cum-corr}, we obtain
\begin{align}
\label{cum-corr_}
\doubleangle{\xi(u_1)\xi(u_2)\dots\xi(u_n)}
\sim
\sum_{\pi(n)} (-1)^{|\pi|-1}
\prod_{B\in\pi(n)}
\overline{\xi^{|B|}}
\,R(u_B-t_0)
\delta(\Delta u_B).
\end{align}
where $\,R(u_B-t_0)
\approx T^{-1}\!\left(\frac{u_B-t_0}{T}\right)^{-(2-\mu)}$.

By the same argument used in Proposition~\ref{prop:universal}, the
dominant contribution in Eq.~\eqref{cum-corr_} arises from the partition
with a single block, i.e.\ $\left|\pi\right|=1$, yielding
\begin{align}
\label{cum-corr_2}
\doubleangle{\xi(u_1)\xi(u_2)\dots\xi(u_n)}
\sim
\overline{\xi^{n}}\,
R(u_1-t_0)
\prod_{j=2}^{n}\delta(u_j-u_{j-1}).
\end{align}

Inserting Eq.~\eqref{cum-corr_2} into Eq.~\eqref{GCum_x_L}, and then into
Eq.~\eqref{MEGx_GEN}, we obtain (setting $t_0=0$)
\begin{equation}
\label{MEGx_I_mu}
\partial_t \tilde P(x\pv t)
\sim 
T^{-1}\!\left(\frac{t}{T}\right)^{-(2-\mu)}\!
\left[
\hat p\!\left(-\I \tilde{\mathcal{L}}_I(t)\right)-1
\right]
\tilde P(x\pv t).
\end{equation}

Removing the interaction representation yields the limit ME already presented in section~\ref{sec:summary}, adapted to the case $1<\mu<2$:
\begin{align}
\label{ME_fin_mu}
\partial_t P(x\pv t)
\sim
&\;
\partial_x C(x)\,P(x\pv t)
\nonumber\\
&+
T^{-1}\!\left(\frac{t}{T}\right)^{-(2-\mu)}\!
\left[
\hat p\!\left(\I\partial_x I(x)\right)-1
\right]
P(x\pv t).
\end{align}
This end the demonstration of Proposition~\ref{prop:approxME}.

The second approach is based on a direct evaluation of the PDF of $x$.
In this case, Eq.~\eqref{pExp_} becomes
\begin{align}
\label{pExp2__}
\tilde P(x\pv t)
&=
\left\langle
\exp_O\!\left[
\int_0^t \!\mathrm{d}u\,
\tilde{\mathcal{L}}_I(u)
\right]
\right\rangle
P(x\pv 0)
\nonumber\\
&\approx
\left\{
1+
\int_0^t\!\mathrm{d}u\,R(u)
\left[
\hat p\!\left(\tilde{\mathcal{L}}_I(u)\right)-1
\right]
\right\}
P(x\pv 0),
\end{align}
where $R(t)=T^{-1}(t/T)^{-(2-\mu)}$ with $1<\mu<2$.

Integrating Eq.~\eqref{MEGx_I_mu} and retaining only terms linear in
$R$, one readily recovers Eq.~\eqref{pExp2__}. Removing the interaction
representation finally gives
\begin{align}
\label{pExp2__P}
P(x\pv t)
\approx
\left\{
1+
\int_0^t\!\mathrm{d}u\,R(u)
\left[
\hat p\!\left(-\I\tilde{\mathcal{L}}_I(u-t)\right)-1
\right]
\right\}
P_a(x\pv t),
\end{align}
where $P_a(x\pv t)=\exp(\mathcal{L}_a t)P(x\pv 0)$ denotes the unperturbed
evolution.

In the simplified case $\mathcal{L}_a=0$, the Laplace transform of
Eq.~\eqref{pExp2__P} yields
\begin{align}
\label{PDF_SDE}
\hat P(k\pv s)
\approx
\frac{1}{s}
\left[
1+
\hat R(s)
\left(
\hat p\!\left(\I\partial_x I(x)\right)-1
\right)
\right]
\hat P(k\pv 0).
\end{align}

Equations~\eqref{pExp2__P}–\eqref{PDF_SDE} are valid under the assumption
$R(t)=T^{-1}(t/T)^{-(2-\mu)}$ with $1<\mu<2$.

\section{The dichotomous jump PDF case\label{sec:dico_leapers}}
The dichotomous jump-PDF case deserves attention, as it exhibits
compelling formal properties despite its limited practical relevance.

The standard symmetric dichotomous process—for a comprehensive discussion of its central role in stochastic processes, see Ref.~\cite{bCSF159} and references therein—assumes one of two values, $\pm a$. It maintains this value for a random duration $\theta$ until it either switches or remains the same based on a random coin toss. This process is of the \emph{step-type}, characterized by the PDF $p(\xi) = \frac{1}{2}\delta(\xi + a) + \frac{1}{2}\delta(\xi - a)$.

For the \emph{spike-type} symmetric dichotomous noise, the PDF of the random variable $\xi$ remains $p(\xi) = \frac{1}{2}\delta(\xi + a) + \frac{1}{2}\delta(\xi - a)$; however, the trajectories consist of discrete spikes rather than constant steps. As demonstrated below, this structural difference leads to several crucial consequences.

We begin by recalling that the \textit{standard} (step-type) dichotomous noise is a stationary process when the waiting time (WT) PDF decays exponentially (i.e., a Poissonian process). In this case, known as telegraph noise, the two-time correlation function also exhibits exponential decay (see, for example, \cite{bblmCSF196, bblmCSF202}):
\begin{align}\label{poisson_}
\langle\xi(t_1)\xi(t_{2})\rangle =
a^2\,
e^{-\frac{1}{\tau}(t_2-t_1)}.
\end{align}
%
Furthermore, for telegraph noise, the $n$-time correlation function satisfies the following factorization property:
\begin{align}
    \label{corrGen_delta_telegraph}
    \langle\xi(t_1)\xi(t_2)\dots\xi(t_n)\rangle = 
    \langle\xi(t_1)\xi(t_2)\rangle \langle\xi(t_3)\xi(t_4)\rangle \dots 
   \langle\xi(t_{n-1})\xi(t_n)\rangle.
\end{align}

This factorization is of fundamental importance because it implies that $\xi(t)$ is $G$-Gaussian; that is, all $G$-cumulants beyond the second order vanish~\cite{bCSF159}. Consequently, this yields a ME characterized by a simple memory kernel, applicable to both the characteristic function of $\xi[t]$ and the PDF of the variable $x$ in the SDE~\eqref{SDE}. A significant result from Ref.~\cite{bCSF159} (specifically Proposition~2 and Lemma~1) is that this factorization property also holds for \emph{non-dichotomous} Poissonian stochastic renewal processes of the \emph{step} type in the limit of large time lags.

In contrast, for Poissonian dichotomous processes of the \emph{spike} type---the primary focus of this work---the two-time correlation function does not exhibit exponential decay. Since the rate function is constant, Eq.~\eqref{cor2t_delta_pre} yields:
\begin{align}
    \label{cor2t_delta_dicho}
    \langle \xi(t_1)\xi(t_2)\rangle_{t_0} &= \overline{\xi_0^2}\,\delta(t_1-t_0)\delta(t_2-t_1) + \frac{a^2}{\tau}\,\delta(t_2-t_1).
\end{align}
For $\overline{\xi_0^2}=0$ (the standard assumption when considering the SDE~\eqref{SDE} with shot noise) and/or for $t_1 > t_0$, this expression coincides with the correlation function of white noise (see also~\cite{bblmCSF196}).

Furthermore, the standard factorization property described in Eq.~\eqref{corrGen_delta_telegraph} does \emph{not} hold for spike-type processes. Instead, they obey a more general, non-stationary factorization (see Appendix~\ref{app:dicho}):
\begin{align}
\label{corrGen_delta_dicho}
\langle \xi(t_1)\xi(t_2)\cdots\xi(t_n)\rangle_{t_0}
&=
\langle \xi(t_1)\xi(t_2)\rangle_{t_0}
\langle \xi(t_3)\xi(t_4)\rangle_{t_2}
\langle \xi(t_5)\xi(t_6)\rangle_{t_4}
\times \nonumber \\
&\qquad \cdots
\times
\langle \xi(t_{n-1})\xi(t_n)\rangle_{t_{n-2}} .
\end{align}
Remarkably, Eq.~\eqref{corrGen_delta_dicho} holds not only for the Poissonian case but for any symmetric dichotomous spike process, \emph{independently of the specific waiting time (WT) PDF}. This shows that the generalized factorization observed explicitly for $n=4$ and $n=6$ is not incidental; rather, it is a generic feature of renewal processes characterized by ``spike-like'' trajectories when $\xi$ is a two-state random variable.

However, despite its elegant form, the generalized factorization property in Eq.~\eqref{corrGen_delta_dicho} does not imply $G$-Gaussianity. The explicit dependence of each factor on the preceding time-point renders the process non-$G$-Gaussian, thereby making this factorization ineffective for truncating the $G$-cumulant expansion. Even in the Poissonian limit, Eq.~\eqref{corrGen_delta_dicho} does not reduce to the standard factorization, as can be immediately verified by substituting Eq.~\eqref{cor2t_delta_dicho} into Eq.~\eqref{corrGen_delta_dicho}.

Indeed, no stochastic renewal process is truly $G$-Gaussian, as $G$-Gaussianity would require an exact truncation of the Master Equation~\eqref{ME_fin} at the second differential order—a condition that is never satisfied for a physically admissible PDF.

\section{Numerical techniques\label{sec:numericaltechinques}} 
In this section we detail the various numerical tools employed in the paper.

\subsection{Integration of the SDE}
The SDE of Eq.~\eqref{SDE} is integrated using the Heun's scheme which enforces Stratonovitch 
calculus~\cite{RMannellaSDE,RMannellaItoStrat,
RMannellaHeun2025}, using discrete time steps $h$. Given
$x(t)$, the elementary evolution over the time step $h$ to get
$x(t+h)$ depends on whether or not within $h$ we are going to
have at least one delta kick. Assume that the next kick is going to
        happen at the time $t_1$.
If there is no delta kick within $h$ (i.e. $t_1 > t+h$), the algorithm reads
\begin{equation}
\left.
\begin{aligned}
\tilde{x} &= x(t) - h C(x(t)) \\
x(t+h) &= x(t) - \frac{h}{2} \left( C(x(t)) + C(\tilde{x}) \right)
\end{aligned} \quad
\right\} \text{NOKICK} \label{eq:nojump}
\end{equation}
If at least one delta kick is expected
        within $h$ (i.e. $t< t_1 < t+h$), we generate both the amplitude of the delta kick at the time $t_1$ ($\xi_1$, from the PDF) and the time to the following kick ($\Delta t_2$, from the WT). If the time of the following kick (i.e.
$t_2 = t_1+\Delta t_2$) is larger than $t+h$, we carry out the integration as
\begin{equation}
\left.
\begin{aligned}
\tilde{x} = & x(t) - h C(x(t)) + \xi_1 I(x(t)) \\
x(t+h) = & x(t) - \frac{h}{2} \left( C(x(t))  + C(\tilde{x})\right) + \frac{\xi_1}{2} \left( I(x(t))+I(\tilde{x})\right)
\end{aligned} \quad 
\right\} \text{KICK} \label{eq:jump}
\end{equation}
        On the other hand, 
        if the time of the following kick
$t_2$ is still within the current $[t,t+h]$ interval, we carry out a KICK
        step as in Eq.~\eqref{eq:jump} using $t_2-t$ as time step instead of $h$ to
        get $x(t_2)$; then  we generate the time to the next kick ($\Delta t_3$,
        from the WT), and
        the kick amplitude at $t_2$ ($\xi_2$, from the PDF).
        Again, we check whether $t_3 = t_2+\Delta t_3 > t+h$: if it does (no further 
        kicks within $[t,t+h]$), we move from $x(t_2)$ to
        $x(t+h)$ with the KICK step of Eq.~\eqref{eq:nojump} using a time step $h-t_2$ rather than $h$, and kick amplitude  $\xi_2$. If, on the other hand, $t_3$ is still within  $[t,t+h]$, we move
        from $x(t_2)$ to $x(t_3)$ using
        a KICK step with time step $t_3-t_2$, and kick amplitude $\xi_2$, and so on, until eventually the time for the next kick exceeds $t+h$.
        The algorithm is then repeated, moving from the interval $[t,t+h]$ to the
        interval $[t+h,t+2 h]$ etc..
The presence of $\frac{\xi_i}{2} \left( I(x(t))+I(\tilde{x})\right)$ in the second stage of a KICK step instead of
$\xi_i I(x(t))$
is the algorithmic implementation of \^Ito calculus. 
In the simulations, given the parameters used, typically we have $h\approx 10^{-3}$; 
$x(t=0)$ is extracted from an appropriate chosen 
initial distribution (typically, either a $\delta(x)$ or 
$\propto e^{-x^2/(2 \sigma_0^2}$; averages are computed over
$1.6 \times 10^8$ independent stochastic trajectories, using OPENMPI for parallelization and RAN2 as basic number generator.

\subsection{Numerical Schemes for Non-Local Fokker-Planck like Equations}
The probability density $P(x,t)$ evolves according to a non-local partial differential equation of the form:
\begin{equation}
\partial_t P(x,t) = \partial_x [C(x) P(x,t)] + \hat{\mathcal{L}} P(x,t)
\end{equation}
where $C(x) $ represents a local drift velocity and $\hat{\mathcal{L}}$ is dealt with as a non-local operator generating jumps modulated by a spatial function $I(x)$.


All three jump models utilize a common numerical architecture to handle the state-dependent non-locality of $\hat{\mathcal{L}}$ 
while ensuring conservation of probability mass.
A central identity used in the discretization for all jump models is the interpretation of the operator $\hat{\mathcal{D}} = \partial_x I(x)$ as the generator of a density-conserving transformation. Unlike a simple translation operator $e^{s\partial_x}$, which merely shifts a function by a constant distance, the operator $e^{s \partial_x I(x)}$ accounts for a state-dependent velocity field $I(x)$. 

\subsubsection*{Formal Derivation of the operator via the Continuity Equation}
The action of the operator can be understood by considering the evolution of a density $P(x,s)$ under the virtual "time" parameter $s$:
\begin{equation}
\partial_s P(x,s) = \partial_x [I(x) P(x,s)]
\end{equation}
This is a one-dimensional Liouville or continuity equation. The formal solution at $s$ is given by $P(x,s) = e^{s \partial_x I(x)} P(x,0)$. Using the method of characteristics, we define the mapping $x = \phi_s(x_0)$ through the autonomous flow:
\begin{equation}
\frac{dx}{ds} = I(x), \quad x(0) = x_0
\end{equation}
The conservation of probability mass requires that $P(x,s)dx = P(x_0,0)dx_0$. Consequently, the action of the exponential operator is defined by the push-forward:
\begin{equation}
e^{s \partial_x I(x)} P(x) = P(x_{\text{dest}}(-s)) \left| \frac{\partial x_{\text{dest}}(-s)}{\partial x} \right|
\label{eq:coreshift1}
\end{equation}
For a one-dimensional autonomous flow, the Jacobian of the transformation $J = |\partial x_{\text{dest}} / \partial x_0|$ can be computed explicitly. By differentiating the characteristic equation with respect to $x_0$ and applying the chain rule, we have that:
\begin{equation}
J = \frac{I(x_{\text{dest}})}{I(x_0)}
\label{eq:coreshift2}
\end{equation}
This ratio represents the local stretching or compression of the phase space volume. If $I'(x) > 0$, the velocity increases with $x$, leading to a "dilution" of the density at the destination; conversely, $I'(x) < 0$ leads to a concentration of probability. This approach ensures that even for complex, non-linear $I(x)$ profiles, the total probability $\int P(x) dx$ remains constant to within machine precision, as any spatial "thinning" of the density is perfectly compensated by the Jacobian scaling~\cite{CROUSEILLES20101927}.

In the numerical implementation, Eqs.~(\ref{eq:coreshift1}) and~(\ref{eq:coreshift2}) allows the non-local integral to be treated as a sum of discrete, weighted transformations. 

\subsubsection*{Dichotomous (Hyperbolic Cosine) Jumps}
In this case, the particles undergo discrete jumps of magnitude $a$ in either direction. The operator is (see Eq.~\eqref{MEGx_dicho}):
\begin{equation}
\hat{\mathcal{L}}_{\text{dico}} = \frac{1}{\tau} \left[ \cosh(a \partial_x I(x)) - 1 \right] = \frac{1}{2\tau} \left[ e^{a \partial_x I(x)} + e^{-a \partial_x I(x)} - 2 \right]
\end{equation}
This is implemented by calculating two discrete advective translations at $s = \pm a$,
i.e.
\begin{align}
\hat{\mathcal{L}}_{\text{dico}} P(x,t) = \frac{1}{2\tau} &\left[ e^{a \partial_x I(x)} + e^{-a \partial_x I(x)} - 2 \right] P(x,t) \nonumber \\
= \frac{1}{2\tau} &\left[  P(x_{\text{dest}}(-a))\frac{I(x_{\text{dest}}(-a))}{I(x)}+\right. \nonumber \\
&\;\;\; \left. P(x_{\text{dest}}(a))\frac{I(x_{\text{dest}}(a))}{I(x)} - 2 P(x,t) \right]
\label{MEGx_dichoNum}
\end{align}
Eq.~\eqref{MEGx_dichoNum} has also an elegant theoretical justification,  obtained using the following formal equalities based on Lie evolution:
\begin{align}
e^{\partial_x I(x) a} P(x\pv t)
& =e^{\partial_x I(x)^\times a}\left[P(x\pv t)\right]e^{\partial_x I(x) a} \nonumber \\
&=P(x_I(x,-a)\pv t)e^{\partial_x I(x) a} \nonumber \\
&=P(x_I(x,-a)\pv t)\left(e^{\partial_x I(x) a}I(x)\frac{1}{I(x)}\right) \nonumber   \\  
& =P(x_I(x,-a)\pv t)e^{\partial_x I(x)^\times a}\left[I(x)\right]
\left(e^{\partial_x I(x) a}
\frac{1}{I(x)}\right) \nonumber  \\
&=P(x_I(x,-a)\pv t)I(x_I(x,-a)) \frac{1}{I(x)} \nonumber 
\end{align}

\subsubsection*{Gaussian Jump Operator}
The Gaussian case assumes a jump distribution following a normal profile. The operator is (see Eq.~\eqref{MEGx_gauss_}):
\begin{equation}
\hat{\mathcal{L}}_{\text{gauss}} = \frac{1}{\tau} \left[ e^{\frac{a^2}{2} (\partial_x I(x))^2} - 1 \right]
\end{equation}
This operator can be approximated via an $n$-point Gauss-Hermite quadrature~\cite{Abramowitz2012-qo}. The nodes $s_k$ and weights $w_k$ sample the Gaussian distribution to approximate the operator as a sum of discrete shifts:
\begin{equation}
\hat{\mathcal{L}}_{\text{gauss}} P \approx \frac{1}{\tau} \left[ \left( \sum_{k=1}^{n} w_k e^{s_k a \partial_x I(x)} P \right) - P \right]
\end{equation}
\subsubsection*{Uniform (Flat) Jump Operator}
The operator corresponds to a uniform jump distribution over the interval $[-\sqrt{3}a, \sqrt{3}a]$ (see Eq.~\eqref{MEGx_flat_}). In practice, the operator is rewritten as
\begin{equation}
\hat{\mathcal{L}}_{\text{flat}} = \frac{1}{\tau} \left[ \frac{\sinh(\sqrt{3} a \partial_x I(x))}{\sqrt{3} a \partial_x I(x)} - 1 \right] = \frac{1}{\tau} \left[ \frac{1}{2} \int_{-1}^{1} e^{s \sqrt{3} a \partial_x I(x)} ds - 1 \right]
\end{equation}
The integral is solved numerically using a $n$-point Gauss-Legendre quadrature~\cite{Abramowitz2012-qo}:
\begin{equation}
\hat{\mathcal{L}}_{\text{flat}} P \approx \frac{1}{\tau} \left[ \left( \sum_{k=1}^{n} w_k e^{s_k \sqrt{3} a \partial_x I(x)} P \right) - P \right]
\end{equation}
Again, each term in the quadrature sum represents a weighted advective translation.

\subsubsection*{Implementation Details}

The three cases above yield PDE's which are flux conservative ones: the typical numerical approach to integrate them would be a Lax Wendroff numerical scheme. On the other hand,
if one Taylor expands the operator for small $a$'s, the first order term $\partial_x I(x)$ vanishes: this suggests that also a Crank Nicolson numerical scheme could be worth 
investigating for the numerical integration of the PDE's. Furthermore, 
        it should be appreciated that an approach similar to the one used to deal with the Jump Operator can be applied for the advective term: the idea is that the PDE
        $ \partial_t P(x,t) = \partial_x C(x) P(x,t)$
        can be integrated using the characteristics method, moving forward the 
        probability distribution found in each grid point. This is the approach used for the 
        numerical results presented in this paper.
        
        With reference to Eq.~\eqref{ME_fin_Pois}, assuming we want to propagate
        $P(x,t)$ from $t$ to $t+h$, the  actual numerical integration was done with the
        following procedure (Strang splitting~\cite{Strang1968}, scheme $S_k^{(5)}$, where $L_x = \partial_x C(x)$ and 
        $L_y = \hat{\mathcal{L}}$):\begin{itemize}
            \item $\partial_t P(x,t) = \partial_x [C(x) P(x,t)]  P(x,t)$ was integrated
            for a time $h/2$ to get $P_{1/2}(x,t+h/2)$
            \item $\partial_t P(x,t) =  \hat{\mathcal{L}} P(x,t+h/2)$ was 
            integrated from $P_{1/2}(x,t)$ for a time $h$ using an Euler step to obtain a tentative $P^*(x,t+h)$ applying the operator to $P_{1/2}(x,t+h/2)$
            \item $\partial_t P(x,t) =  \hat{\mathcal{L}} P(x,t)$ was 
            integrated again from $P_{1/2}(x,t+h/2)$ for a time $h$ using an Euler step but applying the operator to $P^*(x,t+h)$ to get $P^{**}(x,t+h)$
            \item Heun's step: the average of $P^*(x,t+h)$ and $P^{**}(x,t+h)$ was taken, to get
            $P^{***}(x,t+h)$ 
            \item $\partial_t P(x,t) = \partial_x [C(x) P(x,t)]  P(x,t)$ was integrated
            for a time $h/2$ starting from     $P^{***}(x,t+h)$ to finally get $P(x,t+h)$
        \end{itemize}
        For the time dependent problem like the one in Eq.~\eqref{ME_fin} the procedure 
        follows immediately from the algorithm just described.
The actual integration is done with these details:
\begin{description}
\item[Characteristic curve mapping:] For any given jump 
magnitude $s$ the jump destinations are computed using a 4th-order Runge-Kutta (RK45) scheme to account for the spatial variation of $I(x)$ or $C(x)$ when 
            an analytical solution is not available. 
\item [Jacobian factor:] The density value is updated using the pre-calculated Jacobian factor $J= I(x_{\text{dest}}) / I(x_0)$ or $J= C(x_{\text{dest}}) / C(x_0)$.
\item [Interpolation:] Given that the jump 
destination $x_{\text{dest}}$ is typically off-lattice, a linear interpolation distributes probability mass to the two nearest computational grid points. Destination points outside the range of spatial integration are discarded.
\end{description}
Numerical stability is also further enforced through:
\begin{itemize}
\item Zero-Thresholding: Values below $10^{-20}$ in the $P(x,t)$ are truncated to zero; values of $I(x)$ below $10^{-10}$  are truncated to $10^{-10}$.
\item Normalization: The total probability is monitored via the trapezoidal rule and re-normalized to unity at each output step to prevent numerical drift.
\end{itemize}

\section{Conclusions\label{sec:conclusions}}
Numerous systems, such as climate dynamics (burst-like atmospheric forcing 
of slow modes), computational neuroscience (synaptic shot noise), materials 
science (intermittent crack growth), and transport phenomena with trapping 
and release events, exhibit intermittent, non-Gaussian forcing with memory. 
Deterministic drift, temporal correlations, and state-dependent responses 
interact to shape the observable dynamics and statistics in each of these situations.

A unified and flexible framework for describing such phenomena is provided by 
the stochastic model in Eq.~\eqref{SDE}: the fundamental components of 
non-Markovian dynamics driven by discrete, impulsive events are captured in a 
simple formulation by mixing deterministic drift with multiplicative renewal spike
noise. 

In this paper we have presented an exact and unified framework linking  this 
renewal spike/shot-noise to the statistical property of the variable of interest $x$ of the model \eqref{SDE}. 

Our first main result concerns the noise itself. It is a very general, 
closed-form 
expression for the $n$-time correlation functions of the 
spike/shot-renewal processes, 
formulated as a sum over all ordered partitions of the observation times 
(Eq.~\eqref{corrGen_delta} or Eq.~\eqref{corrGen_delta_pre} for a more
compact form). This formulation highlights the role of renewal 
blocks, the rate function $R(t)$, and the jump moments, while interfacing naturally 
and efficiently with the $G$-cumulant formalism.

Our second contribution is the derivation, via $G$-cumulants, of an exact and 
non-asymptotic ME for the PDF of the variable of interest $x$ of the
model~\eqref{SDE}, valid for arbitrary WT distributions and jump-amplitude 
PDFs with finite moments (Eq.~\eqref{ME_fin}). 
To the best of our knowledge, this is the first
time an ME of such generality has been obtained. The equation captures the full
coupling between drift and intermittent forcing, demonstrating that the familiar
additive decomposition of deterministic and stochastic currents--which 
is exact in the Poissonian case--breaks down under heavy-tailed 
renewal statistics.

Notably, written in the interaction representation, the structure of the
general exact ME~\eqref{ME_fin} is identical to that of the standard
Montroll-Weiss-Scher 
result~\citep{Montroll_Weiss_1965,Scher_Motroll_1975},
which applies to the simple drift-less and additive-noise CTRW case
(see Section~\ref{sec:Gcumulants_x}).

Most importantly, this exact nonlocal ME collapses at long times onto a
\emph{universal local-in-time structure} (Eq.~\eqref{ME_universal_intro})
parameterized by the \textit{instantaneous} renewal rate $R(t)$. This
structural simplification is a central result of this work; it is not a
phenomenological closure but follows systematically from the
$G$-cumulant formulation and preserves the nontrivial coupling between
drift and renewal statistics. The universal ME in
Eq.~\eqref{ME_universal_intro} recovers the exact Poissonian result when
$R(t) =\tau^{-1}$ and, by virtue of the Blackwell renewal theorem
\citep{Blackwell1948,FellerVol2}, converges to the same Poissonian limit
whenever the mean waiting time $\tau$ is finite. Furthermore, it
quantitatively captures the progressive weakening of stochastic forcing
for power-law waiting times with $1 < \mu < 2$. Extensive numerical
tests show that this local form reproduces the full nonlocal dynamics of
Eq.~\eqref{ME_fin} with remarkable accuracy, even for short times and
well beyond the regime of separation of time scales invoked by
Proposition~\ref{prop:approxME} to justify it.

Moreover, this simple and local-in-time ME structure allows standard
analytical treatments to gain insight into the PDF, such as the
eigenvector/eigenvalue technique or local/asymptotic analysis. Some
examples of these approaches are presented in Appendix~\ref{app:ME_resuls}.

This framework therefore provides a straightforward route for deriving macroscopic 
evolution equations directly from renewal microdynamics. It is immediately applicable 
to systems driven by intermittent, state-dependent impulses, such as impulsively 
forced climate modes, synaptic shot noise in computational neuroscience, and fatigue 
or crack-growth processes in materials science.

In the present work, we considered the standard setting in which the jump-amplitude 
and waiting-time distributions factorize. However, the extension to more general 
cases where these two random variables are statistically dependent is conceptually 
straightforward. Future research may address multidimensional systems, heavy-tailed 
jump amplitudes, and the data-driven inference of renewal statistics, all of which 
represent promising directions for further development.

\section*{Data availability}
Data are available upon reasonable request.

\section*{Acknowledgments}
We thank the Green Data Center of University of Pisa for providing the computational power needed for the present paper. 

\section*{Funding}
This work was supported by “National Centre for HPC,
Big Data and Quantum Computing,” under the National Recovery and Resilience Plan 
(NRRP), Mission 11 4 Component 2 Investment 1.4 funded from the European Union – 
NextGenerationEU. 

\section*{Appendices}
\appendix
\section{The \mbox{$n$}-time correlation function for the spike noise with renewal\label{app:spikes}}

In this appendix we start from the definition of the multi-time joint
correlation function given in Eq.~\eqref{n_corr_def_2_delta}, with the
aim of deriving a general analytical expression. Since the procedure is
conceptually straightforward but algebraically cumbersome, we first
illustrate the main steps in the simplified case $n=4$, assuming that
the odd moments of $\xi$ vanish (i.e., $p(\xi)$ is an even function).
        
        \subsection{The rigorous derivation of the four-time correlation function for the renewal spike noise\label{app:4-time}}
Starting from the definition of the multi-time correlation function given in
Eq.~\eqref{n_corr_def_2_delta}, we rigorously derive the same result as in
Eq.~\eqref{cor4t_delta} for the four-time correlation function of $\xi[t]$. 

For $n=4$ and assuming the times are ordered, i.e., $t_i\le t_{i+1}$, for any $i\in [1,n-1]$, it is convenient to rewrite
Eq.~\eqref{n_corr_def_2_delta} as follows (we set $t_0 = 0$ and we define $\theta_0=0$):
\begin{align}
\label{4_corr_def_2_delta}
&  \langle\xi( t_1)\xi( t_2)\xi( t_3)\xi( t_4)\rangle\nonumber \\
&=
\int  \sum_{i=0}^{\infty}\;      \sum_{j=i}^{\infty}\sum_{l=j}^{\infty}\sum_{m=l}^{\infty} 
\xi_{i}  \delta\left(t_1-\sum_{k_1=0}^i\theta_{k_1}\right)      \nonumber \\
&\times
\xi_{j}\delta\left(t_2-\sum_{k_1=0}^i\theta_{k_1}-\xsum_{k_2=i+1}^j\theta_{k_2}\right)          \nonumber \\
&\times
\xi_{l}\delta\left(t_3-\sum_{k_1=0}^i\theta_{k_1}-\xsum_{k_2=i+1}^j\theta_{k_2}-\xsum_{k_3=j+1}^l\theta_{k_3}\right) \nonumber \\
&\times
\xi_{m}\delta\left(t_4-\sum_{k_1=0}^i\theta_{k_1}-\xsum_{k_2=i+1}^j\theta_{k_2}-\xsum_{k_3=j+1}^l\theta_{k_3}-\xsum_{k_4=l+1}^m\theta_{k_4}\right)
\nonumber \\
&\times
p_0(\xi_0)\text{d}\xi_0\prod_{q=1}^\infty\psi(\theta_{\!q})\text{d}\theta_{\!q}\,p(\xi_q)\text{d}\xi_q
\end{align}
where we continue adopting the convention that a
primed summation is zero if the upper limit is smaller than the lower one. As we have already observed, 
this condition can occur when two or more times are equal.
The first Dirac delta function, namely
$\delta\!\left(t_1 - \sum_{k_1 = 0}^i \theta_{k_1}\right)$, when
multiplied by the product of the WT PDFs associated with the intervals
$\theta_{k_1}$ and integrated over these waiting times, yields the
probability that the $i$th transition event occurs at time $t_1$, as
defined in Eq.~\eqref{psi_i}. Therefore, after integrating over
$\theta_1,\theta_2...,\theta_i$ we get
\begin{align}
\label{4_corr_def_2_delta_3}
&  \langle\xi( t_1)\xi( t_2)\xi( t_3)\xi( t_4)\rangle\nonumber \\
&=
\int  \sum_{i=0}^{\infty}\;      \sum_{j=i}^{\infty}\sum_{l=j}^{\infty}\sum_{m=l}^{\infty} 
\xi_{i}  \psi_{i}(t_1)
\nonumber \\
&\times
\xi_{j}\delta\left(t_2-t_1-\xsum_{k_2=i+1}^j\theta_{k_2}\right)           \psi(\theta_{i+1})...\psi(\theta_{j}) 
\text{d}\theta_{i+1}...\text{d}\theta_{j}  \nonumber \\
&\times
\xi_{l}\delta\left(t_3-t_1-\xsum_{k_2=i+1}^j\theta_{k_2}-\xsum_{k_3=j+1}^l\theta_{k_3}\right) 
\psi(\theta_{j+1})...\psi(\theta_{l}) \text{d}\theta_{j+1}...\text{d}\theta_{l} \nonumber \\
&\times
\xi_{m}\delta\left(t_4-t_1-\xsum_{k_2=i+1}^j\theta_{k_2}-\xsum_{k_3=j+1}^l\theta_{k_3}-\xsum_{k_4=l+1}^m\theta_{k_4}\right)
\nonumber \\
&\times           \psi(\theta_{l+1})...\psi(\theta_{m}) \text{d}\theta_{l+1}...\text{d}\theta_{m}
\nonumber \\
&\times
p_0(\xi_0)\text{d}\xi_0\prod_{q=1}^\infty p(\xi_q)\text{d}\xi_q.
\end{align}

As before, the first Dirac delta in
Eq.~\eqref{4_corr_def_2_delta_3} and the related WT PDF, when integrated over the corresponding waiting times, gives the probability that a the time $t_2-t_1$ 
we have exactly the $j-i$th transition event. Thus, we can write
\begin{align}
\label{4_corr_def_2_delta_5}
&  \langle\xi( t_1)\xi( t_2)\xi( t_3)\xi( t_4)\rangle\nonumber \\
&=
\int  \sum_{i=0}^{\infty}\;      \sum_{j=i}^{\infty}\sum_{l=j}^{\infty}\sum_{m=l}^{\infty} 
\xi_{i}  \psi_{i}(t_1)
\nonumber \\
&\times
\xi_{j}\psi_{j-i}\left(t_2-t_1\right)  \nonumber \\
&\times
\xi_{l}\delta\left(t_3-t_2-\xsum_{k_3=j+1}^l\theta_{k_3}\right)
\psi(\theta_{j+1})...\psi(\theta_{l}) \text{d}\theta_{j+1}...\text{d}\theta_{l}\nonumber \\
&\times
\xi_{m}\delta\left(t_4-t_2-\xsum_{k_3=j+1}^l\theta_{k_3}-\xsum_{k_4=l+1}^m\theta_{k_4}\right)
\psi(\theta_{l+1})...\psi(\theta_{m}) \text{d}\theta_{l+1}...\text{d}\theta_{m}
\nonumber \\
&\times
p_0(\xi_0)\text{d}\xi_0\prod_{q=1}^\infty p(\xi_q)\text{d}\xi_q.
\end{align}
Repeating this procedure for all intermediate variables $\theta$, we obtain
\begin{align}
\label{4_corr_def_2_delta_6}
&  \langle\xi( t_1)\xi( t_2)\xi( t_3)\xi( t_4)\rangle\nonumber \\
&=
\sum_{i=0}^{\infty}\;      \sum_{j=i}^{\infty}\sum_{l=j}^{\infty}\sum_{m=l}^{\infty} 
\psi_{i}(t_1)
\psi_{j-i}\left(t_2-t_1\right) 
\psi_{l-j}\left(t_3-t_2\right)
\psi_{m-l}\left(t_4-t_3\right)  \nonumber \\
&\times
\int    \xi_{i}         \xi_{j} \xi_{l} \xi_{m}         p_0(\xi_0)\text{d}\xi_0\prod_{q=1}^\infty p(\xi_q)\text{d}\xi_q.
\end{align}

To be consistent with
Eq.~\eqref{cor4t_delta}, we now implement the assumption that all odd moments of $\xi$
vanish.
When averaging over the four $\xi$ variables in
Eq.~\eqref{4_corr_def_2_delta_6}, this assumption implies that the correlation is non-zero
only in two cases: either $i = j = l = m$, or $i = j$, $j\ne l$ and $l = m$. These cases correspond
to the following time coincidences (we recall that we have defined $\psi_0(t):=\delta(t)$):
\begin{itemize}
\item $t_1=t_2=t_3=t_4$,
\item $t_1=t_2$, $t_2\ne t_3$ and $t_3=t_4$,
\end{itemize}   
respectively. Thus we get:
\begin{align}
\label{4_corr_def_2_delta_7}
&  \langle\xi( t_1)\xi( t_2)\xi( t_3)\xi( t_4)\rangle\nonumber \\
&=      \left(\overline{\xi_0^4}\,\delta(t_1)+\overline{\xi^4}\sum_{i=1}^{\infty}       \psi_{i}(t_1)\right)
\delta\left(t_2-t_1\right)  
\delta\left(t_3-t_2\right)  
\delta\left(t_4-t_3\right)  \nonumber \\
&+
\left(\overline{\xi_0^2}\,\delta(t_1)+\overline{\xi^2}\sum_{i=1}^{\infty}       \psi_{i}(t_1)\right)
\delta\left(t_2-t_1\right)  
\overline{\xi^2}\sum_{l=i+1}^{\infty}   \psi_{l-i}\left(t_3-t_2\right)
\delta\left(t_4-t_3\right)      
\end{align}

Finally, by using the definition of the rate function $R$ of Eq.~\eqref{Rtilde}, we obtain
\begin{align}
\label{cor4t_delta_2}
&\langle\xi( t_1)\xi( t_2)\xi( t_3)\xi( t_4)\rangle\nonumber \\
&       =       \left[ \overline{\xi_0^4} \,\delta(t_1)+ \overline{\xi^4}\, R(t_1)\right]  
\delta(t_2-t_1) \delta(t_3-t_2) \delta(t_4-t_3)
\nonumber\\
& 
+\overline{\xi^2}\left[\overline{\xi_0^2}\,\delta(t_1)
+\overline{\xi^2}\,R(t_1) \right] \delta(t_2-t_1) R(t_3-t_2) \delta(t_4-t_3).
\end{align}
Observing that here we have set $t_0=0$, Eq.~\eqref{cor4t_delta_2} is the same as   Eq.~\eqref{cor4t_delta},
ending the demonstration.

We now note that, according to  
the definition of the primed summation in  
Eqs.~\eqref{4_corr_def_2_delta}--\eqref{4_corr_def_2_delta_5}, when $t_3 = t_2$
the second term on the right-hand side of Eq.~\eqref{cor4t_delta_2}  
should be discarded.  On the other hand, if $t_3 = t_2 + \epsilon$, with  
$\epsilon$ small but nonzero,  the same term should be retained (while the first one is zero). In this case,  
the second term approximately takes the form:
\begin{align}
\label{epsilon0}
&\overline{\xi^2} \left[\overline{\xi_0^2}\,\delta(t_1) \overline{\xi^2}\,R(t_1) \right] 
\delta(t_2 - t_1) R(t_3 - t_2) \delta(t_4 - t_3) \nonumber \\
\approx\,&\overline{\xi^2} \left[\overline{\xi_0^2}\,\delta(t_1) \overline{\xi^2}\,R(t_1) \right] 
\delta(t_2 - t_1) R(0) \delta(t_4 - t_3) \ne 0.
\end{align}
The final inequality in \eqref{epsilon0} holds because $R(0) = \psi(0)$, which is nonzero. 

Therefore, from a  
rigorous standpoint, the two terms on the right-hand side of Eq.~\eqref{cor4t_delta_2}  
should not be added together. Instead, they should be treated as mutually exclusive: if all  
time arguments exactly coincide, only the first term is retained; otherwise, only the second term  
should be kept.  

However, the first term is nonzero only when $\epsilon = 0$, i.e., for $t_3 = t_2$, where its  
value is infinite. In contrast, the second term remains finite for $\epsilon=0$ at is apparent
from \eqref{epsilon0}, thus it is negligible with respect to the first one.
Therefore, in Eq.~\eqref{cor4t_delta_2}   we may include both  
terms regardless of the relationship between $t_3$ and $t_2$ without affecting the final  
result.

          \subsection{The general case\label{sec:n-time}}     
        
        For the reader's convenience, we rewrite the  multi-time joint
correlation function given in Eq.~\eqref{n_corr_def_2_delta} as:
{\small 
\begin{align}
\label{app:n_corr_def_2_delta}
&  \langle\xi( t_1)\xi( t_2)...\xi( t_n)\rangle_{t_0}=
\!
\int  \sum_{i_1=0}^{\infty}\;    \sum_{i_2=i_1}^{\infty} \dots\sum_{i_n=i_{n-1}}^{\infty} 
\xi_{i_1}  \xi_{i_2}\cdots \xi_{i_n} 
\; \delta\left(t_1-t_0-\sum_{k_1=0}^{i_1} \theta_{k_1}\right)\nonumber \\
&\times
\delta\left(t_2-t_0-\sum_{k_2=0}^{i_2} \theta_{k_2}\right)
\times...
\delta\left(t_n-t_0-\sum_{k_n=0}^{i_n} \theta_{k_n}\right)
p_0(\xi_0)\text{d}\xi_0        \nonumber \\&\times 
\prod_{q=1}^\infty\psi(\theta_{\!q})\text{d}\theta_{\!q}\,p(\xi_q)\text{d}\xi_q
\end{align}
}
As in the previous $n=4$ case, one could collect the WT PDFs in groups of
``$i$'' elements, corresponding to the sets of ``$i$'' waiting times
involved in each Dirac delta function, and then perform the integrations
over these waiting times. This yields $\psi_i(\theta)$, i.e., the WT PDFs associated
with the $i$ jumps. However, here we adopt an equivalent but alternative procedure:
we perform variable transformations that allow us to move from the waiting
times (i.e., the time intervals between successive events) to the
absolute times of the events.The first   Dirac-delta function in Eq.~\eqref{app:n_corr_def_2_delta}
says that the first event, after the initial time $t_0$, happens at the time 
$u_{i_{1}}=\theta_0+\theta_1+\theta_2+...+\theta_{i_{1}}$ (note: $\theta_0:=0$). Thus, we perform
the  change of variable 
$\theta_1\to u_{i_{1}}=\theta_0+\theta_1+\theta_2+...+\theta_{i_{1}}$, i.e., $\theta_{1}=u_{i_{1}}-\xsum_{k_1=2}^{i_1}\theta_{k_1}$,
where the summation with a prime symbol indicates that it equals zero when the upper  
limit is less than the lower limit. 

Doing that, the multiple integral of Eq.~\eqref{n_corr_def_2_delta} that involves the times $\theta_2,\theta_3, ...,\theta_{i_{1}}$ gives: 
{\small 
\begin{align}
\label{psi_i_2}
&
\int_0^{u_{i_{1}}} \text{d}\theta_2
\int_0^{u_{i_{1}}-\theta_2} \text{d}\theta_3
\int_0^{u_{i_{1}}-\theta_2-\theta_3} \text{d}\theta_4 \dots
\int_0^{u_{i_{1}}-\theta_2 \dots -\theta_{{i_1}-1}} \text{d}\theta_{i_1} \nonumber \\
&\psi(u_{i_{1}}-\theta_2-\theta_3-\theta_4- \dots -\theta_{{i_1}})
\psi(\theta_2)\psi(\theta_3)\psi(\theta_4) \dots \psi(\theta_{{i_1}}):=\psi_{i_1}(u_{i_{1}}).
\end{align}
}
%
It is not difficult to realize that  
Eq.~\eqref{psi_i_2} is an alternative way to write the  ${i_1}$-fold convolution of the
WT PDF, here evaluated at the time $u_{i_{1}}$. In other word,   $\psi_{i_1}(u_{i_{1}})$
of the r.h.s. of 
Eq.~\eqref{psi_i_2} is the same function already introduced in Eq.~\eqref{Rtilde}.
Note that setting $\psi_{0}(\theta):=0$, as already done in Section~\ref{sec:preliminary},
we include also the case where $i_1=0$. 
Thus,
after this change of variable, Eq.~\eqref{n_corr_def_2_delta} becomes
{\small 
\begin{align}
\label{n_corr_def_2_delta_2}
&  \langle\xi( t_1)\xi( t_2)...\xi( t_n)\rangle_{t_0}
=
\int  \sum_{i_1=0}^{\infty}\;    \sum_{i_2=i_1}^{\infty} \dots\sum_{i_n=i_{n-1}}^{\infty} 
\xi_{i_1}  \xi_{i_2}\cdots \xi_{i_n} 
\psi_{i_1}(u_{i_{1}})\delta\left(t_1-t_0-u_{i_{1}}\right)\nonumber \\
&\times
\delta\left(t_2-t_0-u_{i_{1}}-\xsum_{k_2=i_1+1}^{i_2} \theta_{k_2}\right)\times \cdots\nonumber \\
&
\cdots\times \delta\left(t_n-t_0-u_{i_{1}}-
\xsum_{k_2=i_1+1}^{i_2} \theta_{k_2}-...-\xsum_{k_n=i_{n-1}+1}^{i_n} \theta_{k_n}\right)
\nonumber \\ &
\times \text{d}u_{i_{1}}\,\left(\prod_{r=i_1+1}^\infty\psi(\theta_r)\text{d}\theta_r\,\right)
p_0(\xi_0)\text{d}\xi_0\, \prod_{q=1}^\infty p(\xi_q)\text{d}\xi_q.
\end{align}
}
%
%
where the summation with a prime symbol indicates that it equals zero when the upper  
limit is less than the lower limit. 
Now we consider the second Dirac-delta function and we repeat the same change of variables concerning
the sum of waiting times in its argument, i.e., $\theta_{i_1+1}
=u_{i_{2}}-\xsum_{k_2=i_1+2}^{i_2}\theta_{k_2}$ and  Eq.~\eqref{n_corr_def_2_delta_2} reads
{\small 
\begin{align}
\label{n_corr_def_2_delta_2_}
&  \langle\xi( t_1)\xi( t_2)...\xi( t_n)\rangle_{t_0}
=\int  \sum_{i_1=0}^{\infty}\;  \sum_{i_2=i_1}^{\infty} 
\dots\sum_{i_n=i_{n-1}}^{\infty} \xi_{i_1}  \xi_{i_2}\cdots \xi_{i_n} 
\psi_{i_1}(u_{i_{1}})\delta\left(t_1-t_0-u_{i_{1}}\right)
\nonumber \\&\times
\psi_{i_2-i_1}(u_{i_{2}})
\delta\left(t_2-t_0-u_{i_{1}}-u_{i_{2}}\right)
\cdots
\delta\left(t_n-t_0-u_{i_{1}}-u_{i_{2}}-...-\xsum_{k_n=i_{n-1}+1}^{i_n} \theta_{k_n}\right)
\nonumber \\ 
&\times \text{d}u_{i_{1}}\text{d}u_{i_{2}}\left(\prod_{r=i_2+1}^\infty\psi(\theta_r)\text{d}\theta_r\,\right)p_0(\xi_0)\text{d}\xi_0\, \prod_{q=1}^\infty p(\xi_q)\text{d}\xi_q
\end{align}
} 
%
%
By applying the same procedure to the waiting times of all the other Dirac-delta functions, we get
{\small 
\begin{align}
\label{n_corr_def_2_delta_n}
&  \langle\xi( t_1)\xi( t_2)...\xi( t_n)\rangle_{t_0}
=\int  \sum_{i_1=0}^{\infty}\;  \sum_{i_2=i_1}^{\infty} \dots\sum_{i_n=i_{n-1}}^{\infty} \xi_{i_1}  \xi_{i_2}\cdots
\xi_{i_n} \psi_{i_1}(u_{i_{1}})\delta\left(t_1-t_0-u_{i_{1}}\right)
\nonumber \\&\times
\psi_{i_2-i_1}(u_{i_{2}})
\; \delta\left(t_2-t_0-u_{i_{1}}-u_{i_{2}}\right)
\cdots
\nonumber \\
&\cdots \times
\psi_{i_n-i_{n-1}}(u_{i_{n}})\delta\left(t_n-t_0-u_{i_{1}}-u_{i_{2}}-...-u_{i_{n-1}}-u_n\right)
&\nonumber \\ &\times 
\text{d}u_{i_{1}}\text{d}u_{i_{2}}...\text{d}u_{i_{n}}
p_0(\xi_0)\text{d}\xi_0\, \prod_{q=1}^\infty p(\xi_q)\text{d}\xi_q
\end{align}
}

By exploiting the property of the Dirac-delta function, the  integrals with respect to the absolute times $u_1, u_2,...,u_n$  can be easily performed, giving:
{\small 
\begin{align}
\label{n_corr_def_2_delta_n_1}
&  \langle\xi( t_1)\xi( t_2)...\xi( t_n)\rangle_{t_0}
=\int  \sum_{i_1=0}^{\infty}\;   \sum_{i_2=i_1}^{\infty} \dots\sum_{i_n=i_{n-1}}^{\infty} 
\xi_{i_1}  \xi_{i_2}\cdots
\xi_{i_n} \
\psi_{i_1}(t_1-t_0) \nonumber \\ 
&\times \psi_{i_2-i_1}(t_2-t_1)
\cdots \psi_{i_n-i_{n-1}}(t_n-t_{n-1})
p_0(\xi_0)\text{d}\xi_0\, \prod_{q=1}^\infty p(\xi_q)\text{d}\xi_q
\end{align}
}
%
%
We recall that $\psi_{i_k}(t)$ is the PDF to have the $i_{k}+1$-th event at the time $t$ after 
the previous event. Therefore, each term of the sum of Eq.~\eqref{n_corr_def_2_delta_n_1} says that at the time
$t_k$ we have one event (the $(i_k+1)$-th)  of $\xi$ with intensity $\xi_{i_k}$ (this is in agreement
with points 1 and 2 of Section~\ref{sec:corr_leapers_2-4}). 

Now, observing that

\begin{enumerate}
\item 
for${i_{k+1}=i_k}$ we have that $t_{k+1}=t_k$;
\item 
we have ``$i_1\le i_2\le ...\le i_n$''  thus  
$\xi_{i_1}  \xi_{i_2}\cdots \xi_{i_n}$ is a set of ordered factors;
\item 
from the multiple sum of Eq.~\eqref{n_corr_def_2_delta_n_1}
we see that the cases in which a set of $j$ indices are equals, i.e.,   ${i_k}= {i_{k+1}}=\dots= {i_{k+j}}$ (the
lower extreme  of each sum), we have $\xi_{i_{k}}\xi_{i_{k+1}}\dots\xi_{i_{k+j}}=\xi_{i_k}^j$, and $t_k=t_{k+1}=\dots=t_{k+j}$; 
\end{enumerate}
then, for any ``$\infty^n$'' terms  of the multiple sum Eq.~\eqref{n_corr_def_2_delta_n_1}, after averaging
over the fluctuations of $\xi$, we can rearrange the terms
partitioning  the indices  ``$i_1,i_2,...,i_n$''
in blocks 
$\left\{ m_{i_1}\right\}\left\{ m_{i_2}\right\}...\left\{ m_p\right\}$, where
in each block there is a number  $m_i$ of consecutive indices with the same value 
{\small 
\begin{align}
\label{compositions}
&   \left\{i_1,\,i_2=i_1,\,...,i_{m_{i_1}}=i_1\right\}
\left\{
i_{m_{i_1}+1},\,i_{m_{i_1}+2}=i_{m_{i_1}+1},...,i_{m_{i_1}+m_{i_2}}=i_{m_{i_1}+1}\right\}\,...
\nonumber \\[5pt]
&\,\left\{
\,i_{m_{i_1}+m_{i_2}+\dots+m_{p-1}+1},i_{m_{i_1}+m_{i_2}+\dots+m_{p-1}+2}=i_{m_{i_1}+m_{i_2}+\dots+m_{p-1}+1},...\right.
\nonumber \\
&\left. ...
i_{m_{i_1}+m_{i_2}+\dots+m_{p-1}+m_{p}}=i_{m_{i_1}+m_{i_2}+\dots+m_{p-1}+1}\right\}
\end{align}
}
leading to a factor $(\overline{\xi^{m_{i_1}}})(\overline{\xi^{m_{i_2}}})...(\overline{\xi^{m_p}})$
in each term of \eqref{n_corr_def_2_delta_n_1}. Of course,
there are 
$2^{n-1}$ such possible partitions (a block separator between any  position can be turned on or of). 

To simplify the notation,  we rewrite  \eqref{n_corr_def_2_delta_n_1} by exploiting the following change of variables:where $j_1:=i_1$ and $j_k=i_k-i_{k-1}$:
{\small 
\begin{align}
&  \langle\xi( t_1)\xi( t_2)...\xi( t_n)\rangle_{t_0}
\nonumber \\&
=
\int 
\sum_{j_1=0}^{\infty} \sum_{j_2=0}^{\infty} \cdots\sum_{j_n=0}^{\infty} 
\xi_{j_1}  \xi_{j_1+j_2}\cdots
\xi_{j_1+j_2+...+j_n} \psi_{j_1}(t_1-t_0)\psi_{j_2}(t_2-t_1)
\; \times...\nonumber \\&
..\times \psi_{j_n}(t_n-t_{n-1})
p_0(\xi_0)\text{d}\xi_0\, \prod_{q=1}^\infty p(\xi_q)\text{d}\xi_q
\label{n_corr_def_2_delta_n_1_bis}
\end{align}
}

Applied to Eq.~\eqref{n_corr_def_2_delta_n_1_bis}, the aforementioned partition corresponds
to  any composition of the set $j_1,j_2,j_3,...,j_n$ in blocks 
$\left\{ m_{1}\right\}\left\{ m_{2}\right\}...\left\{ m_p\right\}$, where
in each block there is an index $j_k$ different to zero followed by a number $m_i-1$ of consecutive indices $j$ put equal to zero:
{\small 
\begin{align}
&\{j_1,j_2=0,j_3=0,...,j_{m_{1}}=0\}\{j_{m_{1}+1}\ne0,j_{m_{1}+2}=0 ,j_{m_{1}+3}=0 ,\dots,j_{m_{1}+m_{2}}=0\}
\dots\nonumber \\
&\dots\{j_{\underbrace{m_1+m_{2}+\dots+m_{p-1}}_{=n-m_p}+1}\ne0,j_{m_1+m_{2}+\dots+m_{p-1}+2}=0,\dots
\nonumber\\
&\dots,
j_{\underbrace{m_1+m_{2}+\dots+m_{p-1}+m_p}_{=n}}=0\}
\end{align}
}
\noindent
in which $m_1+m_{2}+...+m_{p-1}+m_p=n$. Of course,
for any fixed number $p$ of blocks, there are  
$N(p)=\frac{\left(n-1\right)!}{(p-1) !\left[n-p\right] !}$) possible compositions. 
Thus the total number of compositions is still 
$\sum_{p=1}^{n} N(p)=2^{n-1}$. In this way,  Eq.~\eqref{n_corr_def_2_delta_n_1_bis}
becomes:
{\small 
\begin{align}
\label{corrGen_delta_temp_}
&  \langle\xi( t_1)\xi( t_2)...\xi( t_n)\rangle_{t_0}
=  
\sum_{p=1}^{ n} 
\sum_{\substack{\{m_{i}\}:\\ \sum_{i=1}^p m_{i}={ n}}}
\;\sum_{j_1=0}^{\infty} \,\sum_{j_{m_{1}+1}=1}^{\infty} \cdots\sum_{j_{n-m_p+1}=1}^{\infty} 
(\overline{\xi^{m_{1}}})(\overline{\xi^{m_{2}}})...(\overline{\xi^{m_p}})
\nonumber \\
&\times
\psi_{j_{1}}(t_{1}-t_{0})\psi_{0}(t_{2}-t_{1})
\dots\psi_{0}(t_{m_{1}}-t_{m_{1}-1}) \nonumber \\
&\times
\psi_{j_{m_{1}+1}}(t_{m_{1}+1}-t_{m_{1}})
\psi_{0}(t_{m_{1}+2}-t_{m_{1}+1})
\dots\psi_{0}(t_{m_{1}+m_2}-t_{m_{1}+m_2-1})\times...\nonumber
\\
& ...\times\psi_{j_{n-m_p+1}}(t_{n-m_p+1}-t_{n-m_p})\psi_{0}(t_{n-m_p+2}-t_{n-m_p+1})...\psi_{0}(t_{n}-t_{n-1})  
\end{align}
}

Renaming the indices $j_k$ of the last multiple sum as $j_1=l_1$,
$j_{m_{1}+1}=l_2$,...,$j_{n-m_p+1}=l_p$, we get%
{\small 
\begin{align}
\label{corrGen_delta_temp_2} 
&  \langle\xi( t_1)\xi( t_2)...\xi( t_n)\rangle_{t_0}=
\sum_{p=1}^{ n} 
\sum_{\substack{\{m_{i}\}:\\ \sum_{i=1}^p m_{i}={ n}}}
%
\;\sum_{l_1=0}^{\infty}\; \sum_{l_2=1}^{\infty} \cdots\sum_{l_p=1}^{\infty} 
%
(\overline{\xi^{m_{1}}})(\overline{\xi^{m_{2}}})...(\overline{\xi^{m_p}})
\nonumber \\&
\psi_{l_1}(t_{1}-t_{0})\psi_{0}(t_{2}-t_{1})
\dots \psi_{0}(t_{m_{1}}-t_{m_{1}-1})
\nonumber
\\
& \times
\psi_{l_2}(t_{m_{1}+1}-t_{m_{1}})\psi_{0}(t_{m_{1}+2}-t_{m_{1}+1})
\dots \psi_{0}(t_{m_{1}+m_2}-t_{m_{1}+m_2-1})\times\dots\nonumber
\\
& \dots\times\psi_{l_p}(t_{n-m_p+1}-t_{n-m_p})\psi_{0}(t_{n-m_p+2}-t_{n-m_p+1})...\psi_{0}(t_{n}-t_{n-1})  
\end{align}  
}
After the sum over all the $l_k$ indices, and considering Eq.~\eqref{Rtilde} and that $\psi_{0}(t)=\delta(t)$, Eq.~\eqref{corrGen_delta_temp_2} becomes:
{\small 
\begin{align}
\label{corrGen_delta_temp_3}
&  \langle\xi( t_1)\xi( t_2)...\xi( t_n)\rangle_{t_0}=
\sum_{p=1}^{ n} 
\sum_{\substack{\{m_{i}\}:\\ \sum_{i=1}^p m_{i}={ n}}}
%
(\overline{\xi^{m_{1}}})(\overline{\xi^{m_{2}}})...(\overline{\xi^{m_p}})
\nonumber \\&
\tilde R(t_{1}-t_{0})  \delta(t_{2}-t_{1})
\dots \delta(t_{m_{1}}-t_{m_{1}-1})\times\nonumber
\\
&R(t_{m_{1}+1}-t_{m_{1}})\delta(t_{m_{1}+2}-t_{m_{1}+1})
\dots\delta(t_{m_{1}+m_2}-t_{m_{1}+m_2-1})\times\dots\nonumber
\\
& \dots\times R(t_{n-m_p+1}-t_{n-m_p})\delta(t_{n-m_p+2}-t_{n-m_p+1})...\delta(t_{n}-t_{n-1})  
\end{align}  
}
that is equal to Eq.~\eqref{corrGen_delta}.\\

\section{Brief summary of the basic concepts of $M$-cumulants\label{app:M-cumulants}}
\subsection{The  definition and use of $G$-cumulants\label{sec:Gcumulants}}
A powerful method for deriving the  ME for the random variable $x$ in Eq.~\eqref{SDE}, particularly in the 
general case where the drift term $-C(x) $ is nonzero and the noise may be multiplicative (i.e., $I(x) \ne \text{const}$), involves 
fully exploiting the concept of $G$-cumulants, as developed in a series of previous works~\cite{bbJSTAT4, bCSF148, bCSF159}.  
Therefore, in this section, we shortly derive the $G$-cumulants of the  spike stochastic renewal process $\xi[t]$, 
and from this result, we will try to obtain the ME for the random variable  $x(t)$ of the SDE~\eqref{SDE}.

\subsection{$G$-cumulants for  a stochastic process $\xi[t]$\label{sec:Gcumulants_xi}}
Joint cumulants for a random variable are defined from the characteristic function (CF) of  the same process.  
The CF of a random variable is the Fourier transform
of the PDF of the variable. A stochastic process $\xi[u]$ can be viewed as a multivariate random variable whose components 
are indexed by a parameter, typically the time $u\in [0,t]$,
that ranges over a real interval rather than a discrete set of integers. Therefore,  its CF is defined as (we set $t_0=0$):
\begin{align}
\label{MG}
&       \hat{P}_{\xi[t]}(k(\cdot)\pv t)
:=\langle \exp\left[\text{i}\int_{0}^t \text{d}u\,k(u) \xi(u)\right]\rangle
\nonumber \\
&=\sum_{n=0}^\infty \frac{\I ^n}{n!}
\int^t_0{\text{d}u_n}\int^t_0{\text{d}u_{n-2}}\dots 
\int^t_0\text{d}u_1\,
k(u_1)k(u_2)\dots \times 
\nonumber \\
&       \dots \times k(u_n)
\langle \xi (u_1)\xi (u_2)\dots \xi (u_n)\rangle%
\end{align}
where $k(\cdot)$, the continuous version of the wave vector $\vec{k}$,  is any ``enough'' smooth function of $u$: $k(u)\in\R$ with $u\in[0,t]$. 
As it is apparent from the series expansion in Eq.~\eqref{MG},  the CF is also the generator function of the multi-time correlation functions (or multivariate moments).

Indicating with $\doubleangle{\xi (u_1)\xi (u_2)\dots \xi (u_n)}^{(G)}$ the $n$-time $G$-cumulant  of $\xi[t]$, the corresponding 
generator ${K}^{(G)}_{\xi[t]}(k(\cdot)\pv t)$ is then given by
\begin{align}
\label{CG}
&       {K }^{(G)}_{\xi[t]}(k(\cdot)\pv t)
:=\sum_{n=1}^\infty \frac{\I ^n}{n!}
\int^t_0{\text{d}u_n}\int^t_0{\text{d}u_{n-2}}\dots 
\int^t_0\text{d}u_1\,
k(u_1)k(u_2)\dots \times 
\nonumber \\
&       \dots \times k(u_n)
\doubleangle{\xi (u_1)\xi (u_2)\dots \xi (u_n)}^{(G)}
\end{align}

Correlation functions (or multivariate moments) and $G$-cumulants are related to each other by the following definition involving
the corresponding generator functions: 
(see~\cite[Section 4.4.3]{bbJSTAT4}):
\begin{equation}
\label{MG2}
\hat{{P }}_{\xi[t]}(k(\cdot)\pv t)=        \exp_{G}\left[{K}^{(G)}_{\xi[t]}(k(\cdot)\pv t)\right].
\end{equation}

Equation~\eqref{MG2} corresponds to the standard definition of the cumulant
generating function. The subscript $G$ appearing in the exponential on the
right-hand side indicates, however, that the \emph{total time ordering} (TTO)
projection--denoted by the letter $G$--must be applied to the argument of the
exponential itself.

The TTO is a stricter projection operator than the more commonly used
\emph{partial time ordering} (PTO) projection, denoted by the letter $O$, which
leads to the usual time-ordered exponential (often called the
$t$-exponential and denoted by $\overleftarrow{\exp}[\dots]$). The difference
between the two projections becomes relevant when they act on objects that
depend on at least two time variables.

To illustrate this point, let ${\cal A}(t,u)$ be a function of two time
arguments, $t$ and $u$, taking values in a space of mutually non-commuting
operators. Denoting by $\{ {\cal U} \}_M$ the result of applying a generic map
$M$ to ${\cal U}$, we have
\begin{equation}
\label{exR=O}
\left\{ {\cal A}(t,u)\,{\cal A}(t',u') \right\}_{O}
=
\left\{ {\cal A}(t',u')\,{\cal A}(t,u) \right\}_{O}
=
\begin{cases}
{\cal A}(t,u)\,{\cal A}(t',u') & \text{for } t \ge t', \\
{\cal A}(t',u')\,{\cal A}(t,u) & \text{for } t' \ge t,
\end{cases}
\end{equation}
independently of the values of the second time arguments $u$ and $u'$.

In contrast, the total time ordering projection yields
\begin{equation}
\label{exR=G}
\left\{ {\cal A}(t,u)\,{\cal A}(t',u') \right\}_{G}
=
\left\{ {\cal A}(t',u')\,{\cal A}(t,u) \right\}_{G}
=
\begin{cases}
{\cal A}(t,u)\,{\cal A}(t',u') & \text{for } t \ge u \ge t' \ge u', \\
{\cal A}(t',u')\,{\cal A}(t,u) & \text{for } t' \ge u' \ge t \ge u, \\
0 & \text{otherwise}.
\end{cases}
\end{equation}

From the above example it results apparent that while the $O$ (or PTO) map does not  discard any element (it is always possible to order product of functions respect to one parameter), the $G$ (or TTO) map is a 
true projection, that erase all the product of functions that do not satisfy a total time ordering.  Thus, 
the $O$ map is ineffective when applied to $c$-numbers\footnote{For historical reason we shall use the definition
of $q$-numbers as objects of a {\em non commutative} algebra, as opposed to $c$-numbers that are
objects of a {\em  commutative} algebra. Time dependent operators are generally $q$-numbers.}, 
while the $G$ map 
is effective also with $c$-numbers. 
For our purpose, the key difference between these two definitions of 
generalized exponential functions, is that while the $t$-exponential (reducing to the standard exponential for $c$-numbers), 
gives rise (or originates from) 
a local-time ME, the $G$-exponential leads to (or comes from) a non local ME. 
More precisely, considering the notation $\exp_O[...]:=\overleftarrow{\exp}[...]$, from the definition
\begin{equation}
\label{MO2}
\hat{{P }}_{\xi[t]}(k(\cdot)\pv t)=        \exp_{O}\left[{K}^{(O)}_{\xi[t]}(k(\cdot)\pv t)\right]
\end{equation}
(thus, ${K}^{(O)}_{\xi[t]}(k(\cdot)\pv t)$ is the generator of the ``standard cumulants'' of the stochastic process $\xi[t]$),
then  we have a local time ME:
\begin{equation}
\label{ME_lPTO}
\partial_t      \hat{{P }}_{\xi[t]}(k(\cdot)\pv t)=\left(\partial_t {K}^{(O)}_{\xi[t]}(k(\cdot)\pv t)\right)\hat{{P }}_{\xi[t]}(k(\cdot)\pv t).
\end{equation} 
On the other hand, from \eqref{MG2} we obtain a ME with memory kernel  (for details, see~\cite[Section 4.4.3]{bbJSTAT4} and
\cite{bCSF159}):
\begin{equation}
\label{MEG}
\partial_t      \hat{{P }}_{\xi[t]}(k(\cdot)\pv t)=\int_0^t \text{d}u\,G(k(\cdot)\pv t,u)
\hat{{P }}_{\xi[t]}(k(\cdot)\pv u),
\end{equation} 
where the Green function, or memory kernel $G(k(\cdot)\pv t,u)$ is related to  ${K}^{(G)}_{\xi[t]}(k(\cdot)\pv t)$, 
by
\begin{align}
\label{K_G}
&       {K }^{(G)}_{\xi[t]}(k(\cdot)\pv t)
:=
\int^t_0{\text{d}u}\int^{u}_0{\text{d}u'}G(k(\cdot)\pv u,u').%
\end{align}

By using Eqs.~\eqref{CG} and \eqref{K_G}, the Green function 
$G(k(\cdot)\pv t,u)$ can be, in turn, expanded in series of $G$-
cumulants, see~\cite[Section~4.4.3, Eq.~(90)]{bbJSTAT4}:
\begin{align}
\label{GCum_xi}
& G(k(\cdot)\pv u,u')
=\sum_{n=1}^\infty \int_{u'}^{u} \text{d}u_{n-1} \int_{u'}^{u_{n-1}} \text{d}u_{n-2}...\int_{u'}^{u_{3}} \text{d}u_{2}
\nonumber\\
&\times k(u')k(u_2)\dots k(u_{n-1})k(u)
\doubleangle{\xi(u') \xi(u_{2}) \dots \xi(u_{n-1})\xi(u)}^{(G)} 
\end{align}

As for  standard cumulants, inserting Eqs.~\eqref{MG}-\eqref{CG} in Eq.~\eqref{MG2}, we obtain
%
%
\begin{align}
\label{corrGen_delta_cum}
&\int^t_0{\text{d}u_n}\int^{u_{n}}_0{\text{d}u_{n-1}}\dots 
\int^{u_{2}}_0\text{d}u_1\,
k(u_1)k(u_2)\dots k(u_n) 
\nonumber \\
&       
\langle\xi(u_1)\xi(u_2)...\xi(u_n)\rangle
\nonumber \\
& =
\sum_{p=1}^{ n}
\sum_{\{m_i\}:\sum_{i=1}^p m_i={ n}}   
\int^t_0{\text{d}u_n}\int^{u_{n}}_0{\text{d}u_{n-1}}\dots 
\int^{u_{2}}_0\text{d}u_1\,
k(u_1)k(u_2)\dots k(u_n) 
\nonumber \\
&
\doubleangle{\xi(u_1)\xi(u_2)...\xi(u_{m_1})}^{(G)}
\doubleangle{\xi(u_{m_1+1})\xi(u_{m_1+2})...\xi(u_{m_1+m_2})}^{(G)}\times ...
\nonumber \\
&
...\times  \doubleangle{\xi(u_{m_1+...+m_{p-1}+1})\, \xi_2(u_{m_1+...+m_{p-1}+2}),
...\xi(u_{m_1+...+m_{p-1}+m_p})}^{(G)}.
\end{align}
that, because it must hold for \textit{any} ``dummy'' wave function $k(u)$, leads to
\begin{align}
\label{corrGen_delta_cum_2}      
&\langle\xi(u_1)\xi(u_2)...\xi(u_n)\rangle =
\sum_{p=1}^{ n}
\sum_{\{m_i\}:\sum_{i=1}^p m_i={ n}}   
\doubleangle{\xi(u_1)\xi(u_2)...\xi(u_{m_1})}^{(G)}\nonumber \\
&\times
\doubleangle{\xi(u_{m_1+1})\xi(u_{m_1+2})...\xi(u_{m_1+m_2})}^{(G)}\times ...
\nonumber \\
&
...\times  \doubleangle{\xi(u_{m_1+...+m_{p-1}+1})\, \xi_2(u_{m_1+...+m_{p-1}+2})
...\xi(u_{m_1+...+m_{p-1}+m_p})}^{(G)}.
\end{align}

    The r.h.s. of Eqs.~\eqref{corrGen_delta_cum}-\eqref{corrGen_delta_cum_2} differs from that of the standard cumulants 
by the fact that all the cumulants that in the former are not fully time ordered have been discarded (the $G$ map has been 
applied to all the terms of the sums in the r.h.s. of Eq.~\eqref{corrGen_delta_cum}).

    \section{The $n=8$ case\label{app:nequals8case}}

We show how to work out the eight-time joint correlation function: the compositions of the ordered times are
\begin{align*}
&\vert t_1\,t_2\,t_3\,t_4\,t_5\,t_6\,t_7\,t_8\vert  \pv &(p=1)\\
&\vert t_1\,t_2\,t_3\,t_4\,t_5\,t_6\vert \,t_7\,t_8\vert \text{ , }
\vert t_1\,t_2\,t_3\,t_4\vert \,t_5\,t_6\,t_7\,t_8\vert \text{ , }
\vert t_1\,t_2\vert \,t_3\,t_4\,t_5\,t_6\,t_7\,t_8\vert \pv &(p=2)\\ 
&\vert t_1\,t_2\,t_3\,t_4\vert \,t_5\,t_6\vert \,t_7\,t_8\vert \text{ , }
\vert t_1\,t_2\vert \,t_3\,t_4\,t_5\,t_6\vert \,t_7\,t_8\vert \text{ , }
\vert t_1\,t_2\vert \,t_3\,t_4\vert \,t_5\,t_6\,t_7\,t_8\vert \pv &(p=3)
\\&\vert t_1\,t_2\vert \,t_3\,t_4\vert \,t_5\,t_6\vert \,t_7\,t_8\vert .&(p=4)
\end{align*}
from which, by using  the result \eqref{corrGen_delta_2} of the procedure illustrated
in Proposition~\ref{prop:procedure_fligh}, becomes
\begin{equation}
\label{corr8_}
\begin{array}{ll}
\left\langle\xi(t_1)\xi(t_2)\xi(t_3)\xi(t_4)\xi(t_5)\xi(t_6)\xi(t_7)\xi(t_8)\right\rangle  & 
\\
\;&\;\\
=\frac{\overline{\xi^{8}}}{\overline{\xi^{2}}}\left\langle \xi(t_1)\xi(t_8)\right\rangle 
\;\;\;&(p=1) \\
\;&\;\\
+\big[
\frac{\overline{\xi^{6}}}{\overline{\xi^{2}}} 
\left\langle \xi(t_1)\xi(t_6)\right\rangle  
\left(\left\langle\xi(t_7)\xi(t_8)\right\rangle_{t_6}-\overline{\xi^{2}}\delta(t_{7}-t_{6})\right) \\
+\left(\frac{\overline{\xi^{4}}}{\overline{\xi^{2}}}\right)^2 
\left\langle \xi(t_1)\xi(t_4)\right\rangle  
\left(\left\langle\xi(t_5)\xi(t_8)\right\rangle_{t_4}-\overline{\xi^{2}}\delta(t_5-t_4)\right)\\
+\frac{\overline{\xi^{6}}}{\overline{\xi^{2}}}\left\langle \xi(t_1)\xi(t_2)\right\rangle  
\left(\left\langle\xi(t_3)\xi(t_8)\right\rangle_{t_2}-\overline{\xi^{2}}\delta(t_3-t_2)\right)
\big]
\;\;\;&(p=2)
\\
\;&\;\\
+\frac{\overline{\xi^{4}}}{\overline{\xi^{2}}}
\big[
\left\langle \xi(t_1)\xi(t_4)\right\rangle  
\left(\left\langle\xi(t_5)\xi(t_6)\right\rangle_{t_4} 
-\overline{\xi^{2}}\delta(t_5-t_4)\right)\left(\left\langle\xi(t_7)\xi(t_8)\right\rangle_{t_6} -\overline{\xi^{2}}\delta(t_7-t_6)\right)\\
+\,\left\langle \xi(t_1)\xi(t_2)\right\rangle  
\left(\left\langle\xi(t_3)\xi(t_6)\right\rangle_{t_2} -\overline{\xi^{2}}\delta(t_3-t_2)\right)
\left(\left\langle\xi(t_7)\xi(t_8)\right\rangle_{t_6}-\overline{\xi^{2}}\delta(t_7-t_6)\right)\\
+\left.\left\langle \xi(t_1)\xi(t_2)\right\rangle  
\left(\left\langle\xi(t_3)\xi(t_4)\right\rangle_{t_2}-\overline{\xi^{2}}\delta(t_3-t_2)\right)
\left(\left\langle\xi(t_5)\xi(t_8)\right\rangle_{t_4}
-\overline{\xi^{2}}\delta(t_5-t_4)\right)\right]
\;\;\;&(p=3)
\\
\;&\;\\
+
\big[\left\langle \xi(t_1)\xi(t_2)\right\rangle  
\left(\left\langle\xi(t_3)\xi(t_4)\right\rangle_{t_2}
-\overline{\xi^{2}}\delta(t_4-t_3)\right)
\left(\left\langle\xi(t_5)\xi(t_6)\right\rangle_{t_4}
-\overline{\xi^{2}}\delta(t_6-t_5)\right)\times
\;\;\;&
\\
\;&\;\\
\times 
\left(\left\langle\xi(t_8)\xi(t_7)\right\rangle_{t_6}
-\overline{\xi^{2}}\delta(t_8-t_7)\right)\big]
\;\;\;&(p=4)
\end{array}
\end{equation}

\section{Obtaining the Montroll-Weiss-Scher result for the CTRW\label{app:Montroll_Weiss}}

When the drift  $C(x) $ of \eqref{SDE} is not present and $I(x)=1$, we have the stanrdard L\'evy flights CTRW $\dot x=\xi[t]$, 
for which a ME for the PDF
of $x$ has been
obtained since the pioneering works of Montroll Weiss and Scher 
\cite{Montroll_Weiss_1965,Scher_Motroll_1975}, together with a closed formal expression for the same PDF. In Fourier Laplace transform it reads:

\begin{equation}
\label{MW_ks}
\hat P(k\pv s)=\frac{\hat \Psi(s)}{1-\hat \psi(s)\hat  p(k)} \hat P(k\pv 0)
\end{equation}
where $P(x,0)$ is the
initial PDF of $x$ and $\Psi(t)$ is the survival probability, defined as the
probability that, after a time interval $t$ since the last transition, the
random variable $\xi$ has not changed value:
\begin{equation}
\label{Psi}
\Psi(t)
:= \int_t^\infty \psi(u)\,\mathrm{d}u
= 1-\int_0^t \psi(u)\,\mathrm{d}u
\;\Rightarrow\;
\hat{\Psi}(s) = \frac{1-\hat{\psi}(s)}{s}.
\end{equation}
 
In the space-time variables the Montroll-Weiss equation reads
\begin{equation}
P(x\pv t)=\int_{-\infty}^{+\infty} p(y) \int_0^t \psi(\tau) P(x-y\pv \,t-\tau) \text{d}y \,\text{d} \tau+\Psi(t) P(x\pv 0)
\label{eq:mws}
\end{equation}

Considering that   \textit{for any 
trajectory realization}  $\xi(u)$, with $t_0\le u\le t$, of   
the stochastic process $\xi[t]$, we have 
\begin{equation}
\label{x_CTRW}
{x}\left(t\right)= \! \int^t_0{\xi (u)\text{d}u},
\end{equation} 
then, fixed the time $t$,  the Fourier transform of the PDF, also called the characteristic function (CF) of $x(t)$, is
{\small 
\begin{align}
\label{pExp2x}
&\hat{P}(k\pv t):=
\langle\exp \left[\I k \, x(t)\right]\rangle P(k\pv 0)
=\langle \exp \left[\I k \int^t_0\text{d}u\, \xi(u)\right]\rangle P(k\pv 0)
\nonumber \\
&=\left\{\sum_{n=0}^\infty  \frac{(\I k)^n}{n!}
\int^t_0{\text{d}u_n}\int^{t}_0{\text{d}u_{n-1}}\dots \int^{t}_0\text{d}u_1
\langle \xi (u_1)\xi (u_2)\dots \xi (u_n)\rangle
\right\}P(k\pv 0)
%
%
%
\end{align}
}

In the multiple integral of Eq.~\eqref{pExp2x}, all the times $ u_1, u_2, \dots, u_n $ range 
from 0 to $ t $, and are therefore not ordered. However, since the integrand—namely, the joint 
correlation function 
$ \langle \xi(u_1)\xi(u_2)\dots \xi(u_n) \rangle $—is fully symmetric with respect to the $ n $ 
time variables, the unordered multiple integral can be rewritten as $ n! $ times the ordered 
integral, where 
$ u_1 < u_2 < \dots < u_n $. 
This accounts for all $ n! $ permutations of the times yielding the same result.
However, in the case where a subset of $ m $ times are equal to each other, permutations among those 
identical times should not be counted, thus an extra factor $1/m!$ must be inserted.

Then, by substituting the general result from Eq.~\eqref{corrGen_delta} into Eq.~\eqref{pExp2x}, 
properly accounting for the time ordering and also considering that the Dirac-delta functions
contraint the integrals to be evaluated at ``diagonal'' of the multidimensional hyper-cube
$[0,t]^n$, we obtain:
{\small 
\begin{align}
\label{pExp2x_}
&\hat{P}(k\pv t)=\Bigg\{1
\nonumber \\
&+\sum_{n=1}^\infty (\I k) ^n
\int^t_0{\text{d}u_n}\int^{u_{n}}_0{\text{d}u_{n-1}}\dots \int^{u_{2}}_0\text{d}u_1
\sum_{p=1}^n 
\;\;\sum^\infty_{\substack{m_0,m_1,m_2,...,m_p =0\\ \sum_{i=0}^{p}m_i = n}} \;\;
\prod_{i=0}^{p} \frac{1}{m_i!}
\nonumber \\&
\left( \overline{\xi_0^{m_1}}\,\delta(u_{m_1}-t_0)+\overline{\xi^{m_1}}\, \bm{\delta}_1(\bm{\Delta t}_{m_1})
R(u_{m_1}-t_0)\right)
\nonumber \\&
\times\overline{\xi^{m_2}}\,\bm{\delta}_2(\bm{\Delta t}_{m_2})R(u_{m_1+m_2}-u_{m_1})
\times\overline{\xi^{m_3}}\,\bm{\delta}_3(\bm{\Delta t}_{m_3})R(u_{m_1+m_2+m_3}-u_{m_1+m_2})
\nonumber \\&
\times...\times \overline{\xi^{m_p}}\bm{\delta}_p(\bm{\Delta t}_{m_p})
R(u_{n}-u_{_{n-m_{p}}})\Bigg\}\hat P(k\pv 0)
\nonumber \\
&=\Bigg\{1+\sum_{n=1}^\infty (\I k) ^n
\sum_{p=1}^n \;
\;\;\sum^\infty_{\substack{m_0,m_1,m_2,...,m_p =0\\ \sum_{i=0}^{p}m_i = n}} \;\;
\prod_{i=0}^{p} \frac{\overline{\xi^{m_i}}}{m_i!}
\times\nonumber\\
&
\times  \int^t_0
\text{d}u_{n}
\int^{u_{n}}_0 R(u_{n}-u_{m_1+...+m_{p-1}}){\text{d}u_{m_1+m_2+...+m_{p-1}}}
\nonumber       \\
&\int^{u_{m_1+m_2+...+m_{p-1}}}_0R(u_{m_1+...+m_{p-2}+m_{p-1}}-u_{m_1+...+m_{p-2}}){\text{d}u_{m_1+m_2+...+m_{p-2}}}\dots\nonumber\\
& \int^{u_{m_1+m_2+m_3+m_4+m_5+m_6}}_0
R(u_{m_1+m_2+m_3+m_4+m_5+m_6}-u_{m_1+m_2+m_3+m_4+m_5})
\text{d}u_{m_1+m_2+m_3+m_4+m_5}
\nonumber \\
& \int^{u_{m_1+m_2+m_3+m_4+m_5}}_0
R(u_{m_1+m_2+m_3+m_4+m_5}-u_{m_1+m_2+m_3+m_4})
\text{d}u_{m_1+m_2+m_3+m_4}
\nonumber \\
& \int^{u_{m_1+m_2+m_3+m_4}}_0
R(u_{m_1+m_2+m_3+m_4}-u_{m_1+m_2+m_3})
\text{d}u_{m_1+m_2+m_3}
\nonumber \\
&
\int^{u_{m_1+m_2+m_3}}_0
R(u_{m_1+m_2+m_3}-u_{m_1+m_2})
\text{d}u_{m_1+m_2}
\nonumber \\
&
\int^{u_{m_1+m_2}}_0
R(u_{m_1+m_2}-u_{m_1})
\tilde R(u_{m_1})
\text{d}u_{m_1}\Bigg\}\hat P(k\pv 0)
\end{align}
}
The product of integrals in noting but the integral (the last one, i.e., $\int_0^t \text{d}u_n..$) of the $p$-fold convolution of the rate function $R(u)$, thus,  Laplace transforming  Eq.~\eqref{pExp2x_} yields:
{\small 
\begin{align}
\label{pExp2x_Fin}
\hat{P}(k\pv s)
&=\Bigg\{\frac{1}{s}+\frac{1}{s}\sum_{n=1}^\infty (\I k) ^n
\sum_{p=1}^n 
\left[\hat{ {R}}(s) \right]^{p}
\sum^\infty_{\substack{m_0,m_1,m_2,...,m_p =0\\ \sum_{i=0}^{p}m_i = n}} \;\;
\prod_{i=0}^{p} 
\frac{\overline{\xi^{m_i}}}{m_i!} 
\Bigg\}P(k\pv 0)
\nonumber\\
& =\Bigg\{\frac{1}{s}+\frac{1}{s}\sum_{q=1}^\infty 
\left[\hat R(s)\right]^q        \left(\sum_{m=1}^\infty 
(\I k)^m        \frac{\overline{\xi^{m}}}{m!}\right)^q
\Bigg\}P(k\pv 0) 
\nonumber \\
&
=\Bigg\{
\frac{1}{s}+\frac{1}{s} \frac{\hat R(s) \left[\hat p(k) -1\right]}{1-\hat R(s) \left[\hat p(k) -1\right]}
\Bigg\}P(k\pv 0)
\nonumber\\
&       =\frac{1}{s}\;\frac{1}{1-\hat R(s) \left[\hat p(k) -1\right]}
P(k\pv 0).
\end{align}
}
Finally, by making use of Eqs.~\eqref{Rtilde} and \eqref{Psi}, we arrive precisely at the  
Montroll--Weiss result given in Eq.~\eqref{MW_ks}.

\section{A ME for the generalized CTRW with drift and state dependent dichotomous steps\label{sec.stationary}}
Let us consider the system of interest $x$ of the SDE~\eqref{SDE}, in the generic case where the drift $-C(x) $ is not zero and
the multiplicative function $I(x)$ effectively depends on $x$ (i.e., it is not a constant).

For any realization 
$\xi(u)$ $|u\in[0,t]$ of the noise, the time-evolution
of the PDF of the system \eqref{SDE}, that we indicate with 
$P(x,\xi(t)\pv t)$, 
satisfies the continuity equation:%
\begin{align}
\label{app:stochLiouv}
\partial_t P(x,\xi(t)\pv t)={\mathcal{L}}_a P(x,\xi(t)\pv t) +\,\partial_x I(x) \xi(t)P(x,\xi(t)\pv t),
\end{align} 
where 
\begin{equation}
\label{app:La}
{\mathcal{L}}_a:=\partial_x C(x) 
\end{equation}
is the unpertubed Liouvillian. Then, passing to the interaction representation, we have:
\begin{equation}
\label{continuity}
{\partial}_t \tilde P(x,\xi(t)\pv t)=
\tilde {\mathcal{L}}_I (t)\xi(t) P(x,\xi(t)\pv t)
=\Omega(t)\tilde P(x,\xi(t)\pv t)
\end{equation} 
where 
\begin{equation}
\label{app:Ptilde}
\tilde P(x,\xi(t)\pv t) :=e^{-{\mathcal{L}}_a t}
P(x,\xi(t)\pv t)
\end{equation}
and $\Omega:=\tilde {\mathcal{L}}_I (t) \xi(t) $, 
with\footnote{Taking into account Eq.~\eqref{app:La}, and that, for any couple of operators $\mathcal{A}$ and $\mathcal{B}$ 
the Hadamard lemma says:
$e^{\mathcal{A}} \mathcal{B} e^{-\mathcal{A}}=e^{\mathcal{A}^\times} [\mathcal{B}]$ where 
$\mathcal{A}^\times [\mathcal{B}]:=\mathcal{A}\mathcal{B}-\mathcal{B}\mathcal{A}$ is the standard commutation operation, 
we have the following  chain
of equalities (see also Ref.~\cite{bJMP59}) : $e^{{-\mathcal{L}}_at} {\partial}_{x} I(x) e^{{\mathcal{L}}_at}
=e^{{-\mathcal{L}}_a t^\times} \left]{\partial}_{x} I(x)\right]=
\partial_x C(x)  e^{{-\mathcal{L}}_a t^\times}\left[\frac{I(x)}{C(x) }\right]:=   \partial_x C(x)  \frac{I(x_0(x\pv -t))}{C(x_0(x\pv -t))}$.}
\begin{equation}
\label{app:LITilde}
\tilde {\mathcal{L}}_I (t):= e^{{-\mathcal{L}}_at} {\partial}_{x} I(x) e^{{\mathcal{L}}_at}=\partial_x C(x)  \frac{I(x_0(x\pv t))}{C(x_0(x\pv t))}
\end{equation}
where $x_0(x\pv t)$ is the  unperturbed evolution (i.e. with $\xi(t)=0$) of the system of interest, for a time $t$, 
starting from the initial position at $x(0)=x$. 

It is apparent that  $\Omega[t]$ is a stochastic differential operator
with a one-to-one correspondence between the realizations of 
$\Omega(u)$ and those of $\xi(u)$, with 
$0\le u\le t$. 
Note that the stochastic operator $\Omega(u)$ is a $q$-number, i.e., it does not commute with itself when computed at different times: 
$\left[\Omega(u),\Omega(u^\prime)\right]\ne 0$, 
$u,u^\prime\in\left[0,t\right]$ (i.e.,  
$\Omega(u), u\in\left[0,t\right]$ 
belongs to anon  commutative algebra). Assuming that
the initial preparation of the ``ensemble'' $P(x\pv 0)$ does not
depend on the possible values of $\xi$ (i.e., the initial PDF is factorized
as a reduced PDF for $x$, $P(x\pv 0)$ by the PDF of $\xi$), the temporal
integration of \eqref{continuity}, provided with the average over all the possible realizations
of $\xi(u)$, $u\in\left[0,t\right]$, leads to the reduced PDF
for the variable $x$ at any time $t\geq 0$:
\begin{equation}
\label{pExp}
\tilde P(x\pv t):=\left \langle\tilde  P(x,\xi(t)\pv t)\right \rangle   P(x\pv 0)
=\left\langle 
\overleftarrow{ \exp} \left[ \int^t_0 \text{d}u\,\Omega(u)\right]\right \rangle P(x\pv 0)
\end{equation}
where the symbol $\overleftarrow{ \exp}[...]$ means that to the argument of the exponential function the Partial Time Ordering (PTO) map must be applied,
i.e., $\overleftarrow{ \exp}[...]$ is the standard time ordered, or $t-$ordered exponential function. According to the theory of $M$-cumulant
we will use the notation $\exp_O[...]:=\overleftarrow{ \exp}[...]$, where the symbol ``O''  stays for the PTO map.
By comparing \eqref{pExp} with \eqref{CF_CTRW}, and setting $\I k=1$, the time-evolution operator 
$\langle \exp_O{\left[\int_{0}^{t}\text{d}u\Omega(u)\right]}\rangle$
of \eqref{pExp} can be considered as the $O$-Characteristic Function ($O$-CF) 
of the random operator given by $\mathscr{S}(t):=\int_0^t\text{d}u\Omega(u)$:
\begin{align}
\label{pExp2}
&       {\hat{\mathscr{P}}}_{\mathscr{S}}(k\pv t):= \left\langle \exp_O\!\left[\I k \,\mathscr{S}(t)\right]\right \rangle=
\left\langle  \exp_O\!\left[\I k \int^t_0 \text{d}u\,\Omega(u)\right]\right \rangle      
\nonumber\\
&=\sum_{n=0}^\infty (\I \,k)^n
\int^t_0{\text{d}u_n}\int^{u_{n-1}}_0{\text{d}u_{n-2}}\dots 
\int^{u_{2}}_0\text{d}u_1\,
\partial_x I(x\pv u_1)\partial_x I(x\pv u_2)\dots \times 
\nonumber \\
&       \dots \times \partial_x I(x\pv u_n)
\langle \xi (u_1)\xi (u_2)\dots \xi (u_n)\rangle. 
\end{align}
The reader should appreciate that we have used the following fact:\\
$\left\langle \Omega(u_1)\Omega(u_2)...\Omega(u_n)\right\rangle=
\partial_x I(x\pv u_1)\partial_x I(x\pv u_2)...\partial_x I(x\pv u_n)\langle \xi(u_1)\xi(u_2)...\xi(u_2)\rangle$.

To avoid confusion, we note that while  the CF $\hat{P}_{x}(k\pv t)$
of  \eqref{CF_CTRW} is the Fourier transform of the PDF of $x$, now
$\hat{\mathscr{P}}_{\mathscr{S}}(k\pv t)$
of \eqref{pExp2} is the \textit{time-evolution operator } of the PDF of $x$ (in interaction representation),
as it is clear from  \eqref{pExp}-\eqref{pExp2}.  

Akin to what done in Section~\ref{sec:Gcumulants}, we indicate 
with $\doubleangle{ \Omega (u_1)\Omega (u_2)\dots \Omega (u_n)}^{(G)}$ the $n$-time $G$-cumulant  of $\Omega[t]$. 
The corresponding 
generator $\mathscr{K}^{(G)}_{\Omega[t]}(k(\cdot)\pv t)$, defined by  
\begin{equation}
\label{app:MG2_Omega}
\hat{\mathscr{P }}_{\Omega[t]}(k(\cdot)\pv t)=     \exp_{G}\left[\mathscr{K}^{(G)}_{\Omega[t]}(k(\cdot)\pv t)\right]
\end{equation}
is then 

\begin{align}
\label{app:CG_Omega}
&       \hat{\mathscr{K }}^{(G)}_{\Omega[t]}(k(\cdot)\pv t)
:=\sum_{n=1}^\infty (\I\,k)^n
\int^t_0{\text{d}u_n}\int^{u_{n-1}}_0{\text{d}u_{n-2}}\dots 
\int^{u_{2}}_0\text{d}u_1\,
\nonumber \\
&\times 
\doubleangle{ \Omega (u_1)\Omega (u_2)\dots \Omega (u_n)}^{(G)}
\nonumber \\&
=\sum_{n=1}^\infty \I^n
\int^t_0{\text{d}u_n}\int^{u_{n-1}}_0{\text{d}u_{n-2}}\dots 
\int^{u_{2}}_0\text{d}u_1\,
\partial_x I(x\pv u_1)\partial_x I(x\pv u_2)\dots \times 
\nonumber \\
&       \dots \times \partial_x I(x\pv u_n)
\doubleangle{ \xi (u_1)\xi (u_2)\dots \xi (u_n)}^{(G)}%
\end{align}
where the last equation holds because of the specific property of the $G$-cumulants that, thanks to the TTO map, 
are always proportional to the corresponding cumulants of $\xi[t]$~\cite{bbJSTAT4}.

The equivalent to Eq.~\eqref{MEGx} is the following ME
\begin{equation}
\label{MEGx_Omega}
\partial_t      \hat{\mathscr{P}}_{x}(k\pv t)=\int_0^t \text{d}u\,\mathscr{G}(k\pv t,u)
\hat{\mathscr{P}}_{x}(k\pv u)
\end{equation} 
with
\begin{align}
\label{K_Gx_Omega}
&       \hat{\mathscr{K }}_{x}(k\pv t)
:=
\int^t_0{\text{d}u}\int^{u}_0{\text{d}u'}{\mathscr{G(k\pv u,u')}}.%
\end{align}

\section{The spike dichotomous case\label{app:dicho}}

To demonstrate the generalized factorization property expressed
by Eq.~\eqref{corrGen_delta_dicho},
we rewrite  Eq.~\eqref{corrGen_delta_2}
in a more compact and intriguing
way, also implementing the fact that the odd moments of $\xi$ are zero.
For that, we first  introduce the definitions of the following  closed round ``bra'' and
``ket'' braces that simplifies the formulas
($0<i< j$, $k>0$): %
%
Thus we have
\begin{equation}
\label{triangle_prop}
\limg t_i,t_j\rimg\delta(t_{j+k}-t_{j})=\limg t_i,t_{j+k}\rimg,
\end{equation}
and the following projection property holds:
\begin{equation}
\label{triangle_prop2}
\limg t_i,t_j,t_{j+k}\rimg=     \limg t_i,t_{j+k}\rimg.
\end{equation}
For the sake of convenience, we give also a bilinear  property to these closed round braces:
\begin{equation}
\label{triangle_proj}
\xi^n\limg t_i,t_{j+k}\rimg=\limg \xi^n\, t_i,t_{j+k}\rimg=\limg t_i,t_{j+k}\,\xi^n\rimg.
\end{equation}
It will result convenient to use also a brackets-like notation for the average of a function over the PDF of jumps:
\begin{equation}
\label{mean}
\lceil{f(\xi)}\rceil:=\overline{f(\xi)}.
\end{equation}

By exploiting the alternative definition of average over the random variable  $\xi$, given
in Eq.~\eqref{mean}, the two-time correlation function given of Eq.~\eqref{cor2t_delta_pre} can be rewritten as
\begin{align}
\label{cor2t_delta_new}
\langle \xi(t_1) \xi(t_2) \rangle  &:=\lceil \xi^2 \rceil\limg  t_1,t_2\rimg=\lceil\limg \xi^2\,  t_1,t_2\rimg \rceil .
\end{align}
%
%
%
%

By using these definitions, we are able to rewrite the central result given in Eq.~\eqref{corrGen_delta_2} in the following
way:
\begin{align}
\label{corrGen_delta_2_}
&\langle\xi(t_1)\xi(t_2)\cdots \xi(t_n)\rangle 
\sum_{p=1}^{ n/2} \bigg[
\sum_{\{m_i\}:\sum_{i=1}^p m_i={ n}} 
\nonumber \\&
\lceil 
\limg {\xi}^{m_1} t_{1},t_{m_1}
\big(\rimg\rceil 
\lceil\limg- \rceil\lceil
\big) 
{\xi}^{m_2} t_{m_1+1},t_{m_1+m_2} 
\big(\rimg\rceil 
\lceil\limg
- \rceil\lceil
\big)\times
\nonumber \\
&\times {\xi}^{m_3}t_{m_1+m_2+1},t_{m_1+m_2+m_3}
\big(
\rimg\rceil 
\lceil\limg
- \rceil\lceil
\big)
\times\cdots  
\nonumber \\
&\cdots \times {\xi}^{m_{p-1}}t_{n-m_p-m_{p-1}},t_{n-m_p}
\big(
\rimg\rceil 
\lceil\limg
- \,
\rceil\lceil
\big)
{\xi}^{m_p}t_{n-m_p+1},t_{n}\!\rimg
\rceil
\bigg].
%
\end{align}
%
If we note that  all the possible compositions of $ n $ ordered times 
(or objects, more generally) into blocks of even elements, here separated each other by the vertical bar $\vert$, 
can be expressed as
\begin{align}
\label{corrGen_theta_cu}
\vert  t_1\,t_{2} \left(\vert  + 1 \right) t_{3}\,t_4 \left(\vert  + 1 \right) t_{5}\,t_{6}
\left(\vert  + 1 \right) t_{7}\,t_{8}
\dots \left(\vert  + 1 \right) t_{n-1}\,t_n \vert ,
\end{align}
then, from point~\ref{iteme} of 
Proposition~\ref{prop:procedure_fligh}, it is clear that a more compact 
alternative to Eq.~\eqref{corrGen_delta_2_} for expressing the 
\(n\)-time correlation functions of \(\xi[t]\) is the following:
\begin{align}
\label{corrGen_delta_fin}
&\langle\xi(t_1)\xi(t_2)\cdots \xi(t_n)\rangle 
= \lceil 
\limg {\xi}^{2} t_{1},t_{2}
\big(\rimg{}\rceil 
\lceil\limg- \rceil\lceil+1
\big) 
{\xi}^{2} t_3,t_{4} 
\big(\rimg\rceil 
\lceil\limg
- \rceil\lceil+1
\big)\times
\nonumber \\
&\times {\xi}^{2}t_{5},t_{6}
\big(
\rimg\rceil 
\lceil\limg
- \rceil\lceil+1
\big)
\times\cdots  
\times {\xi}^{2}t_{n-3},t_{n-2}
\big(
\rimg\rceil 
\lceil\limg
- 
\rceil\lceil+1
\big)
{\xi}^{2}t_{n-1},t_{n}\rimg
\rceil.
\end{align}
For simplicity of notation, but without any loss of generality, we assume that $\xi$ takes 
the values $\pm 1$. In this case all even powers of $\xi$ reduce to 1, i.e., 
$\xi^n = 1$ for even $n$. Consequently, the symbolic ``braces'' notation 
$\lceil \cdots \rceil$ introduced earlier can be replaced with the identity. 
Applying this simplification to Eq.~\eqref{corrGen_delta_fin} amounts to setting 
$\rceil\lceil = 1$, yielding:

\begin{align}
\label{corrGen_delta_fin_3}
\langle\xi(t_1)\xi(t_2)...\xi(t_n)\rangle 
=&  
\limg t_{1},t_{2}
\rimg\limg 
t_3,t_{4} 
\rimg\limg
t_{5},t_{6}
\rimg\times\nonumber \\
&... 
\times \limg t_{n-3},t_{n-2}
\rimg \limg t_{n-1},t_{n}\rimg\nonumber \\
= &  \langle\xi( t_{1})\xi( t_{2}
)\rangle 
\langle\xi(t_3)\xi(t_{4})\rangle_{t_2}
\langle\xi(t_{5})\xi(t_{6})\rangle_{t_4}
\times\nonumber \\
&... 
\times \langle\xi( t_{n-3})\xi(t_{n-2})\rangle_{t_{n-4}} 
\langle\xi( t_{n-1})\xi(t_{n})\rangle_{t_{n-2}},
\end{align}
where, for the last equality, we have exploited Eq.~\eqref{cor2t_delta_new}.  
Equation~\eqref{corrGen_delta_fin_3} matches Eq.~\eqref{corrGen_delta_dicho}, 
thereby concluding the demonstration.

\section{Some analytical insight from the universal PDE in~\eqref{ME_universal_intro}\label{app:ME_resuls}}

In the Poissonian case, the ME for the PDF of the process defined by
Eq.~\eqref{SDE}, reported in Eq.~\eqref{ME_fin_Pois}, has already been derived in
the literature (see, e.g., \citep{sprmPRE83,bpPRE98}). However, since
Proposition~\ref{muG2} shows that the same equation also governs the long-time
behavior of the PDF when the WT distribution decays as a power law
with $\mu>2$, it is instructive to analyze this ME in detail.

At first sight, the master equation~\eqref{ME_fin_Pois} may appear cumbersome,
since the operator $\hat p(-\I\partial_x I(x)) - 1
= \sum_{n=1}^\infty \overline{\xi^{n}}\,(\partial_x I(x))^n/n!$ involves an
infinite series of spatial derivatives. Nevertheless, this structure encodes
important physical information.

For instance, the differential order of the equation clearly does \emph{not}
depend directly on the timescale $\tau$, but is instead determined by the
highest nonvanishing moment of the random variable $\xi$. The accuracy of any
truncation of the series depends both on how the moments $\overline{\xi^n}$
grow with increasing $n$ and on how the drift $C(x)$ constrains the spatial
region accessible to the dynamics.

As an illustrative example, consider the white-noise limit of the
SDE~\eqref{SDE}, obtained by replacing $\xi[t]$ with Gaussian white noise. In
this case, the PDF of $x$ obeys the well-known Stratonovich-type Fokker--Planck
equation
\begin{align}
\label{WNcase}
\partial_t P(x\pv t)
&= \partial_x C(x)\, P(x\pv t)
+ \frac{\overline{\xi^{2}}}{2\tau}
\left[\partial_x I(x)\right]^2 P(x\pv t).
\end{align}
If the spatial scales over which $C(x)$ and $I(x)$ vary are not much larger than
the typical magnitude of the moments of $\xi$, the series in
Eq.~\eqref{ME_fin_Pois} does not, in general, converge to the
Fokker--Planck equation~\eqref{WNcase} (the central
limit situation), even at long times. Nevertheless, in
many relevant situations the resulting statistics can still be  controlled.

Indeed, for several cases of practical interest, Eq.~\eqref{ME_fin_Pois} leads
to expressions that are both physically transparent and numerically
tractable (see Section~\ref{sec:numericaltechinques} for the numerical
implementation) and can also be
 simplified in appropriate limiting regimes.

We now discuss some illustrative examples.

When $\xi[t]$ is a symmetric dichotomous process, taking values $\xi = \pm a$, one
obtains
\begin{align}
\label{MEGx_dicho_}
\partial_t P(x\pv t)
=&\;\partial_x C(x) \,P(x\pv t)
+ \frac{1}{\tau}
\sum_{n=1}^\infty \frac{(a^{2})^{n}}{(2n)!}
\left(\partial_x I(x)\right)^{2n} P(x\pv t).
\end{align}
If the series is absolutely convergent, it can be resummed as
\begin{align}
\label{MEGx_dicho}
&\partial_t P(x\pv t)
=\;\partial_x C(x) \,P(x\pv t)
+ \frac{1}{\tau}
\left\{\cosh\!\left[a\,\partial_x I(x)\right] - 1\right\}
P(x\pv t)\nonumber\\
&=\;\partial_x C(x) \,P(x\pv t)
+ \frac{1}{2\tau}
\left[e^{a\partial_x I(x)} + e^{-a\partial_x I(x)} - 2\right]
P(x\pv t)
\end{align}

Note that the two exponential terms
$e^{\pm a\partial_x I(x)}$ in Eq.~\eqref{MEGx_dicho} can be directly interpreted
in terms of advective transport over a ``time'' $a$ in opposite directions,
driven by the drift field $I(x)$ (see
Section~\ref{sec:numericaltechinques}). Note also that for $I(x)=1$ this
transport reduces to a simple spatial shift by $\pm a$.

When the values of $\xi$ are drawn from a Gaussian PDF, namely
$p(\xi)=\exp\!\left(-\xi^{2}/2\sigma^{2}\right)/(\sqrt{2\pi}\sigma)$, one obtains
\begin{align}
\label{MEGx_gauss_}
\partial_t P(x\pv t)
&=\partial_x C(x) \,P(x\pv t)
+ \frac{1}{\tau}
\sum_{n=1}^\infty
\frac{(2n-1)!!\,(\sigma^{2})^{n}}{(2n)!}
\left(\partial_x I(x)\right)^{2n} P(x\pv t).
\end{align}
If the series is absolutely convergent, this expression can be resumed as
\begin{align}
\label{MEGx_gauss}
\partial_t P(x\pv t)
&=\partial_x C(x) \,P(x\pv t)
+ \frac{1}{\tau}
\left\{
e^{\frac{1}{2}\left[\sigma\,\partial_x I(x)\right]^2} - 1
\right\}
P(x\pv t).
\end{align}

When $p(\xi)$ is a flat distribution over the interval
$[-\sqrt{3} a,\,\sqrt{3}a ]$ (for $a=1$ the variance equals unity), one finds
\begin{align}
\label{MEGx_flat_}
\partial_t P(x\pv t)
&=\partial_x C(x) \,P(x\pv t)
+ \frac{1}{\tau}
\sum_{n=1}^\infty
\frac{3^{n} a^{2n}}{(2n+1)!}
\left(\partial_x I(x)\right)^{2n} P(x\pv t),
\end{align}
which leads to
\begin{align}
\label{MEGx_flat}
\partial_t P(x\pv t)
&=\partial_x C(x) \,P(x\pv t)
+ \frac{1}{\tau}
\left\{
\frac{\sinh\!\left(\sqrt{3}a\,\partial_x I(x)\right)}
{\sqrt{3}a\,\partial_x I(x)} - 1
\right\}
P(x\pv t).
\end{align}
The
expressions in Eqs.~\eqref{MEGx_gauss} and \eqref{MEGx_flat}, corresponding to Gaussian
and flat noise, respectively, can also be recast in terms of advective operators,
very similar to the dichotomous case, as
shown in Section~\ref{sec:numericaltechinques}. 
We now study the analytical solution of Eq. (\ref{MEGx_dicho}) in the case of $I(x)=1+ \beta x$. As first step let us write Eq. (\ref{MEGx_dicho}) as the continuity equation, i.e.,

\begin{equation}
\label{con}
\partial_t      P(x,t)
=-\partial_x J(x,t)=\partial_x\left[C(x) P(x, t)+ \frac{a^{2}}{\tau} I(x)\hat O P(x, t)\right]
\end{equation}
where

\begin{equation}
\label{hato}
\hat{O}=\sum_{k=0}^{\infty}\frac{a^{2k}(\partial_x I(x))^{2k+1}}{(2k+2)!}
\end{equation}
We now focus on the equilibrium solution, $\partial_t  P(x,t)=0$.  Imposing a vanishing current we have

\begin{equation}
\label{con2}
C(x) P_{eq}(x)+ \frac{a^{2}}{\tau} I(x)\hat{O}P_{eq}(x)=0
\end{equation}
The eigenfunctions of the operator $\partial_x I(x)$ are

\begin{equation}
\label{hato2}
\partial_x[ I(x) F]=\lambda F
\end{equation}
Solving in the case of $I(x)=1+\beta  x$, we have

\begin{equation}
\label{hato3}
F=(1+\beta  x)^{\frac{\lambda}{\beta } -1}
\end{equation}
Assuming that  $C(x)=\gamma x$ and

\begin{equation}
P_{eq}(x) =\sum_{n=0}^{\infty}c_n (1+\beta  x)^{\frac{\lambda_n}{\beta } -1}
\end{equation}
we have

\begin{eqnarray}\nonumber
&&  \frac{\gamma}{\beta }(1+\beta  x)\sum_{n=0}^{\infty}c_n(1+\beta  x)^{\frac{\lambda_n}{\beta } -1}-\frac{\gamma}{\beta }\sum_{n=0}^{\infty}c_n(1+\beta  x)^{\frac{\lambda_n}{\beta } -1}+
\\\label{hato4}
&&  \frac{a^{2}}{\tau} (1+\beta  x) \sum_{n=0}^{\infty}\sum_{k=0}^{\infty}\frac{a^{2k}\lambda_n^{2k+1}}{(2k+2)!}\sum_{n=0}^{\infty}c_n(1+\beta  x)^{\frac{\lambda_n}{\beta } -1}=0
\end{eqnarray}
Performing the sum on $k$

\begin{eqnarray}\nonumber
&&  \frac{\gamma}{\beta }(1+\beta  x)\sum_{n=0}^{\infty}c_n(1+\beta  x)^{\frac{\lambda_n}{\beta } -1}-\frac{\gamma}{\beta }\sum_{n=0}^{\infty}c_n(1+\beta  x)^{\frac{\lambda_n}{\beta } -1}+
\\\label{hato5}
&&  \frac{a^{2}}{\tau} (1+\beta  x)\sum_{n=0}^{\infty}\frac{ \cosh\! \left(a \lambda_n\right)-1}{\lambda_na^{2}}c_n(1+\beta  x)^{\frac{\lambda_n}{\beta } -1}=0
\end{eqnarray}
After a little algebra

\begin{eqnarray}\nonumber
&&  \frac{\gamma}{\beta } \sum_{n=0}^{\infty}c_n(1+\beta  x)^{\frac{\lambda_n}{\beta }  }-\frac{\gamma}{\beta }\sum_{n=0}^{\infty}c_n(1+\beta  x)^{\frac{\lambda_n}{\beta } -1}+
\\\label{hato6}
&&    \sum_{n=0}^{\infty}\frac{ \cosh\! \left(a \lambda_n\right)-1}{\tau\lambda_n}c_n(1+\beta  x)^{\frac{\lambda_n}{\beta } }=0
\end{eqnarray}

To be able to have a recursive relationship, and an integrable $P_{eq}(x)$, we must set 
    $\lambda_n/\beta=-(\alpha+n)$, where $\alpha$ is to be determined. It this way we obtain
\begin{eqnarray}\nonumber
&&   \sum_{n=0}^{\infty}\frac{c_n}{(1+\beta  x)^{n+\alpha }}-\sum_{n=0}^{\infty}\frac{c_n}{(1+\beta  x)^{n+1+\alpha }}-
\\\label{hato7}
&&    \sum_{n=0}^{\infty}\frac{ \cosh\! \left[a\beta (n+\alpha)\right]-1}{\gamma\tau(n+\alpha)}\frac{c_n}{(1+\beta  x)^{n+\alpha }}=0,
\end{eqnarray}
that yields the following recursion relationship
\begin{eqnarray} 
c_n\left(1-\frac{ \cosh\! \left[a\beta (n+\alpha)\right]-1}{\gamma\tau(n+\alpha)}\right)= c_{n-1}.
\label{hato8}
\end{eqnarray}
   Note that the above equation depends on the two products $\gamma \tau$ and $a \beta$.
   Observing that in terms of $\alpha$ the PDF \eqref{hato4} becomes
    
\begin{equation}
\label{hato4_2}
P_{eq}(x) =\sum_{n=0}^{\infty}\frac{c_n}{ (1+\beta  x)^{n+1+\alpha}}
\end{equation}
the large $x$ behavior of $P(x)$ is 
\begin{equation} 
P_{eq}(y) \approx \frac{A}{(1+\beta x)^{\alpha +1}},\,\, x\to\infty,
\label{power_asymp}
\end{equation}
that is obtained by setting 
 $n=0$ in Eq.~\eqref{hato4_2}. Setting  $c_{-1}=0$
 in \eqref{hato8},  we obtain the following equation for  $\alpha$:

\begin{equation} 
\gamma\tau -\frac{ \cosh\! \left[a\beta  \alpha\right]-1}{ \alpha} =0.
\label{power2}
\end{equation}
This is a transcendental equation that can be solved numerically. In Fig.~\ref{fig:alpha}
we plot $\alpha$ vs $\gamma \tau$ for some fixed values of $\beta$. 
\begin{figure}
    \centering
    \includegraphics[width=0.5\linewidth]{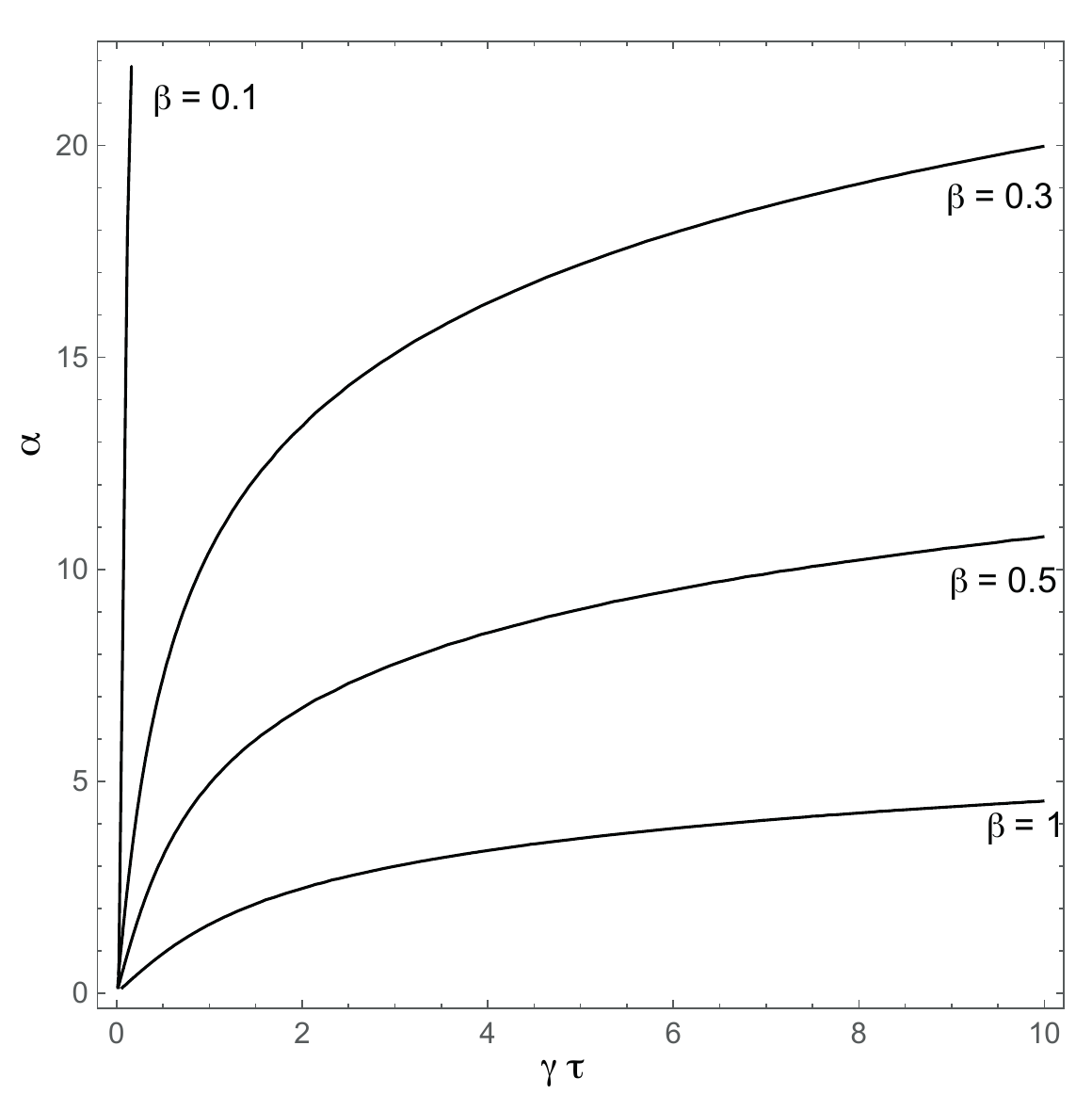}
    \caption{The $\alpha$ exponent vs $\gamma \tau$ for different $\beta$ values (see text for details).}
    \label{fig:alpha}
\end{figure}

If we make the assumption $a\beta \alpha \ll 1$, then we obtain the known result

\begin{equation} 
\alpha\approx \frac{ 2\gamma\tau}{a^2\beta^2} ,\,\,\,\, \frac{ 2\gamma\tau}{a\beta }\ll 1.
\label{power3}
\end{equation}
Next we study Eq. (\ref{MEGx_dicho}) near to the point $I(x)=0$ or $x=-1/\beta $. To do that we perform the variable change $y=1+\beta x$ with $y\to0$. Rewriting Eq. (\ref{con2})
as

\begin{equation}
\label{con2bis}
\gamma\tau(y-1) P_{eq}(y)+y \sum_{k=0}^{\infty}\frac{a^{2k+2}(\partial_y y)^{2k+1}}{(2k+2)!}P_{eq}(y)=0
\end{equation}
and keeping the first order of $y$ in the series we obtain

\begin{eqnarray}\nonumber
&&\left[\gamma\tau\frac{(y-1)}{y}+\cosh[a\beta ]-1\right] P_{eq}(y)+ 
\\\label{con2ter}
&& \left[\frac{1}{2}\left(\cosh[2a\beta ]-1\right) -\cosh(a\beta )+1\right]  yP'_{eq}(y)=0.
\end{eqnarray}
The solution is

\begin{eqnarray}  
P_{eq}(y)\approx A\frac{ \exp\left[-\frac{\gamma\tau}{D y}\right]}{y^{\frac{\gamma\tau+\cosh[a\beta ]-1}{D}}}=
A\frac{ \exp\left[-\frac{\gamma\tau}{D (1+\beta x)}\right]}{(1+\beta x)^{\frac{\gamma\tau+\cosh[a\beta ]-1}{D}}}
\end{eqnarray}
where

\begin{eqnarray} \nonumber
D= \frac{ \cosh[2a\beta ]+1}{2}-\cosh[a\beta ]
\\\label{power3_}
\end{eqnarray}
For $a\beta \ll 1$ we recover the known result

\begin{eqnarray}  
P_{eq}(x)\approx A\frac{ \exp\left[-\frac{2\gamma\tau}{a^2\beta^2 (1+\beta x)}\right]}{(1+\beta x)^{\frac{2\gamma\tau}{a^2\beta^2 }+1}}
\end{eqnarray}
Finally we study the case $I=1$ that, as well as being a new example, it reproduces the behavior near $x=0$ of the previous case. Writing the equation of the interest

\begin{equation}
\label{MEGx_dichobis}
\partial_x [\gamma\tau x P(x)]+   
\left( \cosh\left[a \partial_x \right]-1\right)
P(x).
\end{equation}
Passing to fourier transform
\begin{equation}
\label{MEGx_dichoter}
-k\partial_k [\gamma\tau P(k)]+   
\left( \cos\left[a k\right]-1  \right)
P(k).
\end{equation}
or

\begin{equation}
\label{MEGx_dichoter.1}
P(k) = 
\exp\left[ \int_{0}^{k} \frac{\cos\left[a u\right]-1}{u}du\right] .
\end{equation}
For $x\to\infty$ 
\begin{equation}P(x)\approx \frac{\exp\left[-\frac{x^2}{a^2}\right]}{\sqrt{\pi } a} 
\end{equation} 
while for $x\to 0$ we have

\begin{equation}
\label{MEGx_dichoquat}
P(x)\approx c_1+c_2 x^2+c_3 | x| ^{\frac{1}{\gamma  \tau }-1}
\end{equation}
We infer that for $\gamma  \tau >1$ the $P(x)$ has an integrable divergence in $x=0$, for $1/3<\gamma  \tau <1$ we have an angular point and for $\gamma  \tau <1/3$ we have a smooth maximum. 

The approach used for the analysis of the solution of Eq.~(\ref{MEGx_dicho}) can be easily extended to Eqs.~(\ref{MEGx_gauss}), (\ref{MEGx_flat}).

%

\bibliography{BiblioCentrale}
\end{document}